\documentclass{article}
\usepackage{graphicx}

\usepackage{
amsmath,amssymb,amsfonts,amsthm} 
\usepackage{bm}
\usepackage{colonequals}
\usepackage{comment}
\usepackage{epsfig} \usepackage{latexsym,nicefrac,bbm}
\usepackage
{xspace}
\usepackage{color,fancybox,graphicx,subfigure,fullpage}
\usepackage[top=1in, bottom=1in, left=1in, right=1in]{geometry}
\usepackage{tabularx} 
\usepackage{hyperref} 
\usepackage{cleveref}
\usepackage{pdfsync}
\usepackage[boxruled,linesnumbered]{algorithm2e}
\usepackage{siunitx}

\usepackage{multicol}
\usepackage{enumitem}
\usepackage{float}
\usepackage{tikz}
\usepackage{tabularray}
\usepackage{framed}

\makeatletter
\newcommand{\vast}{\bBigg@{4}}
\newcommand{\Vast}{\bBigg@{5}}
\makeatother

\newtheorem{theorem}{Theorem}[section]
\newtheorem{fact}{Fact}[section]
\newtheorem{conjecture}[theorem]{Conjecture}
\newtheorem{definition}[theorem]{Definition}
\newtheorem{lemma}[theorem]{Lemma}
\newtheorem{remark}[theorem]{Remark}
\newtheorem{proposition}[theorem]{Proposition}

\newtheorem{question}[theorem]{Question}%

\newcommand{\Sym}{\mathrm{Sym}}
\newcommand{\val}{\mathrm{val}}
\newcommand{\EE}{\mathbb{E}}
\newcommand{\Unif}{\mathrm{Unif}}
\newcommand{\avg}{\mathrm{avg}}

\newcommand{\NN}{\mathbb{N}}

\newcommand{\one}{\mathbbm{1}}

\newcommand{\R}{\mathbb{R}}
\newcommand{\RR}{\mathbb{R}}

\newcommand{\QQ}{\mathbb{Q}}
\newcommand{\PP}{\mathbb{P}}

\newcommand{\sO}{\mathcal{O}}

\newcommand{\sI}{\mathcal{I}}
\newcommand{\sJ}{\mathcal{J}}

\newcommand{\ord}{\mathrm{ord}}

\newcommand{\Id}{\mathrm{Id}}
\newcommand{\opt}{\mathrm{opt}}

\newcommand{\eps}{\varepsilon}
\renewcommand{\epsilon}{\varepsilon}

\newcommand{\Var}{\mathrm{Var}}

\newcommand\numberthis{\addtocounter{equation}{1}\tag{\theequation}}

\DeclareMathOperator{\Adv}{Adv}
\DeclareMathOperator{\poly}{poly}

\usepackage{authblk}

\title{Strong Refutation of Random Ordering CSPs}
\author{Xifan Yu\thanks{Email: \texttt{xifan.yu@yale.edu}. Partially supported by ONR Award N00014-24-1-2611.}\,}

\affil{Department of Computer Science, Yale University}

\begin{document}

\maketitle

\begin{abstract}
    In this work, we initiate the study of strongly refuting the satisfiability of random ordering constraint satisfaction problems. We show that there is a polynomial-time $\eps$-refutation algorithm for random ordering CSP with predicate $P$ when the number of clauses is above the threshold $\tilde{\Omega}\left(n^{d/2}/\eps^2\right)$, where $d$ is the coordinate degree of the predicate $P$. We further give a smooth three-way tradeoff between the running time, the clause density, and the refutation strength $\eps$ using the Kikuchi method. Finally, we complement our algorithmic results with a computational lower bound based on the class of low coordinate degree algorithms, providing evidence that the established three-way tradeoff is near optimal.
\end{abstract}

\allowdisplaybreaks{
\section{Introduction}

Constraint satisfaction problems (CSPs) have been a central object in theoretical computer science, with notable examples such as $3$-SAT, MAX-CUT, and graph coloring. There is a long line of research on CSPs, ranging from worst-case approximability \cite{raghavendra2008optimal}, algebraic properties of CSP predicates \cite{bulatov2005classifying, bulatov2017dichotomy, zhuk2020proof}, and fine-grained complexity \cite{impagliazzo2001complexity}.

Beyond worst-case analysis, there is also a vast literature studying the average-case complexity of CSPs, which has received a lot of attention due to its connection to hardness of approximation \cite{feige2002relations}, proof complexity \cite{ben2002gap}, learning theory \cite{daniely2014average}, cryptography \cite{applebaum2010public}, and statistical physics \cite{achlioptas2003threshold}.

\paragraph{CSP Refutation.}

A natural average-case problem associated to random CSPs is the refutation task. Statistically, a random $k$-CSP is known to be unsatisfiable with high probability once its clause density $\frac{m}{n}$ exceeds a large constant depending on $k$. For $k$-SAT, it is moreover conjectured that there exists a constant $\alpha_k$ such that a random $k$-SAT is unsatisfiable with high probability above the clause density threshold $\frac{m}{n} > \alpha_k$, and satisfiable with high probability below the clause density threshold $\frac{m}{n} < \alpha_k$. This conjecture has been settled by \cite{ding2015proof} for all large enough $k$.

However, given a random CSP instance with large enough clause density, can we efficiently refute its satisfiability, i.e., to certify that the instance is unsatisfiable or to declare failure when we cannot produce such a certificate efficiently? Importantly, a refutation algorithm cannot ``cheat'' by always outputting UNSAT, because there is a tiny chance that a random CSP instance with high clause density is still satisfiable and the algorithm should not make any mistakes on such instances. The goal of a refutation algorithm is therefore to efficiently provide a certificate of unsatisfiability with high probability over the distribution of random CSP. A closely related problem is called the strong refutation task, where the goal is not only to with high probability refute the satisfiability of a random CSP, but also to certify that no assignments can satisfy many more clauses than what a random assignment achieves in expectation.

We remark that the strong refutation task is in stark contrast to the search/optimization task, since a strong refutation algorithm needs to certify that any assignment cannot satisfy too many clauses whereas the search task only needs to find one good assignment.

For $k$-XOR, there is a natural spectral algorithm which can strongly refute a random $k$-XOR instance efficiently above the density $m \ge \Omega(n^{k/2}\poly\log(n))$ \cite{allen2015refute}. When the clause density is below the spectral threshold $n \ll m \ll n^{k/2}$, even though no polynomial-time strong refutation algorithm is known (and it is conjectured that efficient strong refutation is impossible below the threshold $n^{k/2}$), there are subexponential-time algorithms \cite{raghavendra2017strongly}, which run in time $\exp(n^\delta)$ for some $\delta < 1$ depending on $m$ and are faster than the brute-force algorithm which iterates over all the potential assignments, that achieve strong refutation below the spectral threshold.  Recently, \cite{chan2026strongly} generalized the theory and studied the refutation problem for random CSP without literals.

Besides the average-case setting, CSP refutation has also been studied in the semirandom settings. \cite{guruswami2022algorithms} studied a smoothed setting of $k$-XOR where the clause pattern is chosen as a worst-case hypergraph instance and the literals in each clause first receive a worst case sign, and are then flipped independently with some small probability. They showed that in the smoothed setting, efficient strong refutation is still achievable above the spectral threshold $m \ge \Omega(n^{k/2}\poly\log(n))$, and moreover they match the time-density tradeoffs in the smoothed case as in the random case of \cite{raghavendra2017strongly}. \cite{hsieh2023simple} further simplified the analysis and obtained better logarithmic dependencies in the density tradeoff for the strong refutation problem in the smoothed setting. \cite{d2022ihara} studied a different semirandom model of $k$-CSP where the clause pattern is sampled as a random hypergraph but the literals within each clause receive a worst-case sign, and they are able to remove the extra logarithmic factors and match the spectral threshold $n^{k/2}$.

\paragraph{Ordering CSPs.}
Parallel to the theory of CSPs, there is also a line of work studying ordering constraint satisfaction problems (OCSP). In a $k$-OCSP, each clause is specified by a $k$-subset $S$ of variables and an ordering predicate on $S$. The question is to understand how many clauses (local ordering constraints) can be satisfied by a global ordering $\pi \in \Sym([n])$. Similar to CSP, there are also a number of natural problems that can be phrased as ordering CSPs, including the maximum acyclic subgraph, higher-order rank aggregation problem, and scheduling problems with local precedence requirements.

There have been a number of works studying ordering CSPs in the worst-case setting. \cite{bodirsky2010complexity} provided a dichotomy between polynomial-time solvability and NP-completeness for ordering predicates. \cite{guruswami2011beating} showed that under the Unique Games Conjecture, ordering CSPs are approximation resistant. \cite{makarychev2015satisfiability} showed a fixed-parameter tractability result for ordering CSPs. \cite{guruswami2012approximating, makarychev2012local} investigated the approximability of bounded occurrence ordering CSPs. \cite{singer2021streaming, saxena2023streaming} studied the hardness of approximating ordering CSPs in the streaming setting with limited space. \cite{makarychev2026approximation} studied the approximation of nearly satisfiable ordering CSPs. 

In the average-case setting, a few random models of ordering CSPs have been considered in the literature. \cite{goerdt2009random} established upper and lower bounds on the satisfiability threshold for random ordering CSPs with cyclic ordering predicates or betweenness predicates. The noisy sorting problem \cite{braverman2009sorting} that has been extensively studied can be viewed as a model of random $2$-OCSP. \cite{kunisky2025statistical} studied the maximum acyclic subgraph problem on a model of random directed graph, which is also a model for random $2$-OCSP.

Unlike the setting of CSP, to the best of our knowledge, there have been little research on the refutation of random ordering CSPs. In this work, we attempt to extend the well-established theory of CSP refutation to the setting of general ordering CSP.

\paragraph{Statistical-to-Computational Gaps.}

In many high-dimensional statistical problems, the phenomenon of statistical-to-computational gaps has been observed, where there is a regime in which the problem is statistically solvable but polynomial-time methods seem to require a much larger signal strength to work. Strong refutation of CSP is one such problem where a statistical-to-computational gap has long been observed. As discussed earlier, a random $k$-SAT instance is known to be unsatisfiable with high probability above a constant clause density and is thus refutable by the brute-force method, whereas it has been observed that no polynomial-time method seems to achieve strong refutation below the spectral threshold $m \approx n^{k/2}$.

It has been a main focus in the average-case complexity community to provide evidence for these statistical-to-computation gaps. One popular method is to prove failure of special classes of powerful algorithms, including the Sum-of-Squares (SoS) hierarchy, low degree polynomial algorithms, and statistical-query algorithms. For CSP refutation, there are SoS lower bounds \cite{kothari2017sum} matching the conjectural algorithmic threshold at the spectral threshold.

The class of low degree polynomial algorithms has received a lot of attention due to the Low Degree Conjecture, proposed in \cite{hopkins2018statistical}, and it has since then been successful in predicting the computational threshold for many problems. See \cite{kunisky2019notes}. In this work, we follow this route and base our computational lower bounds for refutation of ordering CSPs on the class of low (coordinate) degree polynomial algorithms.

\subsection{Our Contribution}

Motivated by the series of works studying refutation of CSPs, we initiate the study of strong refutation of random ordering CSPs.
\begin{question}
    For a random $k$-OCSP with ordering predicate $P$, at what clause density does there exist an efficient strong refutation algorithm?
\end{question}

We answer this question by showing that conditioned on the Low Degree Conjecture, the threshold of the critical density for polynomial-time refutation is around $m \approx n^{d/2}$, where $d = D(P)$ is what we term as the coordinate degree (see Definition~\ref{def:coordinate-degree}) of the ordering predicate $P$. We furthermore obtain a three-way tradeoff between the running time, the clause density, and the refutation strength (Theorem~\ref{thm:kikuchi}), and provide evidence that the tradeoff is optimal up to a polylogarithmic factor (Theorem~\ref{thm:low-deg-hardness}).

\section{Preliminaries}

\subsection{Notations}

For $n \in \NN$, we will denote $[n] = \{1, 2, \dots, n\}$. Let $S$ be a finite set. We use $|S|$ to denote the cardinality of $S$. For a finite set $S$ with $|S| = s$, we use $\Sym(S)$ to denote the set of orderings on $S$ \footnote{Note that this is different from the symmetric group on $S$, commonly denoted as $\Sym(S)$, which is defined as the group of bijections from $S$ to itself. We overload the notation of $\Sym(S)$ in this work.},
\begin{align*}
    \Sym(S) \colonequals \{\mu = (\mu_1, \dots, \mu_s): \mu_i \text{ distinct}, \{\mu_1, \dots, \mu_s\} = S\},
\end{align*}
i.e., the set of tuples corresponding to ordered list of distinct elements of $S$. For $t \le s$, we use $\sO_t(S)$ to denote the set of ordered list of distinct elements of $S$ of length $t$,
\begin{align*}
    \sO_t(S) \colonequals \{\mu = (\mu_1, \dots, \mu_t): \mu_i \text{ distinct}, \{\mu_1, \dots, \mu_t\} \subseteq S\}.
\end{align*}
For a tuple $\mu = (\mu_1, \dots, \mu_s) \in \Sym(S)$, we will often treat $\mu$ as a function from $[s]$ to $S$, and let $\mu(i) = \mu_i$ for $1\le i \le s$ by a slight abuse of notation.

For $\mu = (\mu_1, \dots, \mu_s) \in \Sym(S)$ and $T \subseteq S$, the induced ordering $\mu[T] \in \Sym(T)$ is the tuple obtained from $(\mu_1, \dots, \mu_s)$ by deleting the elements $\mu_i \in S \setminus T$. For $k \in \NN$, we will use $\Id \in \Sym([k])$ to denote the identity ordering: $\Id = (1, 2, \dots, k)$.

When the dimension is clear, we will use $I \in \RR^{n \times n}$ to denote the identity matrix, and $J \in \RR^{n \times n}$ to denote the all-ones matrix. We will use subscript for $I_S$ ($J_S$, resp.) to denote the corresponding identity matrix (all-ones matrix resp.) supported on rows and columns indexed by $S$. For a real symmetric matrix $A$, we will use $\lambda_{\max}(A)$ to denote its largest eigenvalue, $\|A\|_2$ to denote its spectral norm, and $\|A\|_{\text{Fr}}$ to denote its Frobenius norm.

We will use standard asymptotic notations $O(\cdot), o(\cdot), \Omega(\cdot), \omega(\cdot)$. All the ``with high probability'' statements are understood as taking $n \to \infty$.

\subsection{Probability Tools}

We review a list of concentration inequalities that we will use.

\begin{theorem}[Bernstein Inequality]\label{thm:bernstein}
    Suppose $X_1, \dots, X_N$ are independent zero-mean random variables. Suppose $|X_i| \le M$ for all $i \in [N]$. Then, for all $s \ge 0$,
    \begin{align*}
        \Pr\left(\sum_{i=1}^N X_i \ge s\right) \le \exp\left(-\frac{\frac{1}{2}s^2}{\sum_{i=1}^N \EE[X_i^2] + \frac{1}{3}Ms}\right).
    \end{align*}
\end{theorem}

\begin{theorem}[Matrix Bernstein Inequality]\label{thm:matrix-bernstein}
    Suppose $X_1, \dots, X_N$ are independent zero-mean random real symmetric matrices of dimension $n \times n$. Suppose $\lambda_{\max}(X_i) \le M$ for all $i \in [N]$. Then, for all $s \ge 0$,
    \begin{align*}
        \Pr\left(\lambda_{\max}\left(\sum_{i=1}^N X_i\right) \ge s\right) \le n\cdot\exp\left(-\frac{s^2}{\sigma^2 + \frac{1}{3}Ms}\right),
    \end{align*}
    where
    \begin{align*}
        \sigma^2 \colonequals \|\sum_{i=1}^N \EE[X_i^2]\|_2.
    \end{align*}
\end{theorem}

\begin{theorem}[McDiarmid Inequality]\label{thm:mcdiarmid}
    Let $f: \mathcal{X}_1 \times \dots \times \mathcal{X}_N \to \RR$ be a function. We say $f$ satisfies the bounded difference property with parameters $c_1, \dots, c_N$ if for all $(x_1, \dots, x_N) \in \mathcal{X}_1 \times \dots \times \mathcal{X}_N$, and all $i \in [N]$,
    \begin{align*}
        \sup_{x_i' \in \mathcal{X}_i}\left|f(x_1, \dots, x_{i-1}, x_i, x_{i+1}, \dots, x_N) - f(x_1, \dots, x_{i-1}, x_i', x_{i+1}, \dots, x_N) \right| \le c_i.
    \end{align*}

    Let $X_1, \dots, X_N$ be independent random variables with $X_i$ taking values in $\mathcal{X}_i$. If $f$ satisfies the bounded difference property with parameters $c_1, \dots, c_N$, then
    \begin{align*}
        \Pr\left(f(X_1, \dots, X_N) - \EE[f(X_1, \dots, X_N)] \ge s\right) \le \exp\left(-\frac{2s^2}{\sum_{i=1}^N c_i^2}\right).
    \end{align*}
\end{theorem}

\subsection{Ordering Constraint Satisfaction Problem}

An ordering constraint satisfaction problem (OCSP) on a set of variables $x_1, ..., x_n$ is specified by a collection of clauses $C_1, \dots, C_m$, where each clause $C_t$ is specified by a subset of the variables $x_{S_t} \colonequals \{x_i: i\in S_t\}$ indexed by $S_t$ and an ordering constraint $P_t: \Sym(S) \to \{0,1\}$. A clause on $x_S$ is satisfied by some ordering in $\Sym([n])$ if the induced ordering in $\Sym(S)$ satisfies the clause constraint. The goal is to find an ordering of the variables in $\Sym([n])$ that maximizes the number of satisfied clauses. If every clause of an OCSP instance is supported on $k$ variables, it is a $k$-OCSP instance.

In this paper, we consider the following model of random OCSP.
\begin{definition}[$p$-Random OCSP]
    Let $k \ge 2$ and $P: \Sym([k]) \to \{0,1\}$ be an ordering predicate. A $p$-random $k$-OCSP with predicate $P$ is generated in the following way:
    \begin{itemize}
        \item For every $k$-subset $S \in \binom{[n]}{k}$, sample $S$ independently with probability $p$.
        \item For every sampled $S \in \binom{[n]}{k}$, sample a uniformly random ordering $\nu \in \Sym(S)$, and add a clause $(S, P \circ \nu^{-1})$.
    \end{itemize}
\end{definition}

\begin{remark}
    It is also natural to consider other random models of OCSP. One alternative random model has a fixed number of clauses, with each clause drawn independently from the uniform distribution over all clauses. We expect many such models to behave similarly.
\end{remark}

\subsection{Strong Refutation of Random OCSP}

\begin{definition}[Value of OCSP]
Let $k \ge 2$ and $P: \Sym([k]) \to \{0,1\}$ be an ordering predicate. Let $\sI = \{(S_i, P \circ \nu_i^{-1})\}_{i=1}^m$ be a $k$-OCSP instance with predicate $P$ on $m$ clauses. The value of $\sI$ achieved by some $\pi \in \Sym([n])$ is defined as
\begin{align*}
    \val(\sI; \pi) = \frac{1}{m} \sum_{i=1}^m P(\nu_i^{-1}(\pi[S_i])),
\end{align*}
where $\pi[S_i]$ is the induced ordering of $\pi$ on $S_i$. The value of $\sI$ is defined as the maximum value achieved by any $\pi \in \Sym([n])$
\begin{align*}
    \val(\sI) = \max_{\pi \in \Sym([n])} \val(\sI; \pi).
\end{align*}
\end{definition}

Note that the expected value of $\sI$ achieved by a uniform random permutation $\pi \in \Sym([n])$ is
\begin{align*}
    \underset{\pi \sim \Unif(\Sym([n]))}{\EE}[\val(\sI; \pi)] = \underset{\mu \sim \Unif(\Sym([k]))}{\EE}[P(\mu)] = \frac{\left|P^{-1}(1)\right|}{k!} \equalscolon \avg(P).
\end{align*}

Observe that the number of clauses $m$ follows a binomial distribution $\text{Bin}(\binom{n}{k}, p)$ with mean $\overline{m} \colonequals p\binom{n}{k}$. We have the following concentration result for $m$.
\begin{lemma}\label{lem:m-concentration}
    Let $m \sim \text{Bin}(\binom{n}{k}, p)$ and $\overline{m} \colonequals p\binom{n}{k}$. For any $1 > \eps > 0$, we have
    \begin{align*}
        \Pr\left(|m - \overline{m}| \ge \eps \overline{m}\right) \le 2 \exp\left(-\frac{3\eps^2 \overline{m}}{8}\right). 
    \end{align*}
\end{lemma}

\begin{proof}
    By Bernstein inequality (Theorem~\ref{thm:bernstein}),
    \begin{align*}
        \Pr\left(|m - \overline{m}| \ge \eps \overline{m}\right) &\le 2 \exp\left(-\frac{\frac{1}{2}\eps^2 \overline{m}^2}{\binom{n}{k}p(1-p)+\frac{1}{3}\eps \overline{m}}\right)\\
        &\le 2 \exp\left(-\frac{\frac{1}{2}\eps^2 \overline{m}^2}{\overline{m}+\frac{1}{3}\eps \overline{m}}\right)\\
        &\le 2 \exp\left(-\frac{3\eps^2 \overline{m}}{8}\right).
    \end{align*}
\end{proof}

Suppose that the predicate $P$ is not trivial, i.e., $0 < \avg(P) < 1$. By a simple union bound over all $\pi \in \Sym([n])$, we have the following result on the typical value of a random OCSP instance.
\begin{proposition}
    Let $k \ge 2$ and $P: \Sym([k]) \to \{0,1\}$ be a non-trivial ordering predicate. There exists a constant $C > 0$ such that for any $\eps = \eps(n) > 0$, if $p = p(n) > 0$ satisfies 
    $$\overline{m} \ge \frac{C}{\eps^2}n\log n$$
    where $\overline{m} \colonequals p \binom{n}{k}$, then a $p$-random $k$-OCSP instance $\sI$ with predicate $P$ satisfies
    \begin{align*}
        \Pr\left(\left|\val(\sI) - \avg(P)\right| \ge \eps\right) = o(1). 
    \end{align*}
\end{proposition}

\begin{proof}
    Fix an arbitrary $\pi \in \Sym([n])$. Consider a $p$-random $k$-OCSP instance $\sI = \{(S_i, P \circ \nu_i^{-1})\}_{i=1}^m$ with predicate $P$ on $m$ clauses. Note that $m \sim \text{Bin}\left(\binom{n}{k}, p\right)$ is a binomial random variable. We have
    \begin{align*}
        \val(\sI; \pi) = \frac{1}{m} \sum_{i=1}^m P(\nu_i^{-1}(\pi[S_i])).
    \end{align*}

    For any $S \in \binom{[n]}{k}$, let $I_{S}$ denote the indicator variable that $S \in \binom{[n]}{k}$ is sampled by the $p$-random instance, and let $\nu_S \sim \Unif(\Sym(S))$ so that whenever $S$ is sampled, the clause $(S, P \circ \nu_S^{-1})$ is included in the instance. Then, the random variables $I_{S}$ are i.i.d.~$\text{Bern}(p)$ variables, and we may write the numerator and denominator of $\val(\sI; \pi)$ as
    \begin{align*}
        m &= \sum_{\substack{S \in \binom{[n]}{k}}} I_{S},\\
        \sum_{i=1}^m P(\nu_i^{-1}(\pi[S_i])) &= \sum_{S \in \binom{[n]}{k},} I_{S}\cdot P(\nu_S^{-1}(\pi[S])).
    \end{align*}
    Note that both $I_{S}$ and $I_{S} \cdot P(\nu_S^{-1}(\pi[S]))$ are Bernoulli random variables. Moreover, we have
    \begin{align*}
        \Var(I_{S}) &= p(1-p),\\
        \Var\left(I_{S}\left( P(\nu_S^{-1}(\pi[S])) - \avg(P)\right)\right) &\le p(1-p),\\
        \EE[m] &= p \binom{n}{k},\\
        \EE\left[\sum_{S \in \binom{[n]}{k}} I_{S} \left(P(\nu_S^{-1}(\pi[S])) - \avg(P)\right)\right] &= p \sum_{S \in \binom{[n]}{k}} \underset{\nu_S \sim \Unif(\Sym(S))}{\EE} [P(\nu_S^{-1}(\pi[S])) - \avg(P)]\\
        &= 0.
    \end{align*}
    Denote $\overline{m} \colonequals p \binom{n}{k}$.
    By Lemma~\ref{lem:m-concentration},
    \begin{align*}
        \Pr\left(\overline{m} - m \ge \frac{1}{2} \overline{m} \right) &\le  2\exp\left(-\frac{3\overline{m}}{32} \right).
    \end{align*}
    By Bernstein inequality (Theorem~\ref{thm:bernstein}), we get
    \begin{align*}
        \Pr\left(\left|\sum_{S \in \binom{[n]}{k}} I_{S} \left(P(\nu_S^{-1}(\pi[S])) - \avg(P)\right)\right| \ge \frac{1}{2}\eps\overline{m} \right) &\le 2\exp\left(- \frac{\frac{1}{8}\eps^2 \overline{m}^2}{\binom{n}{k}p(1-p)  + \frac{1}{6} \eps \overline{m} }\right)\\
        &\le 2\exp\left(- \frac{\eps^2 \overline{m}^2}{8\overline{m}(1 + \frac{1}{6}\eps)}\right)\\
        &\le 2\exp\left(- \frac{3\eps^2 \overline{m}}{28}\right).
    \end{align*}
    Finally, since
    \begin{align*}
        m &= \sum_{S \in \binom{[n]}{k}} I_{S},\\
        \sum_{i=1}^m \left(P(\nu_i^{-1}(\pi[S_i])) - \avg(P)\right) &=\sum_{S \in \binom{[n]}{k}} I_{S} \left(P(\nu_S^{-1}(\pi[S])) - \avg(P)\right),
    \end{align*}
    if both $m > \frac{1}{2}\overline{m}$ and $|\sum I_{S} \left(P(\nu_S^{-1}(\pi[S])) - \avg(P)\right)| < \frac{1}{2}\eps\overline{m}$, then
    \begin{align*}
        |\val(\sI; \pi) - \avg(P)| &= \left|\frac{1}{m} \sum_{i=1}^m P(\nu_i^{-1}(\pi[S_i])) - \avg(P)\right|\\
        &= \frac{1}{m}\left|\sum_{S \in \binom{[n]}{k}} I_{S} \left(P(\nu_S^{-1}(\pi[S])) - \avg(P)\right)\right|\\
        &< \frac{\frac{1}{2}\eps \overline{m}}{\frac{1}{2}\overline{m}}\\
        &= \eps. 
    \end{align*}
    We therefore conclude that for fixed $\pi \in \Sym([n])$,
    \begin{align*}
        \Pr(|\val(\sI; \pi) - \avg(P)| \ge \eps) \le 2\exp\left(- \frac{3\eps^2 \overline{m}}{28}\right) + 2\exp\left(-\frac{3\overline{m}}{32} \right) \le 4\exp\left(- \frac{3\eps^2 \overline{m}}{32}\right).
    \end{align*}
    Taking a union bound over all $\pi \in \Sym([n])$, we have
    \begin{align*}
        \Pr\left(\left|\val(\sI) - \avg(P)\right| \ge \eps\right) &\le n!\cdot 4\exp\left(- \frac{3\eps^2 \overline{m}}{32}\right)\\
        &\le 4\exp\left(n\log n-\frac{3\eps^2 \overline{m}}{32}\right).
    \end{align*}
    When $\overline{m} \ge \frac{C}{\eps^2} n\log n$ for some constant $C > \frac{32}{3}$, we get the desired bound.
\end{proof}

While once $\overline{m} \ge \tilde{\Omega}(n)$, we know that with high probability the random $k$-OCSP satisfies $\val(\sI) \le \avg(P) + \eps$ for a fixed constant $\eps > 0$, algorithmically certifying this bound is non-trivial. In this paper, we are interested in the following task of strong refutation of OCSP satisfiability when a random OCSP instance is far from being satisfiable.
\begin{definition}[Strong Refutation of $k$-OCSP]
    Let $k \ge 2$ and $P: \Sym([k]) \to \{0,1\}$ be a non-trivial ordering predicate. Let $\eps = \eps(n) > 0$. An algorithm is said to achieve $\eps$-refutation of a random $k$-OCSP instance $\sI$ with predicate $P$ if with high probability, the algorithm returns a certificate that $\val(\sI) \le \avg(P) + \eps$.
\end{definition}
The naive strong refutation algorithm proceeds by computing the value of the random instance on each of the $n!$ orderings in $\Sym([n])$, but it takes time $\exp(O(n\log n))$. The goal of this work is to investigate the clause density of OCSP above which efficient strong refutation algorithms exist, and the three-way tradeoff between the running time of the refutation algorithm, the clause density $p$, and the refutation strength $\eps$.

\subsection{Efron-Stein Decomposition}

Before we state our main results, let us give an overview of the Efron-Stein Decomposition, a useful tool to study OCSP, which has previously been used in \cite{makarychev2015satisfiability} in the context of fixed-parameter tractability of the number of satisfiable clauses in OCSP, and in \cite{chan2026strongly} in the context of refutation of CSP without literals.

To set up the Efron-Stein decomposition for OCSP, we follow the idea in \cite{makarychev2015satisfiability} and extend the domain of OCSP from the set of permutations $\Sym([n])$ to a product domain $[0, 1]^n$. Then, for any $x = (x_1, \dots, x_n) \in [0,1]^n$, we define the value of an OCSP on $x$ to be its value on the induced ordering $\ord(x)$ defined by ordering $x_1, \dots, x_n \in [0,1]$ under a fixed tie-breaking rule. Note that for a uniform random $x \sim \Unif([0,1]^n)$, the induced permutation $\ord(x)$ follows the uniform distribution $\Unif(\Sym([n]))$.

At a high level, the Efron-Stein decomposition decomposes any function $f \in L^2([0,1]^n)$ into mutually orthogonal pieces indexed by subsets of coordinates, similar to how Boolean Fourier analysis decomposes a Boolean function into the orthonormal Walsh-Fourier basis. 

\begin{theorem}[Efron-Stein Decomposition (see {\cite[Section 3.3]{makarychev2015satisfiability}})] \label{thm:EF}
    Let $n \in \NN$ and $f\in L^2([0,1]^n)$. The Efron-Stein decomposition of $f$ is
    \begin{align*}
        f(x_1, \dots, x_n) &= \sum_{S \subseteq [n]} f_S(x_S),
    \end{align*}
    where $x_S = (x_i)_{i \in S}$ is the collection of variables indexed by $S$, and 
    \begin{align*}
        f_S(x_S) \colonequals \sum_{T \subseteq S} (-1)^{|S| - |T|} \underset{x \sim \Unif([0,1]^n)}{\EE}[f(x) \vert x_T].
    \end{align*}
\end{theorem}

The Efron-Stein decomposition satisfies the following property.
\begin{fact}[Orthogonality of Efron-Stein decomposition {\cite[Proposition A.4]{kunisky2025low}}]\label{fact:ES-orthogonal}
    For $f, g \in L^2([0,1]^n)$, let
    \begin{align*}
        f(x) &= \sum_{S \subseteq [n]} f_S(x_S),\\
        g(x) &= \sum_{S \subseteq [n]} g_S(x_S),
    \end{align*}
    be their Efron-Stein decompositions. Then, for any $S \ne T \subseteq [n]$,
    \begin{align*}
        \underset{x \sim \Unif([0,1]^n)}{\EE}[f_S(x_S)g_T(x_T)] = 0.
    \end{align*}
\end{fact}

\begin{definition}\label{def:coordinate-degree}
    Let $n \in \NN$ and $f \in L^2([0,1]^n)$. Let the Efron-Stein decomposition of $f$ be
    \begin{align*}
        f(x_1, \dots, x_n) = \sum_{S \subseteq [n]} f_S(x_S).
    \end{align*}

    For $d \in \NN$, define
    \begin{align*}
        f^{=d}(x) \colonequals \sum_{\substack{S \subseteq [n]:\\ |S|=d}} f_S(x_S).
    \end{align*}

    Then, the coordinate degree of $f$ is
    \begin{align*}
        D(f) \colonequals \max\left\{d \in \NN: f^{=d} \ne 0\right\}.
    \end{align*}
\end{definition}

\begin{fact}\label{fact:junta-degree}
    If a function $f \in L^2([0,1]^n)$ depends only on a subset of at most $d$ coordinates, i.e., $f$ is a $d$-junta, then its coordinate degree $D(f)$ is at most $d$.
\end{fact}

In the setting of OCSP, let $P: \Sym([k]) \to \{0,1\}$ be any ordering predicate. Let $f: [0,1]^k \to \{0,1\}$ be defined as $f(x) = P(\ord(x))$. By a slight abuse of notation, we define the coordinate degree $D(P)$ of the predicate $P$ to be the coordinate degree $D(f)$ of $f$.

\subsection{Average-Case Complexity}

Since the problem of strongly refuting an OCSP is an average-case problem, classical theory of NP-hardness does not apply. Therefore, we need to rely on some notion of average-case hardness. A long line of work has developed various frameworks of showing computational hardness for average-case problems against concrete classes of problems, including Sum-of-Squares, low-degree polynomials, statistical-query algorithms.

In this work, we will focus on average-case computational hardness based on the Low Degree Conjecture, informally stated as follows.

\begin{definition}[Strong Detection]
    Let $N = N(n) \in \NN$. Let $\QQ = \QQ_n$ and $\PP = \PP_n$ be two sequences of distributions supported on $\Omega^N$. A test $f: \Omega^N \to \{0,1\}$ is said to achieve strong detection between $\PP$ and $\QQ$ if
    \begin{align*}
        \PP(f = 0) + \QQ(f = 1) = o(1).
    \end{align*}
\end{definition}

\begin{definition}[Strong Separation]
    Let $N = N(n) \in \NN$. Let $\QQ = \QQ_n$ and $\PP = \PP_n$ be two sequences of distributions supported on $\Omega^N$. A function $f: \Omega^N \to \RR$ is said to achieve strong separation between $\PP$ and $\QQ$ if
    \begin{align*}
        \max\left\{\sqrt{\Var_{\QQ}[f], \Var_{\PP}[f] }\right\} = o\left(|\EE_{\PP}[f] - \EE_{\QQ}[f]|\right).
    \end{align*}
\end{definition}

Clearly, by Chebyshev's inequality, thresholding a function that achieves strong separation gives a test that achieves strong detection. Low Degree Conjecture posits that for natural high dimensional hypothesis testing problems, the failure of low degree polynomials in achieving strong separation rules out any efficient algorithms for strong detection.

\begin{conjecture}[Low Degree Conjecture (Informal. See \cite{hopkins2018statistical})]
    Let $N = N(n) \in \NN$, where $N$ is bounded by a polynomial in $n$. Let $\QQ = \QQ_n$ and $\PP = \PP_n$ be sequences of distributions supported on $\Omega^N$, where $\QQ$ is a product distribution. For ``sufficiently nice'' sequences of distributions $\QQ$ and $\PP$, if no polynomial of degree at most $D$ achieves strong separation between $\PP$ and $\QQ$ for some $D \ge \log(n)^{1 + \eps}$ and constant $\eps > 0$, then no polynomial time algorithm achieves strong detection between $\PP$ and $\QQ$.
\end{conjecture}

At a high level, low (coordinate) degree functions are regarded by the Low Degree Conjecture as a proxy for efficient algorithms for solving ``nice'' high-dimensional hypothesis testing problems. The failure of this class of algorithms constitutes evidence for the average-case computational hardness of a hypothesis testing problem. We defer more detailed discussion of low coordinate degree algorithms and the version of Low Degree Conjecture (See Conjecture~\ref{conj:generalized-low-deg-conj}) we will use to Section~\ref{sec:low-deg}.

\section{Main Results}

\subsection{Algorithms for Refutation}

We now state our algorithmic results for strong refutation of OCSP. Our first result is the following:

\begin{theorem}[Kikuchi Method for Refuting OCSP] \label{thm:kikuchi}
    Let $k \ge 2$ and $P: \Sym([k]) \to \{0,1\}$ be a non-trivial ordering predicate with coordinate degree $d$. Let $p = p(n) > 0$, $\eps = \eps(n) > 0$, and $\ell = \ell(n) \in \NN$. There exist absolute constants $\alpha = \alpha(k) > 0$ and $C = C(k) > 0$ depending only on $k$ such that if $d-1 \le \ell \le \alpha n$, and
    \begin{align*}
        \overline{m} \ge \begin{cases}
            C\frac{n}{\eps^2}\left(\frac{n}{\ell}\right)^{\frac{d}{2}-1}\log(n)^{2d-1} & \quad \text{even } d,\\
            C\frac{n}{\eps^2}\left(\frac{n}{\ell}\right)^{\frac{d}{2}-1}\log(n)^{2d-\frac{3}{2}} & \quad \text{odd } d,
        \end{cases}
    \end{align*}
    where $\overline{m} = p \binom{n}{k}$, then there is an algorithm that achieves $\eps$-refutation of a $p$-random OCSP instance with predicate $P$ on $n$ variables, and runs in time $n^{O(\ell)}$. 
\end{theorem}

\begin{remark}
    The result is based on the Kikuchi method, recently introduced by \cite{wein2019kikuchi} as a method to design spectral algorithms for problems involving tensors. The Kikuchi method has since then found a variety of applications \cite{alrabiah2023near, kothari2024exponentiala, kothari2024exponentialb}, including CSP refutation \cite{guruswami2022algorithms, hsieh2023simple, chan2026strongly}. However, we emphasize that the Kikuchi method does not straightforwardly apply to the setting of OCSPs due to the non-product domain $\Sym([n])$ of OCSP. In this work, we introduce a rank decomposition procedure (see Section~\ref{sec:rank-decom}) to apply the Kikuchi method in the OCSP setting, which comes at the cost of a polylogarithmic blowup in the strong refutation guarantee compared to the setting of CSPs (see \cite[Theorem 1.5]{chan2026strongly}).
\end{remark}

Our second result removes the extra polylogarithmic loss in Theorem~\ref{thm:kikuchi} when the refutation strength $\eps > 0$ is an absolute constant, which we obtain via a black-box reduction to the CSP setting using what we call the bucketing method and matches the guarantee in \cite[Theorem 1.5]{chan2026strongly}.

\begin{theorem}[Bucketing Method for Refuting OCSP]\label{thm:bucketing}
    Let $k \ge 2$, $P: \Sym([k]) \to \{0,1\}$ be a non-trivial ordering predicate with coordinate degree $d$, and $\eps > 0$ be a constant. Let $p = p(n) > 0$ and $\ell = \ell(n) \in \NN$. Then, there exist constants $\alpha = \alpha(k,\eps) > 0$ and $C = C(k, \eps) > 0$ depending only on $k$ and $\eps$ such that if $d-1\le \ell \le \alpha n$ and
    \begin{align*}
        \overline{m} \ge 
            C \cdot n \left(\frac{n}{\ell}\right)^{\frac{d}{2}-1} \log(n)
    \end{align*}
    where $\overline{m} = p \binom{n}{k}$, then there is an algorithm that achieves $\eps$-refutation of a $p$-random OCSP instance with predicate $P$ on $n$ variables, and runs in time $n^{O_{k,\eps}(\ell)}$. 
\end{theorem}

\begin{remark}
    We emphasize that Theorem~\ref{thm:bucketing} stops working when $\eps = \eps(n) > 0$ gets too small as $n \to \infty$. The reason is that the exponent of running time $n^{O_{\eps}(l)}$ will blow up when $\eps$ is not constant.
\end{remark}

\subsection{Lower Bound for Refutation}

Next, we state the following computational lower bound for strong refutation of OCSP, which shows that our algorithmic results are essentially optimal.

\begin{theorem}[Hardness of Refuting OCSP]\label{thm:low-deg-hardness}
    Let $k \ge 2$ and $P: \Sym([k]) \to \{0,1\}$ be a non-trivial ordering predicate with coordinate degree $d$. Let $p = p(n) > 0$, $\eps = \eps(n) > 0$, and $D = D(n) \in \NN$. Suppose that
    \begin{align*}
        \eps \in \left(0,\frac{1}{2^{2k+1}(2k)!}\right], \quad D \le \frac{n}{k}.
    \end{align*}
    Assuming the Generalized Low Degree Conjecture (see Conjecture~\ref{conj:generalized-low-deg-conj}), if 
    \begin{align*}
        \overline{m} = o\left(\frac{n}{\eps^2} \left(\frac{n}{D}\right)^{\frac{d}{2}-1}\right), \text{ and } \overline{m} \eps^2 = \omega(1)
    \end{align*}
    where $\overline{m} = p \binom{n}{k}$, then no algorithm that runs in time $\exp(O(D/\poly\log(n)))$ achieves $\eps$-refutation of a $p$-random OCSP instance with predicate $P$ on $n$ variables.
\end{theorem}

\begin{remark}
    We note that the lower bound in Theorem~\ref{thm:low-deg-hardness} matches the threshold of Theorem~\ref{thm:kikuchi} and Theorem~\ref{thm:bucketing} up to a polylogarithmic factor.
\end{remark}

\section{Refutation of OCSP}

\subsection{Reduction to Refutation of Canonical Ordering Predicate}

We first show that to obtain strong refutation of a generic ordering predicate, it is enough to obtain a two-sided strong refutation of canonical ordering predicates. To this end, we apply the Efron-Stein decomposition to the setting of OCSP. 

\begin{definition}[Canonical Ordering Predicate]
    Let $S$ be a finite set and $\mu \in \Sym(S)$. The canonical ordering predicate $\Phi_{S, \mu}: \Sym(S) \to \{0,1\}$ is defined as
    \begin{align*}
        \Phi_{S, \mu}(\tau) \colonequals \one\{\tau = \mu\}, \quad \forall\, \tau \in \Sym(S).
    \end{align*}
\end{definition}

\begin{proposition}\label{prop:ordering-predicate-decom}
    Let $k \ge 2$ and $P: \Sym([k]) \to \{0,1\}$ be an ordering predicate with coordinate degree $D(P) = d$. Then, the predicate $P$ can be expressed as
    \begin{align*}
        P(\mu) = C_P + \sum_{\substack{T \subseteq [k]:\\ 2\le |T| \le d}} \sum_{\sigma \in \Sym(T)} c_{T, \sigma} \cdot \Phi_{T, \sigma}(\mu[T]),
    \end{align*}
for some coefficients $c_{T, \sigma} \in \RR$ for $T \subseteq [k]$ and $\sigma \in \Sym(T)$ and a constant $C_P \in \RR$. Moreover, the coefficients satisfy
    \begin{align*}
        |c_{T, \sigma}| \le 2^{|T|}.
    \end{align*}
\end{proposition}

\begin{proof}
    Let $f: [0,1]^k \to \{0,1\}$ be defined as
    \begin{align*}
        f(x_1, \dots, x_k) = P(\ord(x)),
    \end{align*}
    and
    \begin{align*}
        f(x) = \sum_{\substack{S \subseteq [k]:\\ |S|\le d}} f_S(x_S)
    \end{align*}
    be its Efron-Stein decomposition, where the summation is over $S \subseteq [k]$ of size at most $d$ since $D(f) = D(P) = d$. Recall that the Efron-Stein components take the form
    \begin{align*}
        f_S(x_S) = \sum_{T \subseteq S} (-1)^{|S| - |T|} \underset{x \sim \Unif([0,1]^k)}{\EE}[f(x) \vert x_T].
    \end{align*}

    Now fix any $z = (z_1, \dots, z_k) \in [0, 1/3]^k$ of distinct elements. For any $U \subseteq [k]$, define a map $h_{U}: [0,1]^k \to [0,1]^k$ by
    \begin{align*}
        h_U(x)_i = \begin{cases}
            x_i & \quad \text{ if } i \in U,\\
            z_i & \quad \text{ if } i \notin U. 
        \end{cases}
    \end{align*}
    Using the identity
    \begin{align*}
        \sum_{\substack{T \subseteq [k]:\\ T \supseteq U}} (-1)^{|T| - |U|} &= \begin{cases}
            1 & \quad \text{ if } U = [k],\\
            0 & \quad \text{ otherwise},
        \end{cases}
    \end{align*}
    we can rewrite $f$ as
    \begin{align*}
        f(x) &= \sum_{U \subseteq [k]} f(h_U(x)) \one\{U = [k]\}\\
        &= \sum_{U \subseteq [k]} f(h_U(x)) \sum_{\substack{T \subseteq [k]:\\ T \supseteq U}} (-1)^{|T| - |U|}\\
        &= \sum_{T \subseteq [k]} \sum_{U \subseteq T} (-1)^{|T|-|U|} f(h_U(x)). \numberthis \label{eq:mobius}
    \end{align*}

    Plugging the Efron-Stein decomposition of $f$ into \eqref{eq:mobius}, we get
    \begin{align*}
        f(x) &= \sum_{T \subseteq [k]} \sum_{U \subseteq T} (-1)^{|T|-|U|} f(h_U(x))\\
        &= \sum_{T \subseteq [k]} \sum_{U \subseteq T} (-1)^{|T|-|U|} \sum_{\substack{S \subseteq [k]: \\ |S| \le d}} f_S(h_U(x)) \numberthis \label{eq:mobius-ES}
    \end{align*}
    We claim that only $T \subseteq [k]$ of size at most $d$ contributes to the sum \eqref{eq:mobius-ES}. To see this, fix any $T \subseteq [k]$ such that $|T| \ge d+1$. Then, for any $S \subseteq [k]$, $|S| \le d$, there exists $i\in T \setminus S$, and thus
    \begin{align*}
        &\sum_{U \subseteq T} (-1)^{|T|-|U|} f_S(h_U(x))\\
        &=  \sum_{U \subseteq T \setminus \{i\}} (-1)^{|T|-|U|} f_S(h_U(x)) + \sum_{U \subseteq T \setminus \{i\}} (-1)^{|T|-|U \cup \{i\}|} f_S(h_{U \cup \{i\}}(x))\\
        &= \sum_{U \subseteq T \setminus \{i\}} (-1)^{|T|-|U|} \left(f_S(h_U(x)) - f_S(h_{U \cup \{i\}}(x))\right)\\
        &= 0,
    \end{align*}
    since $h_U(x)$ and $h_{U \cup \{i\}}(x)$ are identical on coordinates indexed by $S$.
    
    As a result, for any $T \subseteq [k]$, $|T| \ge d+1$, the inner sum vanishes as follows
    \begin{align*}
        &\sum_{U \subseteq T} (-1)^{|T|-|U|} \sum_{\substack{S \subseteq [k]: \\ |S| \le d}} f_S(h_U(x))\\
        &= \sum_{\substack{S \subseteq [k]: \\ |S| \le d}} \sum_{U \subseteq T} (-1)^{|T|-|U|}  f_S(h_U(x))\\
        &= 0.
    \end{align*}
    We can therefore further rewrite \eqref{eq:mobius-ES} as
    \begin{align*}
        f(x) &= \sum_{T \subseteq [k]} \sum_{U \subseteq T} (-1)^{|T|-|U|} \sum_{\substack{S \subseteq [k]: \\ |S| \le d}} f_S(h_U(x))\\
        &= \sum_{\substack{T \subseteq [k]:\\ |T| \le d}} \sum_{U \subseteq T} (-1)^{|T|-|U|} \sum_{\substack{S \subseteq [k]: \\ |S| \le d}} f_S(h_U(x))\\
        &= \sum_{\substack{T \subseteq [k]:\\ |T| \le d}} \sum_{U \subseteq T} (-1)^{|T|-|U|} f(h_U(x)), \numberthis \label{eq:local-ordering-decom}
    \end{align*}
    where in the last step we substitute the Efron-Stein decomposition back with $f$.

    Define $g_T: [0,1]^k \to \RR$ by
    \begin{align*}
        g_T(x) \colonequals \sum_{U \subseteq T} (-1)^{|T|-|U|} f(h_U(x)).
    \end{align*}
    Then, by \eqref{eq:local-ordering-decom}, we have
    \begin{align}
        f(x) = \sum_{\substack{T \subseteq [k]:\\ |T| \le d }} g_T(x). \label{eq:gT-expansion}
    \end{align}

    We will next show that for any $T \subseteq [k]$, $|T| \le d$ and any $x \in [2/3, 1]^k$, the value of $g_T(x)$ only depends on the induced ordering $\ord(x_T)$ of the coordinates indexed by $T$. Consider any two $x, y\in [2/3, 1]^k$ whose induced orderings $\ord(x_T)$ and $\ord(y_T)$ are identical. Recall that $h_U(x), h_U(y) \in [0,1]^k$ keep all the coordinates of $x$ and $y$ indexed by $U$, and replace all the coordinates of $x$ and $y$ outside of $U$ by coordinates of the fixed vector $z \in [0, 1/3]^k$. In particular, since the coordinates of $x$ and $y$ take value in $[2/3, 1]$, the coordinates of $h_U(x)$ and $h_U(y)$ outside of $U$ precede the coordinates indexed by $U$. Moreover, since the coordinates of $h_U(x)$ and $h_U(y)$ outside of $U$ are identical, and the coordinates of $h_U(x)$ and $h_U(y)$ indexed by $U$ have the same ordering, we conclude that $\ord(h_U(x)) = \ord(h_U(y)) \in \Sym([k])$ for any $U \subseteq T$. Consequently, 
    \begin{align*}
        g_T(x) &= \sum_{U \subseteq T} (-1)^{|T|-|U|} f(h_U(x))\\
        &= \sum_{U \subseteq T} (-1)^{|T|-|U|} P(\ord(h_U(x)))\\
        &= \sum_{U \subseteq T} (-1)^{|T|-|U|} P(\ord(h_U(y)))\\
        &= \sum_{U \subseteq T} (-1)^{|T|-|U|} f(h_U(y))\\
        &= g_T(y).
    \end{align*}

    Also note that
    \begin{align*}
        |g_T(x)| &= \left|\sum_{U \subseteq T} (-1)^{|T|-|U|} f(h_U(x))\right| \le 2^{|T|}
    \end{align*}
    as $|f(x)| \le 1$. Thus, for any $x \in [2/3, 1]^k$, since $g_T(x)$ depends only on $\ord(x_T)$, we have
    \begin{align}
        g_T(x) &= \sum_{\sigma \in \Sym(T)} \one\{\ord(x_T) = \sigma\} \cdot c_{T, \sigma} = \sum_{\sigma \in \Sym(T)} c_{T, \sigma}\cdot  \Phi_{T, \sigma}(\ord(x_T)), \label{eq:gT}
    \end{align}
    where $c_{T, \sigma} \colonequals g_T(y)$ for any arbitrary $y \in [2/3, 1]^k$ such that $\ord(y_T) = \sigma$, and the coefficients satisfy $|c_{T, \sigma}| = |g_T(y)| \le 2^{|T|}$.

    Specially, if $T = \emptyset$ or $|T| = 1$, then $g_T(x)$ is a constant function since there is a unique ordering on $T$. For any $\mu \in \Sym([k])$, pick $x^{(\mu)} \in [2/3, 1]^k$ such that $\ord(x^{(\mu)}) = \mu$. Substituting \eqref{eq:gT} into \eqref{eq:gT-expansion} and using the definition that $f(x) = P(\ord(x))$, we get that for any $\mu \in \Sym([k])$,
    \begin{align*}
        P(\mu) &= f(x^{(\mu)})\\
        &= \sum_{\substack{T\subseteq [k]:\\ |T| \le d}} g_T(x^{(\mu)})\\
        &= C_P + \sum_{\substack{T\subseteq [k]:\\ 2 \le |T| \le d}} \sum_{\sigma \in \Sym(T)} c_{T, \sigma}\cdot  \Phi_{T, \sigma}(\ord(x^{(\mu)}_T))\\
        &= C_P + \sum_{\substack{T\subseteq [k]:\\ 2 \le |T| \le d}} \sum_{\sigma \in \Sym(T)} c_{T, \sigma}\cdot  \Phi_{T, \sigma}(\mu[T]),
    \end{align*}
    where $C_P$ is the sum of the constant functions $g_T$ indexed by $T \subseteq [k]$ with $|T| \le 1$. This finishes the proof of the canonical ordering predicate decomposition of a generic ordering predicate $P$, and the coefficients in the decomposition satisfy the required bound $|c_{T, \sigma}|\le 2^{|T|}$.

    \end{proof}

    Now we show that to strongly refute a random OCSP instance with a generic ordering predicate $P: \Sym([k]) \to \{0,1\}$, it is enough to give a two-sided strong refutation of some specific random OCSP instance with canonical ordering predicates using the decomposition in Proposition~\ref{prop:ordering-predicate-decom}.

    \begin{definition}[Two-Sided Strong Refutation]
        Let $k \ge 2$ and $P: \Sym([k]) \to \{0,1\}$ be a non-trivial ordering predicate. Let $\eps = \eps(n) > 0$. An algorithm is said to achieve two-sided $\eps$-refutation of a random $k$-OCSP instance $\sI$ with predicate $P$ if with high probability, the algorithm returns a certificate that $|\val(\sI; \pi) - \avg(P)| \le \eps$ for every $\pi \in \Sym([n])$.
    \end{definition}

    \begin{definition}[$(k,p)$-Random $t$-OCSP]\label{def:kp-OCSP}
        Let $k \ge 2$, $2 \le t \le k$, and $P: \Sym([t]) \to \{0,1\}$ be an ordering predicate. A $(k,p)$-random $t$-OCSP with predicate $P$ is generated in the following way:
    \begin{itemize}
        \item For every $k$-subset $S \in \binom{[n]}{k}$, sample $S$ independently with probability $p$.
        \item For every sampled $S \in \binom{[n]}{k}$, sample $\nu = (\nu_1, \dots, \nu_t) \in \sO_t(S)$ independently and uniformly at random, and add a clause $(T, P \circ \nu^{-1})$ where $T = \{\nu_1, \dots, \nu_t\}$.
    \end{itemize}
    \end{definition}

    We note that a $(k,p)$-random $t$-OCSP instance could potentially contain multiple clauses supported on the same $T \in \binom{[n]}{t}$. $(k,p)$-random OCSP is a distribution of random OCSP that naturally arises when we apply the decomposition in Proposition~\ref{prop:ordering-predicate-decom}. The next proposition shows how to turn two-sided refutation for $(k,p)$-random $t$-OCSP with canonical ordering predicates for all $t \le d$ into a refutation algorithm for $p$-random OCSP with a generic ordering predicate with coordinate degree $d$.
    \begin{proposition}\label{prop:reduction-to-canonical}
        Let $k \ge 2$ and $P: \Sym([k]) \to \{0,1\}$ be a non-trivial ordering predicate with coordinate degree $D(P) = d$. Let $p = p(n) > 0$ and $\eps = \eps(n) > 0$.

        Suppose that for every $2\le t \le d$, there exists a two-sided $\eps$-refutation algorithm for a $(k, p)$-random $t$-OCSP instance with the canonical ordering predicate $\Phi_{[t], \Id}$ that runs in time $F_t$, then there exists an $(2^{2k}k!)\eps$-refutation algorithm for a $p$-random $k$-OCSP instance with the predicate $P$ that runs in time $O\left(\overline{m} + \max_{2 \le t \le d} F_t\right)$, where $\overline{m} = p \binom{n}{k}$.
    \end{proposition}

    \begin{proof}
        By Proposition~\ref{prop:ordering-predicate-decom}, there exists a decomposition of the predicate $P$ as
        \begin{align}
            P(\mu) = C_P + \sum_{\substack{T \subseteq [k]:\\ 2\le |T| \le d}} \sum_{\sigma \in \Sym(T)} c_{T, \sigma} \cdot \Phi_{T, \sigma}(\mu[T]), \label{eq:P-decom}
        \end{align}
        for some coefficients $c_{T, \sigma} \in \RR$ for $T \subseteq [k]$ and $\sigma \in \Sym(T)$ that satisfies $|c_{T, \sigma}| \le 2^{|T|}$ and a constant $C_P \in \RR$. Note that we have
        \begin{align*}
            \avg(P) &= \underset{\mu \sim \Unif(\Sym([k]))}{\EE}\left[C_P + \sum_{\substack{T \subseteq [k]:\\ 2\le |T| \le d}} \sum_{\sigma \in \Sym(T)} c_{T, \sigma} \cdot \Phi_{T, \sigma}(\mu[T])\right]\\
            &= C_P + \sum_{\substack{T \subseteq [k]:\\ 2\le |T| \le d}} \sum_{\sigma \in \Sym(T)} c_{T, \sigma} \cdot \frac{1}{|T|!}. \numberthis \label{eq:avg-decom}
        \end{align*}
        Moreover, the decomposition \eqref{eq:P-decom} can be computed in $O_k(1)$ time.

        Given this decomposition \eqref{eq:P-decom} of the predicate $P$, we may decompose the value of a random $k$-OCSP instance $\sI = \{(S_i, P \circ \nu_i^{-1})\}_{i=1}^m$ as
        \begin{align*}
            \val(\sI; \pi) &= \frac{1}{m} \sum_{i=1}^m P(\nu_i^{-1}(\pi[S_i]))\\
            &= \frac{1}{m} \sum_{i=1}^m \left(C_P + \sum_{\substack{T \subseteq [k]:\\ 2\le |T| \le d}} \sum_{\sigma \in \Sym(T)} c_{T, \sigma} \cdot \Phi_{T, \sigma}(\nu_i^{-1}(\pi[S_i])[T])\right)\\
            &= C_P + \sum_{\substack{T \subseteq [k]:\\ 2\le |T| \le d}} \sum_{\sigma \in \Sym(T)} c_{T, \sigma} \left(\frac{1}{m}\sum_{i=1}^m\Phi_{T, \sigma}(\nu_i^{-1}(\pi[S_i])[T])\right)\\
            &= C_P + \sum_{\substack{T \subseteq [k]:\\ 2\le |T| \le d}} \sum_{\sigma \in \Sym(T)} c_{T, \sigma} \left(\frac{1}{m}\sum_{i=1}^m\Phi_{[|T|], \Id}(\sigma^{-1} \nu_i^{-1}\pi[\nu_i(T)])\right)\\
            &= C_P + \sum_{\substack{T \subseteq [k]:\\ 2\le |T| \le d}} \sum_{\sigma \in \Sym(T)} c_{T, \sigma} \cdot \val(\sI_{T,\sigma}; \pi),
        \end{align*}
        where $\sI_{T,\sigma}$ is the $t$-OCSP instance with predicate $\Phi_{[t], \Id}$ with $t \colonequals |T|$, obtained by adding a clause $(\nu_i(T), \Phi_{[t], \Id} \circ \sigma^{-1} \nu_i^{-1})$ for every clause $(S_i, P \circ \nu_i^{-1})$ of $\sI$. Note that $\sI_{T, \sigma}$ follows the distribution of $(k, p)$-random $t$-OCSP with predicate $\Phi_{[t], \Id}$, because
        \begin{itemize}
            \item every $S \in \binom{[n]}{k}$ is sampled independently with probability $p$, and $\nu \in \Sym(S)$ is sampled uniformly at random for every $S$,
            \item $\nu \circ \sigma \in \sO_t(S)$ is a uniformly random ordered list of distinct elements of $S$ of length $t$ as $\nu \sim \Unif(\Sym(S))$, and $\nu(T) = \nu(\sigma(T))$ as $\sigma \in \Sym(T)$.
        \end{itemize}
        As a result, a two-sided $\eps$-refutation algorithm for a $(k,p)$-random $t$-OCSP instance with the canonical ordering predicate $\Phi_{[t], \Id}$ on $m$ clauses with high probability certifies that
        \begin{align*}
            \max_{\pi \in \Sym([n])}\left|\val(\sI_{T, \sigma}; \pi) - \frac{1}{t!}\right| \le \eps.
        \end{align*}

    Taking a union bound over the failure probabilities over the two-sided $\eps$-refutation algorithm applied to $\sI_{T, \sigma}$ for all $T \subseteq [k]: 2 \le |T| \le d$ and all $\sigma \in \Sym(T)$, with high probability we get a certificate that
    \begin{align*}
            \max_{\pi \in \Sym([n])}\left|\val(\sI_{T, \sigma}; \pi) - \frac{1}{|T|!}\right| \le \eps, \quad \forall\, T \subseteq [k]: 2 \le |T| \le d, \forall\,\sigma \in \Sym(T), 
        \end{align*}
    and this certificate is computed in time $O_k(m + \max_{2\le t \le d} F_t)$, by constructing the OCSP instances $\sI_{T, \sigma}$ and applying the two-sided strong refutation algorithms to them.

    Finally, using the above certificate and \eqref{eq:avg-decom}, we get
    \begin{align*}
        &\val(\sI) - \avg(P)\\
        &= \max_{\pi \in \Sym([n])} \left(C_P + \sum_{\substack{T \subseteq [k]:\\ 2\le |T| \le d}} \sum_{\sigma \in \Sym(T)} c_{T, \sigma} \cdot \val(\sI_{T,\sigma}; \pi)\right) - \left(C_P + \sum_{\substack{T \subseteq [k]:\\ 2\le |T| \le d}} \sum_{\sigma \in \Sym(T)} c_{T, \sigma} \cdot \frac{1}{|T|!}\right)\\
        &\le \sum_{\substack{T \subseteq [k]:\\ 2\le |T| \le d}} \sum_{\sigma \in \Sym(T)} |c_{T, \sigma}| \max_{\pi \in \Sym([n])} \left|\val(\sI_{T, \sigma}; \pi) - \frac{1}{|T|!}\right|\\
        &\le 2^k \cdot k! \cdot 2^k \cdot \eps\\
        &= (2^{2k}k!)\eps,
    \end{align*}
    which gives the desired $(2^{2k}k!)\eps$-refutation algorithm with the specified running time for $\sI$.

    \end{proof}

\subsection{Kikuchi Method}

In this section, we prove Theorem~\ref{thm:kikuchi}, which provides strong refutation for OCSP using the Kikuchi method with a smooth tradeoff between running time, clause density, and the refutation strength.

We first give an overview of Kikuchi method. In the setting of random $k$-XOR, Kikuchi method builds a Kikuchi matrix from the labeled hypergraph underlying the $k$-XOR instance, and use the spectral norm of the Kikuchi matrix to certify that no variable assignments can satisfy an unusually large fraction of clauses. More precisely, suppose $k$ is even and the $k$-XOR instance over the domain $\{\pm 1\}^n$ have $m$ clauses taking the form $\prod_{t \in S_i} x_t = b_i$ where $b_i \in \{\pm 1\}$, the level-$\ell$ Kikuchi matrix $M$ is a matrix indexed by $\ell$-subsets of $[n]$, whose $(A,B)$-entry takes value $\sum_{i \in [m]} b_i\cdot \one\{A \triangle B = S_i\}$. Then, bounding the spectral norm of $M$ certifies an upper bound on the fraction of satisfiable clauses, since 
\begin{align*}
    &\max_{x \in \{\pm 1\}^n } \frac{1}{m} \sum_{i=1}^m\left(\one\left\{\prod_{t \in S_i} x_t = b_i\right\} - \frac{1}{2}\right)\\
    &= \max_{x \in \{\pm 1\}^n } \frac{1}{m} \sum_{i=1}^m\frac{b_i \prod_{t \in S_i} x_t}{2}\\
    &= \max_{x \in \{\pm 1\}^n } \frac{1}{2m} \sum_{i=1}^m \frac{1}{\binom{k}{k/2}\binom{n-k}{\ell-k/2}} b_i\sum_{ \substack{A, B \in \binom{[n]}{\ell}:\\ A \triangle B = S_i } } \left(\prod_{s \in A} x_s\right)\left(\prod_{t \in B} x_t\right)\\
    &= \max_{x \in \{\pm 1\}^n } \frac{1}{2\binom{k}{k/2}\binom{n-k}{\ell-k/2}m}(x^{\otimes \ell})^\top M x^{\otimes \ell},
\end{align*}
where $x^{\otimes \ell} \in \RR^{\binom{n}{\ell} }$ represents the vector indexed by $\ell$-subsets of $[n]$ whose coordinates are $x^A = \prod_{t \in A} x_t$. Note that the norm of $x^{\otimes \ell}$ satisfies $\|x^{\otimes \ell}\|^2 = \binom{n}{\ell}$, and thus certifying a good upper bound on the spectral norm of $M$ certifies an upper bound on the fraction of satisfiable clauses. For more details, see \cite[Section 6]{wein2019kikuchi}.

However, the vanilla Kikuchi method does not directly extend to the setting of OCSPs due to the non-product structure of $\Sym([n])$. We therefore introduce the following rank decomposition procedure.

\subsubsection{Rank Decomposition Procedure}\label{sec:rank-decom}
Fix $n \in \NN$. Let $L = \lceil \log_2(n+1) \rceil$. We introduce the parity encoding of elements of $[n]$ as follows.

\begin{definition}[Parity Encoding]
    For $x \in [n]$, we denote the binary representation of $x$ as
\begin{align*}
    b(x) = (b_1(x), \dots, b_L(x)) \in \{0,1\}^L,
\end{align*}
with $b_1(x)$ being the most significant bit and $b_L(x)$ being the least significant bit. We denote the parity encoding of $x$ as
\begin{align*}
    e(x) = (e_1(x), \dots, e_L(x)) \colonequals \left((-1)^{b_1(x)}, \dots, (-1)^{b_L(x)}\right) \in \{\pm 1\}^L.
\end{align*}
\end{definition}

We may rewrite the canonical ordering predicate in terms of the parity encoding using the following lemma.
\begin{lemma}\label{lem:rank-decom}
    For $x, y \in [n]$, the basic ordering indicator of $x < y$ can be written as
    \begin{align*}
        \one\{x < y\} = \sum_{\alpha, \beta \subseteq [L]} c_{\alpha, \beta} e(x)^{\alpha} e(y)^{\beta},
\end{align*}
where $e(x)^{\alpha} \colonequals \prod_{t \in \alpha} e_t(x)$ and
\begin{align*}
    c_{\alpha, \beta} &= \sum_{r=1}^L \frac{1}{2^{r+1}} \one\{\alpha \triangle \beta \subseteq \{r\}, \alpha \cup \beta \subseteq [r] \} (-1)^{\one\{r \in \beta\}}.
\end{align*}

Moreover, for $S \in \binom{[n]}{t}$ and $\mu = (\mu_1, \dots, \mu_t) \in \Sym(S)$, the canonical ordering predicate $\Phi_{S, \mu}$ can be expressed as
\begin{align*}
    \Phi_{S,\mu}(\pi[S]) &= \one\{\pi[S] = \mu\}\\
    &= \one\{\pi^{-1}(\mu_1) < \pi^{-1}(\mu_2) < \dots < \pi^{-1}(\mu_t)\}\\
    &= \sum_{\Gamma \in \Theta} C_\Gamma \cdot r(\pi, \mu)^\Gamma,
\end{align*}
where we denote
\begin{align*}
    C_\Gamma &\colonequals \sum_{\substack{\{(\alpha_i, \beta_i)\}_{i=1}^{t-1}:\\ \alpha_i, \beta_i \subseteq [L]}} \left(\prod_{i=1}^{t-1} c_{\alpha_i, \beta_i} \right) \one\{\alpha_1 = \gamma_1\} \one\{\beta_{t-1} = \gamma_t\} \left(\prod_{i=2}^{t-1} \one\{\alpha_i \triangle \beta_{i-1} = \gamma_i\}\right),\\
    r(\pi, \mu) &= (r(\pi, \mu_1), r(\pi, \mu_2), \dots, r(\pi, \mu_t))\\
    &\colonequals (e(\pi^{-1}(\mu_1)), e(\pi^{-1}(\mu_2)), \dots, e(\pi^{-1}(\mu_t))) \in \left(\{\pm 1\}^{L}\right)^t,\\
    r(\pi, \mu)^\Gamma &\colonequals \prod_{i=1}^t r(\pi, \mu_i)^{\gamma_i},
\end{align*}
and let
\begin{align*}
    \Theta \colonequals \{\Gamma = \{\gamma_i\}_{i=1}^t: \gamma_i \subseteq [L],  \forall\, i \in [t], \text{ and } C_\Gamma \ne 0\}
\end{align*}
be the set of all $\Gamma$ for which $C_\Gamma$ is nonzero.
\end{lemma}

\begin{proof}
Note
\begin{align*}
    \one\{b_i(x) = b_i(y)\} &= \frac{1 + e_i(x)e_i(y)}{2},\\
    \one\{b_i(x) < b_i(y)\} &= \frac{(1 + e_i(x))(1 - e_i(y))}{4}.
\end{align*}
As a result, we may express the basic ordering indicator as
\begin{align*}
    \one\{x < y\} &= \sum_{r=1}^L \prod_{s < r} \one\{b_s(x) = b_s(y)\} \one\{b_r(x) < b_r(y)\}\\
    &= \sum_{r=1}^L \frac{1}{2^{r+1}}(1 + e_r(x))(1 - e_r(y))\prod_{s < r} (1 + e_s(x)e_s(y))\\
    &= \sum_{\alpha, \beta \subseteq [L]} c_{\alpha, \beta} e(x)^{\alpha} e(y)^{\beta}, \numberthis \label{eq:order-decom}
\end{align*}
where $e(x)^{\alpha} \colonequals \prod_{t \in \alpha} e_t(x)$ and the coefficients take the form 
\begin{align}
    c_{\alpha, \beta} &= \sum_{r=1}^L \frac{1}{2^{r+1}} \one\{\alpha \triangle \beta \subseteq \{r\}, \alpha \cup \beta \subseteq [r] \} (-1)^{\one\{r \in \beta\}}. \label{eq:coeff}
\end{align}
Using \eqref{eq:order-decom} and \eqref{eq:coeff}, for $S \subseteq [n]$ of size $|S| = t$ and $\mu \in \Sym(S)$, we may write the canonical ordering predicate $\Phi_{S, \mu}$ as
\begin{align*}
    \Phi_{S,\mu}(\pi[S]) &= \one\{\pi[S] = \mu\}\\
    &= \one\{\pi^{-1}(\mu(1)) < \pi^{-1}(\mu(2)) < \dots < \pi^{-1}(\mu(|S|))\}\\
    &= \prod_{i=1}^{t-1} \one\{\pi^{-1}(\mu(i)) < \pi^{-1}(\mu(i+1))\}\\
    &= \prod_{i=1}^{t-1} \sum_{\alpha_i, \beta_i \subseteq [L]} c_{\alpha_i, \beta_i} e(\pi^{-1}(\mu(i)))^{\alpha_i} e(\pi^{-1}(\mu(i+1)))^{\beta_i}\\
    &= \sum_{\substack{\{(\alpha_i, \beta_i)\}_{i=1}^{t-1}:\\ \alpha_i, \beta_i \subseteq [L]}} \left(\prod_{i=1}^{t-1} c_{\alpha_i, \beta_i} \right) e(\pi^{-1}(\mu(1)))^{\alpha_1} e(\pi^{-1}(\mu(t)))^{\beta_{t-1}} \prod_{i=2}^{t-1} e(\pi^{-1}(\mu(i)))^{\alpha_i \triangle \beta_{i-1}}, \numberthis \label{eq:ordering-predicate-decom}
\end{align*}
where the last step follows because $e(x)^A e(x)^B = e(x)^{A \triangle B} \left(e(x)^{A \cap B}\right)^2 = e(x)^{A \triangle B}$, as $e(x) \in \{\pm 1\}^L$. Recall as in the statement of Lemma~\ref{lem:rank-decom} that, given any $S \subseteq [n]$ with $|S| = t$, $\mu \in \Sym(S)$, and $\Gamma = \{\gamma_i\}_{i=1}^{t}$ where $\gamma_i \subseteq [L]$ for every $1 \le i \le t$, we denote
\begin{align*}
    C_\Gamma &\colonequals \sum_{\substack{\{(\alpha_i, \beta_i)\}_{i=1}^{t-1}:\\ \alpha_i, \beta_i \subseteq [L]}} \left(\prod_{i=1}^{t-1} c_{\alpha_i, \beta_i} \right) \one\{\alpha_1 = \gamma_1\} \one\{\beta_{t-1} = \gamma_t\} \left(\prod_{i=2}^{t-1} \one\{\alpha_i \triangle \beta_{i-1} = \gamma_i\}\right),\\
    r(\pi, \mu) &= (r(\pi, \mu_1), r(\pi, \mu_2), \dots, r(\pi, \mu_t))\\
    &\colonequals (e(\pi^{-1}(\mu_1)), e(\pi^{-1}(\mu_2)), \dots, e(\pi^{-1}(\mu_t))) \in \left(\{\pm 1\}^{L}\right)^t,\\
    r(\pi, \mu)^\Gamma &\colonequals \prod_{i=1}^t r(\pi, \mu_i)^{\gamma_i},
\end{align*}
and let
\begin{align*}
    \Theta \colonequals \{\Gamma = \{\gamma_i\}_{i=1}^t: \gamma_i \subseteq [L],  \forall\, i \in [t], \text{ and } C_\Gamma \ne 0\}
\end{align*}
be the set of all $\Gamma$ for which $C_\Gamma$ is nonzero. 
Using these notations, \eqref{eq:ordering-predicate-decom} can be written as
\begin{align*}
    \Phi_{S,\mu}(\pi[S]) &= \one\{\pi^{-1}(\mu_1) < \pi^{-1}(\mu_2) < \dots < \pi^{-1}(\mu_{|S|})\} = \sum_{\Gamma \in \Theta} C_\Gamma \cdot r(\pi, \mu)^\Gamma.
\end{align*}
\end{proof}

We will need the following bound on the coefficients $C_\Gamma$ in Lemma~\ref{lem:rank-decom}.

\begin{lemma}\label{lem:coeff-bound}
    In the rank decomposition stated in Lemma~\ref{lem:rank-decom}, the coefficients satisfy
    \begin{align*}
        \sum_{\Gamma \in \Theta} |C_\Gamma| \le L^{t-1},
    \end{align*}
    where $L = \lceil \log_2(n+1)\rceil$.
\end{lemma}
\begin{proof}
    Recall the coefficients $C_\Gamma$ take the form
    \begin{align*}
        C_\Gamma &\colonequals \sum_{\substack{\{(\alpha_i, \beta_i)\}_{i=1}^{t-1}:\\ \alpha_i, \beta_i \subseteq [L]}} \left(\prod_{i=1}^{t-1} c_{\alpha_i, \beta_i} \right) \one\{\alpha_1 = \gamma_1\} \one\{\beta_{t-1} = \gamma_t\} \left(\prod_{i=2}^{t-1} \one\{\alpha_i \triangle \beta_{i-1} = \gamma_i\}\right),
    \end{align*}
    where $c_{\alpha,\beta}$ are given by \eqref{eq:coeff} as
    \begin{align*}
        c_{\alpha, \beta} &= \sum_{r=1}^L \frac{1}{2^{r+1}} \one\{\alpha \triangle \beta \subseteq \{r\}, \alpha \cup \beta \subseteq [r] \} (-1)^{\one\{r \in \beta\}}. 
    \end{align*}
    Then,
    \begin{align*}
        \sum_{\Gamma \in \Theta} |C_\Gamma| &= \sum_{\Gamma \in \Theta} \left|\sum_{\substack{\{(\alpha_i, \beta_i)\}_{i=1}^{t-1}:\\ \alpha_i, \beta_i \subseteq [L]}} \left(\prod_{i=1}^{t-1} c_{\alpha_i, \beta_i} \right) \one\{\alpha_1 = \gamma_1\} \one\{\beta_{t-1} = \gamma_t\} \left(\prod_{i=2}^{t-1} \one\{\alpha_i \triangle \beta_{i-1} = \gamma_i\}\right)\right|\\
        &\le \sum_{\Gamma \in \Theta} \sum_{\substack{\{(\alpha_i, \beta_i)\}_{i=1}^{t-1}:\\ \alpha_i, \beta_i \subseteq [L]}} \left(\prod_{i=1}^{t-1} \left|c_{\alpha_i, \beta_i}\right| \right) \one\{\alpha_1 = \gamma_1\} \one\{\beta_{t-1} = \gamma_t\} \left(\prod_{i=2}^{t-1} \one\{\alpha_i \triangle \beta_{i-1} = \gamma_i\}\right)\\
        &= \sum_{\substack{\{(\alpha_i, \beta_i)\}_{i=1}^{t-1}:\\ \alpha_i, \beta_i \subseteq [L]}} \left(\prod_{i=1}^{t-1} \left|c_{\alpha_i, \beta_i}\right| \right) \sum_{\Gamma \in \Theta}\one\{\alpha_1 = \gamma_1\} \one\{\beta_{t-1} = \gamma_t\} \left(\prod_{i=2}^{t-1} \one\{\alpha_i \triangle \beta_{i-1} = \gamma_i\}\right)\\
        &\le \sum_{\substack{\{(\alpha_i, \beta_i)\}_{i=1}^{t-1}:\\ \alpha_i, \beta_i \subseteq [L]}} \left(\prod_{i=1}^{t-1} \left|c_{\alpha_i, \beta_i}\right| \right)\\
        &= \left(\sum_{\alpha, \beta \subseteq [L]}\left|c_{\alpha, \beta}\right|\right)^{t-1}, \numberthis \label{ineq:coeff-bound}
    \end{align*}
    where the inequality in the second-to-last line follows because for every $\{(\alpha_i, \beta_i)\}_{i=1}^{t-1}$, there is at most one $\Gamma \in \Theta$ such that the indicators for $\alpha_1 = \gamma_1, \beta_{t-1} = \gamma_t$, and $\alpha_i \triangle \beta_{i-1} = \gamma_i$ for every $2 \le i \le t-1$ hold.

    Summing over $\alpha, \beta \subseteq [L]$ and using the expression \eqref{eq:coeff}, we have
    \begin{align*}
        \sum_{\alpha, \beta \subseteq [L]}\left|c_{\alpha, \beta}\right| &= \sum_{\alpha, \beta \subseteq [L]}\left|\sum_{r=1}^L \frac{1}{2^{r+1}} \one\{\alpha \triangle \beta \subseteq \{r\}, \alpha \cup \beta \subseteq [r] \} (-1)^{\one\{r \in \beta\}}\right|\\
        &\le \sum_{\alpha, \beta \subseteq [L]}\sum_{r=1}^L \frac{1}{2^{r+1}} \one\{\alpha \triangle \beta \subseteq \{r\}, \alpha \cup \beta \subseteq [r] \}\\
        &= \sum_{r=1}^L \frac{1}{2^{r+1}} \sum_{\alpha, \beta \subseteq [r]} \one\{\alpha \triangle \beta \subseteq \{r\} \}\\
        &= \sum_{r=1}^L \frac{1}{2^{r+1}} \sum_{\alpha, \beta \subseteq [r]} \one\{\beta = \alpha, \text{ or } \beta = \alpha \triangle \{r\} \}\\
        &= \sum_{r=1}^L \frac{1}{2^{r+1}} \sum_{\alpha\subseteq [r]} 2\\
        &= \sum_{r=1}^L 1\\
        &= L,
    \end{align*}
    and from \eqref{ineq:coeff-bound} we conclude
    \begin{align*}
        \sum_{\Gamma \in \Theta} |C_\Gamma| \le \left(\sum_{\alpha, \beta \subseteq [L]}\left|c_{\alpha, \beta}\right|\right)^{t-1} \le L^{t-1}.
    \end{align*}
\end{proof}

\subsubsection{Certifying Tensor Evaluation}

\begin{definition}[Centered Occurrence Count Tensor]\label{def:F}
    Let $\sI = \{(T_i, \Phi_{[t],\Id} \circ \nu_i^{-1})\}_{i=1}^m$ be a $t$-OCSP instance with the canonical ordering predicate $\Phi_{[t],\Id}$ on $m$ clauses. Define an asymmetric tensor $F \in \RR^{[n]^t}$, the centered occurence count tensor, whose entries are given by
\begin{align}
    F(v_1, v_2, \dots, v_t) &= \begin{cases}
        \left(\sum_{i=1}^m \one\{\nu_i(j) = v_j, \forall\, 1\le j\le t\}\right) - \frac{p \binom{n}{k}}{\binom{n}{t}t!} & \quad \text{ if } v_1, \dots, v_t \text{ are distinct},\\
        0 & \quad \text{ otherwise}.
    \end{cases} \label{eq:F-def}
\end{align}
\end{definition}

We may derive the tensor $F$ in the following alternative way. Recall that according to Definition~\ref{def:kp-OCSP}, an equivalent way to draw $(k,p)$-random $t$-OCSP is as follows:
\begin{itemize}
    \item First draw an Erd\H{o}s-R\'{e}nyi $k$-uniform hypergraph $H$ where each possible $S \in \binom{[n]}{k}$ is included in $E(H)$ independently with probability $p$.
    \item For every $S \in \binom{[n]}{k}$, draw a local ordering $\mu_S \in \Sym(S)$ independently and uniformly at random.
    \item Then, for every $(v_1, \dots, v_t) \in \sO_t([n])$, set
    \begin{align*}
        F(v_1, \dots, v_t) &= \sum_{S \in \binom{[n]}{k}} \sum_{\substack{\mu \in \Sym(S):\\ \mu(j) = v_j, \forall\, 1\le j \le t}} \left(\one\{S \in E(H), \mu = \mu_S\} - \frac{p}{k!}\right)\\
        &=\sum_{\substack{S \in \binom{[n]}{k}:\\ \{v_1, \dots, v_t\} \subseteq S} } \left(\one\{S \in E(H), \mu_S(j) = v_j, \forall\, 1\le j \le t\} - \frac{p(k-t)!}{k!}\right)\\
        &= \left(\sum_{i=1}^m \one\{\nu_i(j) = v_j, \forall\, 1\le j \le t\}\right) - \binom{n-t}{k-t}\cdot\frac{p(k-t)!}{k!}\\
        &= \left(\sum_{i=1}^m \one\{\nu_i(j) = v_j, \forall\, 1\le j \le t\}\right) - \frac{p\binom{n}{k}}{\binom{n}{t}t!},
    \end{align*}
    which recovers Definition~\ref{def:F} above. Since $\one\{S \in E(H), \mu = \mu_S\} - \frac{p}{k!}$ has mean zero, $F$\footnote{This tensor is also defined at the beginning of \cite[Section 5.3]{chan2026strongly}, but we note that there is a small discrepancy between $p$ in their formula for $G_R$ and the probability $p$ used in their Definition 3.2. One needs to replace $p$ in $G_R$ with $\hat{p} \colonequals \frac{p}{k!}$ in order for $G_R$ to be centered random variables. } can be regarded as the centered occurrence count tensor when $\sI$ is a $(k,p)$-random $t$-OCSP.
\end{itemize}

We may upper bound the deviation of OCSP value from $\frac{1}{t!}$ using the following lemma.

\begin{lemma}\label{lem:val-deviation}
    The following bound holds simultaneously for all $\pi \in \Sym([n])$:
    \begin{align*}
    \left|m\cdot \val(\sI; \pi) - \frac{\overline{m}}{t!}\right| &\le L^{t-1} \max_{y^{(1)}, y^{(2)}, \dots, y^{(t)} \in \{\pm 1\}^n} \left|\langle F,  y^{(1)} \otimes y^{(2)} \otimes \dots \otimes y^{(t)}\rangle\right|,
\end{align*}
where $L = \lceil \log_2(n+1)\rceil$ and $\overline{m} = \binom{n}{k}p$.
\end{lemma}

\begin{proof}

We may express
\begin{align*}
    &m\cdot \val(\sI; \pi) - \frac{\overline{m}}{t!}\\
    &= \left(\sum_{i=1}^m \Phi_{[t], \Id} (\nu_i^{-1}(\pi[S_i]))\right) - \frac{\overline{m}}{t!}\\
    &= \left(\sum_{i=1}^m \one\{\nu_i^{-1}(\pi[S_i]) = \Id\}\right) - \frac{\overline{m}}{t!}\\
    &= \left(\sum_{i=1}^m \one\{\pi[S_i] = \nu_i\}\right) - \frac{\binom{n}{k}p}{t!}\\
    &= \left(\sum_{v = (v_1, \dots, v_t) \in \sO_t([n])} \left(\sum_{i=1}^m \one\{\nu_i(j) = v_j, \forall\, 1 \le j \le t\}\right) \one\{\pi^{-1}(v_1) < \pi^{-1}(v_2) < \dots < \pi^{-1}(v_t)\} \right) - \frac{\binom{n}{k}p}{t!}\\
    &= \sum_{v = (v_1, \dots, v_t) \in \sO_t([n]) } \left(\left(\sum_{i=1}^m \one\{\nu_i(j) = v_j, \forall\, 1 \le j \le t\}\right) - \frac{p \binom{n}{k}}{\binom{n}{t}t!}\right) \one\{\pi^{-1}(v_1) < \dots < \pi^{-1}(v_t)\}\\
    &= \sum_{v = (v_1, \dots, v_t) \in \sO_t([n]) } F(v_1, \dots, v_t)\one\{\pi^{-1}(v_1) <  \dots < \pi^{-1}(v_t)\},
\end{align*}
where the second-to-last step holds because for every set $S \subseteq [n]$ with $|S| = t$ and every $\pi \in \Sym([n])$, there is exactly one way to order the elements of $S$ as $(v_1, \dots, v_t) \in \Sym(S) \subseteq \sO_t([n])$ such that $\pi^{-1}(v_1) < \pi^{-1}(v_2) < \dots < \pi^{-1}(v_t)$, and thus the total amount subtracted is equal to $\binom{n}{t} \cdot \frac{p \binom{n}{k}}{\binom{n}{t}t!} = \frac{\binom{n}{k}p}{t!} = \frac{\overline{m}}{t!}$. Now we further use the rank decomposition formula in Lemma~\ref{lem:rank-decom} to get
\begin{align*}
    &m \cdot\val(\sI; \pi) - \frac{\overline{m}}{t!}\\
    &= \sum_{\substack{v_1, \dots, v_t \in [n]:\\ v_i \text{ distinct}} } F(v_1, \dots, v_t)\one\{\pi^{-1}(v_1) <\dots < \pi^{-1}(v_t)\}\\
    &= \sum_{S \subseteq [n]: |S| = t }\sum_{\nu \in \Sym(S) } F(\nu(1), \dots, \nu(t))\one\{\pi^{-1}(\nu(1)) <  \dots < \pi^{-1}(\nu(t))\}\\
    &= \sum_{S \subseteq [n]: |S| = t }\sum_{\nu \in \Sym(S) } F(\nu(1), \dots, \nu(t)) \sum_{\Gamma \in \Theta} C_{\Gamma} r(\pi, \nu)^\Gamma  && \quad \text{(Lemma~\ref{lem:rank-decom})}\\
    &= \sum_{\Gamma \in \Theta } C_{\Gamma}  \langle F, r_\pi^\Gamma \rangle,
\end{align*}
where $r_\pi^\Gamma \in \RR^{[n]^t}$ denotes the tensor with entries
\begin{align*}
    r_\pi^\Gamma(v) = r(\pi, v)^\Gamma, \quad \forall\, v \in [n]^t,
\end{align*}
where for convenience, we overload the notation and denote for $v = (v_1, v_2, \dots, v_t) \in [n]^t$ with not necessarily distinct coordinates that
\begin{align*}
    r(\pi, v) &= (r(\pi, v_1), r(\pi, v_2), \dots, r(\pi, v_t))\\
    &\colonequals (e(\pi^{-1}(v_1)), e(\pi^{-1}(v_2)), \dots, e(\pi^{-1}(v_t))) \in \left(\{\pm 1\}^{L}\right)^t,\\
    r(\pi, v)^\Gamma &\colonequals \prod_{i=1}^t r(\pi, v_i)^{\gamma_i}
\end{align*}

Thus, we may bound
\begin{align}
    \left|m \cdot\val(\sI; \pi) - \frac{\overline{m}}{t!}\right| \le \sum_{\Gamma \in \Theta } |C_{\Gamma}| \cdot\left|\langle F, r_\pi^\Gamma \rangle\right|,\label{ineq:tensor-bound}
\end{align}
and it remains to bound $\left|\langle F, r_\pi^\Gamma \rangle\right|$ for every $\pi \in \Sym([n])$ and every $\Gamma \in \Theta$.

Recall that we denote $r(\pi, j) = e(\pi^{-1}(j))$ for $\pi \in \Sym([n])$ and $j \in [n]$. For $\Gamma \in \Theta$, $\pi \in \Sym([n])$, and $i \in [t]$, define a vector $x^{(\pi, \Gamma, i)} \in \{\pm 1\}^n$ whose entries are given by
\begin{align*}
    x^{(\pi, \Gamma, i)}_j \colonequals r(\pi, j)^{\gamma_i} = e(\pi^{-1}(j))^{\gamma_i}, \quad \forall\, j\in [n].
\end{align*}
Then, we can rewrite the entries of $r_\pi^\Gamma$ as
\begin{align*}
    r_\pi^\Gamma(v) = r(\pi, v)^\Gamma = \prod_{i=1}^t r(\pi, v_i)^{\gamma_i}= \prod_{i=1}^t x_{v_i}^{(\pi, \Gamma, i)}, \quad \forall\, v \in [n]^t,
\end{align*}
and express $r_\pi^\Gamma$ as a rank-one tensor product
\begin{align*}
    r_\pi^\Gamma = x^{(\pi, \Gamma, 1)} \otimes x^{(\pi, \Gamma, 2)} \otimes \dots \otimes x^{(\pi, \Gamma, t)}.
\end{align*}

Since $x^{(\pi, \Gamma, i)} \in \{\pm 1\}^n$ for any $\Gamma \in \Theta, \pi \in \Sym([n])$, and $i \in [t]$, we may upper bound
\begin{align*}
    \left|\langle F, r_\pi^\Gamma \rangle\right| &= \left|\langle F, x^{(\pi, \Gamma, 1)} \otimes x^{(\pi, \Gamma, 2)} \otimes \dots \otimes x^{(\pi, \Gamma, t)} \rangle\right|\\
    &\le \max_{y^{(1)}, y^{(2)}, \dots, y^{(t)} \in \{\pm 1\}^n} \left|\langle F,  y^{(1)} \otimes y^{(2)} \otimes \dots \otimes y^{(t)}\rangle\right| \numberthis \label{ineq:F-bd}
\end{align*}
simultaneously for all $\Gamma \in \Theta$ and $\pi \in \Sym([n])$. Plugging \eqref{ineq:F-bd} into \eqref{ineq:tensor-bound} and using the coefficient bound from Lemma~\ref{lem:coeff-bound}, we get
\begin{align*}
    \left|m \cdot \val(\sI; \pi) - \frac{\overline{m}}{t!}\right| &\le \sum_{\Gamma \in \Theta } |C_{\Gamma}|\left|\langle F, r_\pi^\Gamma \rangle\right|\\
    &\le L^{t-1} \max_{y^{(1)}, y^{(2)}, \dots, y^{(t)} \in \{\pm 1\}^n} \left|\langle F,  y^{(1)} \otimes y^{(2)} \otimes \dots \otimes y^{(t)}\rangle\right|
\end{align*}
as desired, where $L = \lceil \log_2(n+1)\rceil$.
\end{proof}

\subsubsection{Kikuchi Matrix}

Note that Lemma~\ref{lem:val-deviation} reduces the problem of two-sided strong refutation of OCSP with canonical ordering predicate to certifying the value of the asymmetric tensor $F$ on rank-one tensor products of $\{\pm 1\}$-valued vectors, and \cite{chan2026strongly} provided tools for this tensor certification task using the Kikuchi method. To do so, let us define the Kikuchi matrix we need, and use the spectral norm bound on the Kikuchi matrix to bound 
\begin{align*}
    \max_{y^{(1)}, y^{(2)}, \dots, y^{(t)} \in \{\pm 1\}^n} \left|\langle F,  y^{(1)} \otimes y^{(2)} \otimes \dots \otimes y^{(t)}\rangle\right|.
\end{align*}

As usual in application of Kikuchi methods, we need to treat the case of odd-order $F$ differently from the case of even-order $F$. Let us start with the simpler case of $F \in \RR^{[n]^t}$ with even $t$.

\begin{definition}[Kikuchi Matrix for Even-Order Tensor]\label{def:kikuchi-even}
    Fix $t \in \NN$ even. Let $F \in \RR^{[n]^t}$ be an asymmetric tensor, and $1\le \ell \le n$ be a parameter. Let $\sJ = \sJ^{t,\ell}$ denote
    \begin{align*}
        \sJ \colonequals \left\{(A_1, A_2, \dots, A_t): A_i \subseteq [n], \forall\, 1\le i \le t, \text{ and } \sum_{i=1}^t |A_i| = \ell\right\}.
    \end{align*}
    
    We define the level-$\ell$ Kikuchi matrix $M = M^{(t,\ell)}(F) \in \RR^{\sJ \times \sJ}$ as follows. For \[\mathcal{A} = (A_1, \dots, A_t), \mathcal{B} = (B_1, \dots, B_t) \in \sJ,\] the $(\mathcal{A}, \mathcal{B})$-entry of $M$ is
    \begin{align*}
        M(\mathcal{A}, \mathcal{B}) &= \begin{cases}
            F(v_1, \dots, v_t) & \quad \text{ if } A_i \triangle B_i = \{v_i\}, \forall\, 1\le i \le t,\\
            0 & \quad \text{ otherwise}.
        \end{cases}
    \end{align*}
\end{definition}
Then, we may certify the evaluation of an even-order tensor $F \in \RR^{[n]^t}$ on tensor product of vectors in $\{\pm 1\}^n$ using the spectral norm of Kikuchi matrices. 
\begin{lemma}[Spectral Certification for Even-Order Tensor with Kikuchi Matrix]\label{lem:kikuchi-certification-even}
    Fix $t \in \NN$ even. Let $F \in \RR^{[n]^t}$ be an asymmetric tensor, and $\frac{t}{2}\le \ell \le n$ be a parameter.

    Let $M = M^{t, \ell}(F)$ be the level-$\ell$ Kikuchi matrix of $F$. Then,
    \begin{align*}
        \max_{y^{(1)}, y^{(2)}, \dots, y^{(t)} \in \{\pm 1\}^n} \left|\langle F,  y^{(1)} \otimes y^{(2)} \otimes \dots \otimes y^{(t)}\rangle\right| \le \frac{\binom{tn}{\ell}}{\binom{tn-t}{\ell - \frac{t}{2}}\binom{t}{\frac{t}{2}}} \|M\|_2.
    \end{align*}
\end{lemma}
\begin{proof}
    For any fixed $v_1, \dots, v_t \in [n]$, observe that there are exactly $\binom{tn-t}{\ell - \frac{t}{2}}\binom{t}{\frac{t}{2}}$ pairs of $\mathcal{A} = (A_1, \dots, A_t), \mathcal{B} = (B_1, \dots, B_t) \in \sJ$ such that $A_i \triangle B_i = \{v_i\}$ for all $1\le i \le t$. This is because
    \begin{itemize}
        \item $\sum_i |A_i| = \ell, \sum_i |B_i| = \ell$,
        \item and $A_i \triangle B_i = \{v_i\}$ for all $1\le i \le t$
    \end{itemize}
    force that for exactly half of $i \in [t]$, $A_i = B_i \sqcup \{v_i\}$, and for the other half of $i \in [t]$, $B_i = A_i \sqcup \{v_i\}$. Thus, there are $\binom{tn-t}{\ell - \frac{t}{2}}$ ways to choose $(A_1 \cap B_1, \dots, A_t \cap B_t)$, and $\binom{t}{\frac{t}{2}}$ ways to distribute $v_1, \dots, v_t$ evenly between $\mathcal{A}$ and $\mathcal{B}$. We denote
    \begin{align*}
        d_{\ell} \colonequals \binom{tn-t}{\ell - \frac{t}{2}}\binom{t}{\frac{t}{2}}.
    \end{align*}
    
    Consequently, for any $y^{(1)}, y^{(2)}, \dots, y^{(t)} \in \{\pm 1\}^n$, we have
    \begin{align*}
        &\left|\langle F, y^{(1)} \otimes y^{(2)}\otimes \dots\otimes y^{(t)} \rangle\right|\\
        &= \left|\sum_{v_1, \dots, v_t \in [n]} F(v_1, \dots, v_t) \prod_{i=1}^t y^{(i)}_{v_i}\right|\\
        &= \frac{1}{d_{\ell}} \left|\sum_{v_1, \dots, v_t \in [n]} \sum_{\substack{\mathcal{A} = (A_1, \dots, A_t),\\ \mathcal{B} = (B_1, \dots, B_t) \in \sJ:\\ A_i \triangle B_i = \{v_i\}, \forall\, 1\le i \le t }} F(v_1, \dots, v_t) \prod_{i=1}^t y^{(i)}_{v_i}\right|\\
        &= \frac{1}{d_{\ell}} \left|\sum_{v_1, \dots, v_t \in [n]} \sum_{\substack{\mathcal{A} = (A_1, \dots, A_t),\\ \mathcal{B} = (B_1, \dots, B_t) \in \sJ:\\ A_i \triangle B_i = \{v_i\}, \forall\, 1\le i \le t }} M(\mathcal{A}, \mathcal{B}) \prod_{i=1}^t \prod_{j \in A_i \triangle B_i}y^{(i)}_{j}\right| && \quad (A_i \triangle B_i = \{v_i\})\\
        &= \frac{1}{d_{\ell}} \left|\sum_{v_1, \dots, v_t \in [n]} \sum_{\substack{\mathcal{A} = (A_1, \dots, A_t),\\ \mathcal{B} = (B_1, \dots, B_t) \in \sJ:\\ A_i \triangle B_i = \{v_i\}, \forall\, 1\le i \le t }} M(\mathcal{A}, \mathcal{B}) \prod_{i=1}^t \left(\prod_{j \in A_i }y^{(i)}_{j}\right)\left(\prod_{j \in B_i }y^{(i)}_{j}\right)\right| && \quad (y^{(i)} \in \{\pm 1\}^n)\\
        &= \frac{1}{d_{\ell}} \left| \sum_{\substack{\mathcal{A} = (A_1, \dots, A_t),\\ \mathcal{B} = (B_1, \dots, B_t) \in \sJ}} M(\mathcal{A}, \mathcal{B}) \prod_{i=1}^t \left(\prod_{j \in A_i }y^{(i)}_{j}\right)\left(\prod_{j \in B_i }y^{(i)}_{j}\right)\right|\\
        &= \frac{1}{d_{\ell}} \left| \sum_{\substack{\mathcal{A} = (A_1, \dots, A_t),\\ \mathcal{B} = (B_1, \dots, B_t) \in \sJ}} M(\mathcal{A}, \mathcal{B}) x_{\mathcal{A}} x_{\mathcal{B}}\right|,
    \end{align*}
    where we use $x \in \{\pm 1\}^{\sJ}$ to denote the vector with entries
    \begin{align*}
        x_{\mathcal{A}} = \prod_{i=1}^t \prod_{j \in A_i} y_j^{(i)}
    \end{align*}
    for $\mathcal{A} = (A_1, \dots, A_t) \in \sJ$. Since $x \in \{\pm 1\}^{\sJ}$ has a fixed norm $\|x\|_2 = \sqrt{|\sJ|} = \sqrt{\binom{tn}{\ell}}$, we conclude
    \begin{align*}
        \frac{1}{d_{\ell}} \left| \sum_{\substack{\mathcal{A} = (A_1, \dots, A_t),\\ \mathcal{B} = (B_1, \dots, B_t) \in \sJ}} M(\mathcal{A}, \mathcal{B}) x_{\mathcal{A}} x_{\mathcal{B}}\right| = \frac{1}{\binom{tn-t}{\ell - \frac{t}{2}}\binom{t}{\frac{t}{2}}} \left|x^\top M x\right| \le \frac{1}{\binom{tn-t}{\ell - \frac{t}{2}}\binom{t}{\frac{t}{2}}} \|M\|_2 \|x\|_2^2 = \frac{\binom{tn}{\ell}}{\binom{tn-t}{\ell - \frac{t}{2}}\binom{t}{\frac{t}{2}}} \|M\|_2.
    \end{align*}
\end{proof}

Next, we consider the case of odd-order tensor.

\begin{definition}[Kikuchi Matrix for Odd-Order Tensor]\label{def:kikuchi-odd}
    Fix $t \in \NN$ odd. Let $F \in \RR^{[n]^t}$ be an asymmetric tensor, and $1\le \ell \le n$ be a parameter. Let $\tilde{\sJ} = \tilde{\sJ}^{t,\ell}$ denote
    \begin{align*}
        \tilde{\sJ} \colonequals \left\{((A_1^{(1)}, A_2^{(1)}, \dots, A_{t-1}^{(1)}), (A_1^{(2)}, A_2^{(2)}, \dots, A_{t-1}^{(2)})): A_i \subseteq [n], \forall\, 1\le i \le t-1, \text{ and } \sum_{c\in \{1,2\}}\sum_{i=1}^{t-1} |A_i^{(c)}| = \ell\right\}.
    \end{align*}
    
    We define the level-$\ell$ Kikuchi matrix $M = M^{(t,\ell)}(F) \in \RR^{\tilde{\sJ} \times \tilde{\sJ}}$ as follows. For \[\mathcal{A} = ((A_1^{(1)}, \dots, A_{t-1}^{(1)}), (A_1^{(2)}, \dots, A_{t-1}^{(2)})), \mathcal{B} = ((B_1^{(1)}, \dots, B_{t-1}^{(1)}), B_2^{(2)}, \dots, B_{t-1}^{(2)})) \in \tilde{\sJ},\] the $(\mathcal{A}, \mathcal{B})$-entry of $M$ is
    \begin{align*}
        M(\mathcal{A}, \mathcal{B}) &= \begin{cases}
            \sum_{u \in [n]} F(v_1^{(1)}, \dots, v_{t-1}^{(1)}, u)F(v_1^{(2)}, \dots, v_{t-1}^{(2)}, u) & \quad \begin{aligned}
            \text{ if } \forall\, c \in \{1,2\}, \sum_{i=1}^{t-1}|A_i^{(c)}| = \sum_{i=1}^{t-1}|B_i^{(c)}|,\\
            A_i^{(c)} \triangle B_i^{(c)} = \{v_i^{(c)}\}, \forall\, 1\le i \le t-1,\\
            \text{and }(v_1^{(1)}, \dots, v_{t-1}^{(1)}) \ne (v_1^{(2)}, \dots, v_{t-1}^{(2)}), 
            \end{aligned}\\
            0 & \quad \text{ otherwise}.
        \end{cases}
    \end{align*}
\end{definition}

\begin{remark}
    We note that this is identical to the definition of Kikuchi matrix for odd-order tensor used in \cite{chan2026strongly}. The conditions $\sum_{i=1}^{t-1}|A_i^{(1)}| = \sum_{i=1}^{t-1}|B_i^{(1)}|$ and $A_i^{(1)} \triangle B_i^{(1)} = \{v_i^{(1)}\}, \forall\, 1\le i \le t-1$ is equivalent to that: among the $t-1$ $v_i^{(1)}$ for $1\le i \le t-1$, exactly half of them belong to $A_i^{(1)}$ and half of them belong to $B_i^{(1)}$.
\end{remark}

Similarly, we may certify the evaluation of an odd-order tensor $F \in \RR^{[n]^t}$ with Kikuchi matrices. 
\begin{lemma}[Spectral Certification for Odd-Order Tensor with Kikuchi Matrix]\label{lem:kikuchi-certification-odd}
    Fix $t \in \NN$ odd. Let $F \in \RR^{[n]^t}$ be an asymmetric tensor, and $t-1\le \ell \le n$ be a parameter.

    Let $M = M^{t, \ell}(F)$ be the level-$\ell$ Kikuchi matrix of $F$. Then,
    \begin{align*}
        \max_{y^{(1)}, y^{(2)}, \dots, y^{(t)} \in \{\pm 1\}^n} \left|\langle F,  y^{(1)} \otimes y^{(2)} \otimes \dots \otimes y^{(t)}\rangle\right| \le \sqrt{\frac{n\binom{2(t-1)n}{\ell}}{\binom{2(t-1)n-2(t-1)}{\ell - (t-1)}\binom{t-1}{\frac{t-1}{2}}^2} \|M\|_2 + n \|F\|_{\text{Fr}}^2 }.
    \end{align*}
\end{lemma}

\begin{proof}
    Note that for any fixed $v^{(1)} = (v_1^{(1)}, \dots, v_{t-1}^{(1)}), v^{(2)} = (v_1^{(2)}, \dots, v_{t-1}^{(2)})\in \sO_{t-1}([n])$, there are exactly $\binom{2(t-1)n-2(t-1)}{\ell - (t-1)}\binom{t-1}{\frac{t-1}{2}}^2$ pairs of $\mathcal{A} = ((A_1^{(1)}, \dots, A_{t-1}^{(1)}),(A_1^{(2)}, \dots, A_{t-1}^{(2)})), \mathcal{B} = ((B_1^{(1)}, \dots, B_{t-1}^{(1)}), (B_1^{(2)}, \dots, B_{t-1}^{(2)})) \in \tilde{\sJ}$ such that for each $c \in \{1,2\}$, $\sum_{i=1}^{t-1} |A_i^{(c)}| = \sum_{i=1}^{t-1} |B_i^{(c)}|$ and $A_i^{(c)} \triangle B_i^{(c)} = \{v_i^{(c)}\}$ for all $1\le i \le t-1$. This is because 
    \begin{itemize}
        \item $\sum_{c \in \{1,2\} }\sum_i |A_i^{(c)}| = \ell, \sum_{c \in \{1,2\} }\sum_i |B_i^{(c)}| = \ell$,
        \item $\sum_{i=1}^{t-1} |A_i^{(c)}| = \sum_{i=1}^{t-1} |B_i^{(c)}|$ for each $c \in \{1,2\}$,
        \item and $A_i^{(c)} \triangle B_i^{(c)} = \{v_i^{(c)}\}$ for all $1\le i \le t-1$ and $c \in \{1,2\}$
    \end{itemize}
    force that for each $c \in \{1,2\}$, for exactly half of $i \in [t-1]$, $A_i^{(c)} = B_i^{(c)} \sqcup \{v_i^{(c)}\}$, and for the other half of $i \in [t]$, $B_i^{(c)} = A_i^{(c)} \sqcup \{v_i^{(c)}\}$. Thus, there are $\binom{2(t-1)n-2(t-1)}{\ell - (t-1)}$ ways to choose \[((A_1^{(1)} \cap B_1^{(1)}, \dots, A_{t-1}^{(1)} \cap B_{t-1}^{(1)}), (A_1^{(2)} \cap B_1^{(2)}, \dots, A_{t-1}^{(2)} \cap B_{t-1}^{(2)})),\]
    $\binom{t-1}{\frac{t-1}{2}}$ ways to distribute $v_1^{(1)}, \dots, v_{t-1}^{(1)}$ evenly between $(A_1^{(1)}, \dots, A_{t-1}^{(1)})$ and $(B_1^{(1)}, \dots, B_{t-1}^{(1)})$, and $\binom{t-1}{\frac{t-1}{2}}$ ways to distribute $v_1^{(2)}, \dots, v_{t-1}^{(2)}$ evenly between $(A_1^{(2)}, \dots, A_{t-1}^{(2)})$ and $(B_1^{(2)}, \dots, B_{t-1}^{(2)})$. We denote
    \begin{align*}
        d_{\ell} \colonequals \binom{2(t-1)n-2(t-1)}{\ell - (t-1)}\binom{t-1}{\frac{t-1}{2}}^2.
    \end{align*}

    Let $y^{(1)},  \dots, y^{(t)} \in \{\pm 1\}^n$. Following the strategy in \cite{chan2026strongly}, we use the Cauchy-Schwarz trick to get
    \begin{align*}
        &\langle F, y^{(1)} \otimes  \dots\otimes y^{(t)}\rangle^2\\
        &= \left(\sum_{v_1, \dots, v_t \in [n]} F(v_1, \dots, v_t) \prod_{i=1}^t y^{(i)}_{v_i}\right)^2\\
        &= \left(\sum_{v_1, \dots, v_{t-1}, u \in [n]} F(v_1, \dots, v_{t-1}, u) \left(\prod_{i=1}^{t-1} y^{(i)}_{v_i}\right) y^{(t)}_u\right)^2\\
        &\le \left(\sum_{u \in [n]} \left(\sum_{v_1, \dots, v_{t-1} \in [n]} F(v_1, \dots, v_{t-1}, u) \left(\prod_{i=1}^{t-1} y^{(i)}_{v_i}\right)\right)^2\right) \left(\sum_{u\in [n]} (y_u^{(t)})^2\right)\\
        &= n  \sum_{\substack{v^{(1)} = (v_1^{(1)}, \dots, v_{t-1}^{(1)}),\\ v^{(2)} = (v_1^{(2)}, \dots, v_{t-1}^{(2)}) \in \sO_{t-1}([n]) }} \left(\sum_{u \in [n]}F(v_1^{(1)}, \dots, v_{t-1}^{(1)}, u)F(v_1^{(2)}, \dots, v_{t-1}^{(2)}, u)\right) \left(\prod_{c \in \{1,2\}}\prod_{i=1}^{t-1} y^{(i)}_{v_i^{(c)}}\right) && \quad (y^{(t)} \in \{\pm 1\}^n)\\
        &= \frac{n}{d_{\ell}} \sum_{\substack{v^{(1)} = (v_1^{(1)}, \dots, v_{t-1}^{(1)}),\\ v^{(2)} = (v_1^{(2)}, \dots, v_{t-1}^{(2)}) \in \sO_{t-1}([n]):\\ v^{(1)}\ne v^{(2)} }} \sum_{\substack{\mathcal{A} = ((A_1^{(c)}, \dots, A_{t-1}^{(c)}))_{c=1}^2,\\ \mathcal{B} = ((B_1^{(c)}, \dots, B_{t-1}^{(c)}))_{c=1}^2 \in \tilde{\sJ}:\\ \forall\, c \in \{1,2\}, \sum_{i} |A_i^{(c)}| = \sum_{i} |B_i^{(c)}|,\\ A_i^{(c)} \triangle B_i^{(c)} = \{v_i^{(c)}\}, \forall\,1\le i \le t-1}} M(\mathcal{A}, \mathcal{B}) \left(\prod_{c \in \{1,2\}}\prod_{i=1}^{t-1} \prod_{j \in A_i^{(c)} \triangle B_i^{(c)}} y^{(i)}_{j}\right)\\
        &\quad + n\sum_{(v_1, \dots, v_{t-1}) \in \sO_{t-1}([n]) } \sum_{u \in [n]} F(v_1, \dots, v_{t-1}, u)^2\\
        &= \frac{n}{d_{\ell}} \sum_{\substack{\mathcal{A} = ((A_1^{(c)}, \dots, A_{t-1}^{(c)}))_{c=1}^2,\\ \mathcal{B} = ((B_1^{(c)}, \dots, B_{t-1}^{(c)}))_{c=1}^2 \in \tilde{\sJ}}} M(\mathcal{A}, \mathcal{B}) \prod_{c \in \{1,2\}}\prod_{i=1}^{t-1} \left(\prod_{j \in A_i^{(c)}} y^{(i)}_{j}\right)\left(\prod_{j \in B_i^{(c)}} y^{(i)}_{j}\right) + n\|F\|_{\text{Fr}}^2 && \quad (y^{(i)} \in \{\pm 1\}^n)\\
        &= \frac{n}{d_{\ell}} \sum_{\substack{\mathcal{A} = ((A_1^{(c)}, \dots, A_{t-1}^{(c)}))_{c=1}^2,\\ \mathcal{B} = ((B_1^{(c)}, \dots, B_{t-1}^{(c)}))_{c=1}^2 \in \tilde{\sJ}}} M(\mathcal{A}, \mathcal{B}) x_{\mathcal{A}} x_{\mathcal{B}} + n\|F\|_{\text{Fr}}^2,
    \end{align*}
    where we use $x \in \{\pm 1\}^{\tilde{\sJ}}$ to denote the vector with entries
    \begin{align*}
        x_{\mathcal{A}} = \prod_{c \in \{1,2\}} \prod_{i=1}^{t-1} \prod_{j \in A_i^{(c)}} y_j^{(i)}
    \end{align*}
    for $\mathcal{A} = ((A_1^{(1)}, \dots, A_{t-1}^{(1)}), (A_1^{(2)}, \dots, A_{t-1}^{(2)})) \in \tilde{\sJ}$. Since $x \in \{\pm 1\}^{\tilde{\sJ}}$ has a fixed norm $\|x\|_2 = \sqrt{|\tilde{\sJ}|} = \sqrt{\binom{2(t-1)n}{\ell}}$, we conclude
    \begin{align*}
        &\langle F, y^{(1)} \otimes y^{(2)}\otimes \dots\otimes y^{(t)}\rangle^2\\
        &\le \frac{n}{d_{\ell}} \sum_{\substack{\mathcal{A} = ((A_1^{(c)}, \dots, A_{t-1}^{(c)}))_{c=1}^2,\\ \mathcal{B} = ((B_1^{(c)}, \dots, B_{t-1}^{(c)}))_{c=1}^2 \in \tilde{\sJ}}} M(\mathcal{A}, \mathcal{B}) x_{\mathcal{A}} x_{\mathcal{B}} + n\|F\|_{\text{Fr}}^2\\
        &= \frac{n}{\binom{2(t-1)n-2(t-1)}{\ell - (t-1)}\binom{t-1}{\frac{t-1}{2}}^2} x^\top M x + n\|F\|_{\text{Fr}}^2\\
        &\le \frac{n}{\binom{2(t-1)n-2(t-1)}{\ell - (t-1)}\binom{t-1}{\frac{t-1}{2}}^2} \|M\|_2 \|x\|_2^2+ n\|F\|_{\text{Fr}}^2\\
        &= \frac{n\binom{2(t-1)n}{\ell}}{\binom{2(t-1)n-2(t-1)}{\ell - (t-1)}\binom{t-1}{\frac{t-1}{2}}^2} \|M\|_2 + n\|F\|_{\text{Fr}}^2.
    \end{align*}
\end{proof}

Since we may use the spectral norm of Kikuchi matrices to bound
\[\max_{y^{(1)}, y^{(2)}, \dots, y^{(t)} \in \{\pm 1\}^n} \left|\langle F,  y^{(1)} \otimes y^{(2)} \otimes \dots \otimes y^{(t)}\rangle\right|\]
in both the case of even-order tensor $F$ (Lemma~\ref{lem:kikuchi-certification-even}) and the case of odd-order tensor $F$ (Lemma~\ref{lem:kikuchi-certification-odd}, together with the Frobenius norm of $F$), and that the spectral norm of a real symmetric matrix can be computed in time polynomial in the dimension of the matrix, what remains is to provide high probability bound on the spectral norm of Kikuchi matrices in Definition~\ref{def:kikuchi-even} and Definition~\ref{def:kikuchi-odd}. Now, we may directly apply the following results on the spectral norm bound of Kikuchi matrices and the Frobenius norm bound of the centered occurrence tensor $F$.

\begin{theorem}[Norm Bound of Kikuchi Matrices {\cite[Lemma 6.2]{chan2026strongly}}]\label{thm:Kikuchi-norm}
    Let $\sI = \{(S_i, \Phi_{[t], \Id} \circ \nu_i^{-1})\}_{i=1}^m$ be a $(k,p)$-random $t$-OCSP instance. Let $F \in \RR^{[n]^t}$ be the asymmetric tensor defined in Definition~\ref{def:F} whose entries are centered counts of occurrences of clauses in $\sI$. Let $\frac{t}{2}\le \ell \in \NN$ and $M = M^{t, \ell}(F)$ be the level-$\ell$ Kikuchi matrix of $F$. Then, there exists a constant $C = C(k) > 0$ and $\alpha = \alpha(k) > 0$ such that, if $\ell \le \alpha n$ and
    \begin{align*}
        \overline{m} \ge C \frac{n^{\frac{t}{2} }}{\ell^{\frac{t}{2} -1 }}\log(n),
    \end{align*}
    where $\overline{m} = \binom{n}{k}p$, then with probability $1 - o(1)$,
    \begin{align*}
        \|M\|_2 \le \begin{cases}
            C\sqrt{\frac{\overline{m}}{n^{\frac{t}{2}} }}\ell^{\frac{t}{4}}\sqrt{\ell \log n} & \quad t \text{ even},\\
            C\frac{\overline{m}}{n^{\frac{t}{2}}} \ell^{\frac{t}{2}}\sqrt{\log n} & \quad t \text{ odd}.
        \end{cases}
    \end{align*}
\end{theorem}

\begin{remark}
    The Kikuchi matrix $M$ is built from the centered occurrence count tensor $F$ defined in Definition~\ref{def:F}, which follows the same distribution as the tensor $C_S$ \footnote{As mentioned before, the definition of $C_S$ in \cite[Section 5.3]{chan2026strongly} has a small typo, where $p$ needs to be replaced by $\hat{p} = \frac{p}{k!}$ in order to be consistent with Definition 3.2.~of \cite{chan2026strongly}. In fact, in the proofs, they implicitly treat $p = \frac{\overline{m}}{n^k}$ and diverge from the $p = \frac{\overline{m}}{\binom{n}{k}}$ in their Definition 3.2. We note that this is only a constant factor change in $p$ for constant $k$, and thus up to a constant, the same norm bound goes through. } considered in \cite[Section 5.3.]{chan2026strongly} where $S \subseteq [k]$ has size $|S| = t$. To see this, recall from Definition~\ref{def:F} that in our case, the tensor $F \in \RR^{[n]^t}$ is built as follows:
    \begin{itemize}
        \item First draw a Erd\H{o}s-R\'{e}nyi $k$-uniform hypergraph $H$ where each possible $R \in \binom{[n]}{k}$ is included in $E(H)$ independently with probability $p$.
    \item For every $R \in \binom{[n]}{k}$, draw a local ordering $\mu_R \in \Sym(R)$ independently and uniformly at random.
    \item Then, for every $(v_1, \dots, v_t) \in \sO_t([n])$, set
    \begin{align*}
        F(v_1, \dots, v_t) &= \sum_{R \in \binom{[n]}{k}} \sum_{\substack{\mu \in \Sym(R):\\ \mu(j) = v_j, \forall\, 1\le j \le t}} \left(\one\{R \in E(H), \mu = \mu_R\} - \frac{p}{k!}\right).
    \end{align*}
    \end{itemize}
    In particular, it is exactly the tensor $C_S$ defined in \cite[Section 5.2]{chan2026strongly} for the special case of $S = \{1, \dots, t\} \subseteq [k]$, as for $S = \{1, \dots, t\}$,
    \begin{align*}
        C_S(v_1, \dots, v_t) &= \sum_{\substack{(u_1, \dots, u_k) \in \sO_k([n]):\\ u_j = v_j, \forall\, 1 \le j \le t }} \left(\one\{R \colonequals\{u_1, \dots, u_k\} \in E(H),  (u_1, \dots, u_k) = \mu_R\} - \frac{p}{k!}\right)\\
        &= \sum_{R \in \binom{[n]}{k}} \sum_{\substack{\mu \in \Sym(R):\\ \mu(j) = v_j, \forall\, 1\le j \le t}} \left(\one\{R \in E(H), \mu = \mu_R\} - \frac{p}{k!}\right).
    \end{align*}

    In order to deal with general non-Boolean domain, \cite{chan2026strongly} additionally introduced tensor $C_{S, \beta}$ and Kikuchi matrix $M_{S, \beta}$. In our case, since we are dealing with evaluation of $F$ on $\{\pm 1\}$ rank-$1$ tensors, we do not need this machinery of considering $\beta$ an indicator of a specific configuration in $D^S$ for a generic non-Boolean domain $D$. Thus, we may treat $|D| = 1$ and disregard $\beta$ when using the result of \cite{chan2026strongly}. Specifically, our Kikuchi matrix $M = M^{t, \ell}(F)$ follows exactly the distribution of $(S, \ell)$-kikuchi matrix $M_{\beta}$ with $|S| = t$ considered in \cite{chan2026strongly}, where $|D| = 1$ and $\beta$ is the trivial indicator that always holds, which we can completely ignore.
\end{remark}

\begin{theorem}[Frobenius Norm Bound of $F$ {\cite[Claim B.1]{chan2026strongly}} ] \label{thm:frobenious-norm}
    Let $\sI = \{(S_i, \Phi_{[t], \Id} \circ \nu_i^{-1})\}_{i=1}^m$ be a $(k,p)$-random $t$-OCSP instance. Let $F \in \RR^{[n]^t}$ be the centered occurrence count tensor defined in Definition~\ref{def:F}. Then there exists a constant $C = C(k) > 0$ such that with probability $1 - o(1)$,
    \begin{align*}
        \|F\|_{\text{Fr}}^2 \le C \overline{m}\sqrt{\log n}. 
    \end{align*}
\end{theorem}

\begin{remark}
    Since our tensor $F \in \RR^{[n]^t}$ follows the same distribution as the tensor $C_S \in \RR^{[n]^t}$ defined in \cite[Section 5.3]{chan2026strongly} where $S \subseteq [k]$ has size $|S| = t$, and
    \begin{align*}
        \|F\|_{\text{Fr}}^2 &= \sum_{(v_1, \dots, v_t) \in [n]^t} F(v_1, \dots, v_t)^2\\
        &= \sum_{u \in [n]}\sum_{(v_1, \dots, v_{t-1}) \in [n]^{t-1}} F(v_1, \dots, v_{t-1}, u)^2,
    \end{align*}
    \cite[Claim B.1]{chan2026strongly} applies with $y_{(v_1, \dots, v_{t-1})} = 1$ for all $(v_1, \dots, v_{t-1}) \in [n]^{t-1}$.
\end{remark}

\subsubsection{Finishing Theorem~\ref{thm:kikuchi}}

Now we are ready to finish the proof of strong refutation via Kikuchi method.

\begin{proof}[Proof of Theorem~\ref{thm:kikuchi}]
    By Proposition~\ref{prop:reduction-to-canonical}, it suffices to give a two-sided $\eps'$-refutation algorithm for $(k,p)$-random $t$-OCSP instance $\sI$ with the canonical ordering predicate $\Phi_{[t], \Id}$ for $\eps' = \frac{\eps}{2^{2k}k!}$ that runs in time $n^{O(\ell)}$ when
    \begin{align*}
        \overline{m} \ge C\frac{n}{\eps'^2}\left(\frac{n}{\ell}\right)^{\frac{d}{2}-1}\log(n)^{2d-1},
    \end{align*}
    for some constant $C = C(k) > 0$, for every $2 \le t \le d$. By Lemma~\ref{lem:val-deviation}, for a $(k,p)$-random $t$-OCSP instance $\sI$ with the canonical ordering predicate $\Phi_{[t], \Id}$, we have
    \begin{align*}
        \left|m \cdot \val(\sI; \pi) - \frac{\overline{m}}{t!}\right| \le L^{t-1} \max_{y^{(1)}, \dots, y^{(t)} \in \{\pm 1\}^n } \left|\langle F, y^{(1)} \otimes \dots \otimes y^{(t)}\rangle\right|,
    \end{align*}
    where $L = \lceil \log_2(n+1) \rceil$.
    \begin{itemize}
        \item First we consider the case of even $t$. By Lemma~\ref{lem:kikuchi-certification-even}, for $M = M^{t, \ell}(F)$, we have
        \begin{align*}
            \left|m \cdot \val(\sI; \pi) - \frac{\overline{m}}{t!}\right| &\le L^{t-1} \max_{y^{(1)}, \dots, y^{(t)} \in \{\pm 1\}^n } \left|\langle F, y^{(1)} \otimes \dots \otimes y^{(t)}\rangle\right|\\
            &\le L^{t-1} \frac{\binom{tn}{\ell}}{\binom{tn-t}{\ell - \frac{t}{2}}\binom{t}{\frac{t}{2}}} \|M\|_2,
        \end{align*}
        for all $\pi \in \Sym([n])$.
        When $\overline{m} \ge C'\frac{n^{\frac{d}{2}}}{\ell^{\frac{d}{2}-1}}\log(n)\ge C'\frac{n^{\frac{t}{2}}}{\ell^{\frac{t}{2}-1}}\log(n)$ for some constant $C' = C'(k) > 0$, by Theorem~\ref{thm:Kikuchi-norm} we get that with probability $1 - o(1)$,
        \begin{align*}
            \left|m \cdot \val(\sI; \pi) - \frac{\overline{m}}{t!}\right| &\le L^{t-1} \frac{\binom{tn}{\ell}}{\binom{tn-t}{\ell - \frac{t}{2}}\binom{t}{\frac{t}{2}}} \|M\|_2\\
            &\le L^{t-1} \frac{\binom{tn}{\ell}}{\binom{tn-t}{\ell - \frac{t}{2}}\binom{t}{\frac{t}{2}}}  C'\sqrt{\frac{\overline{m}}{n^{\frac{t}{2}} }}\ell^{\frac{t}{4}}\sqrt{\ell \log n}.
        \end{align*}
        Consequently,
        \begin{align*}
            &\left|\val(\sI; \pi) - \frac{1}{t!}\right|\\
            &\le \frac{1}{m}\left|m \cdot\val(\sI; \pi) - \frac{\overline{m}}{t!}\right| + \frac{1}{t!}\left| \frac{\overline{m}}{m} - 1\right|\\
            &\le \frac{1}{m} L^{t-1} \frac{\binom{tn}{\ell}}{\binom{tn-t}{\ell - \frac{t}{2}}\binom{t}{\frac{t}{2}}}  C'\sqrt{\frac{\overline{m}}{n^{\frac{t}{2}} }}\ell^{\frac{t}{4}}\sqrt{\ell \log n} + \frac{1}{t!}\left| \frac{\overline{m}}{m} - 1\right|.
        \end{align*}
        By Lemma~\ref{lem:m-concentration},
        \begin{align*}
            \Pr\left(|m - \overline{m}| \ge \eps' \overline{m}\right) \le 2 \exp\left(-\frac{3\eps'^2 \overline{m}}{8}\right)= o(1),
        \end{align*}
        where we use that $\eps^2 \overline{m} \ge \Omega\left(n \left(n/\ell\right)^{\frac{d}{2}-1}\log(n)^{2d-1}\right)$ and $\eps' = \Omega(\eps)$. Therefore, with probability $1 - o(1)$, $|m - \overline{m}| < \eps' \overline{m}$ and
        \begin{align*}
            &\left|\val(\sI; \pi) - \frac{1}{t!}\right|\\
            &\le \frac{1}{m} L^{t-1} \frac{\binom{tn}{\ell}}{\binom{tn-t}{\ell - \frac{t}{2}}\binom{t}{\frac{t}{2}}}  C'\sqrt{\frac{\overline{m}}{n^{\frac{t}{2}} }}\ell^{\frac{t}{4}}\sqrt{\ell \log n} + \frac{1}{t!}\left| \frac{\overline{m}}{m} - 1\right|\\
            &\le \frac{1}{(1 - \eps')\overline{m}} L^{t-1} \frac{\binom{tn}{\ell}}{\binom{tn-t}{\ell - \frac{t}{2}}\binom{t}{\frac{t}{2}}}  C'\sqrt{\frac{\overline{m}}{n^{\frac{t}{2}} }}\ell^{\frac{t}{4}}\sqrt{\ell \log n} + \frac{\eps'}{t!}.
        \end{align*}
        For $\frac{t}{2} \le \ell \le \frac{1}{2} n$, we have
        \begin{align*}
            \frac{\binom{tn}{\ell}}{\binom{tn-t}{\ell - \frac{t}{2}}\binom{t}{\frac{t}{2}}} &= \frac{1}{\binom{t}{\frac{t}{2}}} \frac{(tn)!}{(tn-t)!} \frac{\left(\ell - \frac{t}{2}\right)!}{\ell!} \frac{\left(tn - \ell - \frac{t}{2}\right)!}{(tn- \ell)!} \\
            &\le \frac{1}{\binom{t}{\frac{t}{2}}} (tn)^t \left(\frac{1}{\ell - \frac{t}{2} + 1}\right)^{\frac{t}{2}} \left(\frac{1}{tn - \ell - \frac{t}{2} + 1}\right)^{\frac{t}{2}}\\
            &\le C'' \left(\frac{n}{\ell}\right)^{\frac{t}{2}},
        \end{align*}
        for some constant $C'' = C''(k) > 0$. Therefore,
        \begin{align*}
            &\left|\val(\sI; \pi) - \frac{1}{t!}\right|\\
            &\le \frac{1}{(1 - \eps')\overline{m}} L^{t-1} \frac{\binom{tn}{\ell}}{\binom{tn-t}{\ell - \frac{t}{2}}\binom{t}{\frac{t}{2}}}  C'\sqrt{\frac{\overline{m}}{n^{\frac{t}{2}} }}\ell^{\frac{t}{4}}\sqrt{\ell \log n} + \frac{\eps'}{t!}\\
            &\le 2C' \cdot L^{t-1} \cdot C'' \left(\frac{n}{\ell}\right)^{\frac{t}{2}} \sqrt{\frac{1}{\overline{m} n^{\frac{t}{2}}}} \ell^{\frac{t}{4}} \sqrt{\ell \log n} + \frac{\eps'}{t!}\\
            &\le 2C'C''\cdot \frac{ \left(\frac{n}{\ell}\right)^{\frac{t}{4}} \sqrt{\ell \log n} \cdot L^{t-1}}{\sqrt{\overline{m}}} + \frac{\eps'}{t!}.
        \end{align*}
        Recall $L = \lceil \log_2(n+1) \rceil$. When $\overline{m} \ge C\frac{n}{\eps^2}\left(\frac{n}{\ell}\right)^{\frac{d}{2}-1}\log(n)^{2d-1}$ for a large enough constant $C = C(k) > 0$, we get
        \begin{align*}
            &\left|\val(\sI; \pi) - \frac{1}{t!}\right|\\
            &\le 2C'C''\cdot \sqrt{\frac{n \left(\frac{n}{\ell}\right)^{\frac{t}{2}-1} L^{2t-2}\log(n)}{\overline{m}}} + \frac{\eps'}{t!}\\
            &\le \eps',
        \end{align*}
        as desired. Moreover, since constructing the tensor $F$ and Kikuchi matrix $M = M^{t, \ell}(F)$, computing the spectral norm $\|M\|_2$, and counting the number of clauses $m$ can be done in time $n^{O(\ell)}$, we conclude that above the stated clause density, we have a two-sided $\eps'$-refutation for a $(k,p)$-random $t$-OCSP instance with the canonical ordering predicate $\Phi_{[t], \Id}$ in time $n^{O(\ell)}$. 
        \item Next, we consider the case of odd $t$. By Lemma~\ref{lem:kikuchi-certification-odd}, for $M = M^{t, \ell}(F)$, we have
        \begin{align*}
            \left|m \cdot \val(\sI; \pi) - \frac{\overline{m}}{t!}\right| &\le L^{t-1} \max_{y^{(1)}, \dots, y^{(t)} \in \{\pm 1\}^n } \left|\langle F, y^{(1)} \otimes \dots \otimes y^{(t)}\rangle\right|\\
            &\le L^{t-1} \sqrt{\frac{n\binom{2(t-1)n}{\ell}}{\binom{2(t-1)n-2(t-1)}{\ell - (t-1)}\binom{t-1}{\frac{t-1}{2}}^2} \|M\|_2 + n \|F\|_{\text{Fr}}^2 },
        \end{align*}
        for all $\pi \in \Sym([n])$.
        When $\overline{m} \ge C'\frac{n^{\frac{d}{2}}}{\ell^{\frac{d}{2}-1}}\log(n)\ge C'\frac{n^{\frac{t}{2}}}{\ell^{\frac{t}{2}-1}}\log(n)$ for some constant $C' = C'(k) > 0$, by Theorem~\ref{thm:Kikuchi-norm} and Theorem~\ref{thm:frobenious-norm}, we get that with probability $1 - o(1)$,
        \begin{align*}
            \left|m \cdot \val(\sI; \pi) - \frac{\overline{m}}{t!}\right| &\le L^{t-1} \sqrt{\frac{n\binom{2(t-1)n}{\ell}}{\binom{2(t-1)n-2(t-1)}{\ell - (t-1)}\binom{t-1}{\frac{t-1}{2}}^2} \|M\|_2 + n \|F\|_{\text{Fr}}^2 }\\
            &\le L^{t-1} \sqrt{\frac{n\binom{2(t-1)n}{\ell}}{\binom{2(t-1)n-2(t-1)}{\ell - (t-1)}\binom{t-1}{\frac{t-1}{2}}^2} C'\frac{\overline{m}}{n^{\frac{t}{2}}} \ell^{\frac{t}{2}}\sqrt{\log n} + n \cdot C' \overline{m}\sqrt{\log n} }\\
            &= L^{t-1}\sqrt{\left(\frac{\binom{2(t-1)n}{\ell}}{\binom{2(t-1)n-2(t-1)}{\ell - (t-1)}\binom{t-1}{\frac{t-1}{2}}^2} \left(\frac{\ell}{n}\right)^{\frac{t}{2}} + 1\right) C'n \overline{m}\sqrt{\log n}}.
        \end{align*}
        Consequently,
        \begin{align*}
            &\left|\val(\sI; \pi) - \frac{1}{t!}\right|\\
            &\le \frac{1}{m}\left|m \cdot\val(\sI; \pi) - \frac{\overline{m}}{t!}\right| + \frac{1}{t!}\left| \frac{\overline{m}}{m} - 1\right|\\
            &\le \frac{1}{m}L^{t-1} \sqrt{\left(\frac{\binom{2(t-1)n}{\ell}}{\binom{2(t-1)n-2(t-1)}{\ell - (t-1)}\binom{t-1}{\frac{t-1}{2}}^2} \left(\frac{\ell}{n}\right)^{\frac{t}{2}} + 1\right) C'n \overline{m}\sqrt{\log n}} + \frac{1}{t!}\left| \frac{\overline{m}}{m} - 1\right|.
        \end{align*}
        As before, by Lemma~\ref{lem:m-concentration},
        \begin{align*}
            \Pr\left(|m - \overline{m}| \ge \eps' \overline{m}\right) = o(1).
        \end{align*}
        Therefore, with probability $1 - o(1)$, $|m - \overline{m}| < \eps' \overline{m}$ and
        \begin{align*}
            &\left|\val(\sI; \pi) - \frac{1}{t!}\right|\\
            &\le \frac{1}{m}L^{t-1} \sqrt{\left(\frac{\binom{2(t-1)n}{\ell}}{\binom{2(t-1)n-2(t-1)}{\ell - (t-1)}\binom{t-1}{\frac{t-1}{2}}^2} \left(\frac{\ell}{n}\right)^{\frac{t}{2}} + 1\right) C'n \overline{m}\sqrt{\log n}} + \frac{1}{t!}\left| \frac{\overline{m}}{m} - 1\right|\\
            &\le \frac{1}{(1 - \eps')\overline{m}}L^{t-1} \sqrt{\left(\frac{\binom{2(t-1)n}{\ell}}{\binom{2(t-1)n-2(t-1)}{\ell - (t-1)}\binom{t-1}{\frac{t-1}{2}}^2} \left(\frac{\ell}{n}\right)^{\frac{t}{2}} + 1\right) C'n \overline{m}\sqrt{\log n}} + \frac{\eps'}{t!}.
        \end{align*}
        For $t-1 \le \ell \le \frac{1}{2} n$, we have
        \begin{align*}
            &\frac{\binom{2(t-1)n}{\ell}}{\binom{2(t-1)n-2(t-1)}{\ell - (t-1)}\binom{t-1}{\frac{t-1}{2}}^2}\\
            &= \frac{1}{\binom{t-1}{\frac{t-1}{2}}^2} \frac{(2(t-1)n)!}{(2(t-1)n-2(t-1))!} \frac{\left(\ell - (t-1)\right)!}{\ell!} \frac{\left(2(t-1)n - \ell - (t-1)\right)!}{(2(t-1)n- \ell)!} \\
            &\le \frac{1}{\binom{t-1}{\frac{t-1}{2}}^2} (2(t-1)n)^{2(t-1)} \left(\frac{1}{\ell - (t-1)+1}\right)^{t-1} \left(\frac{1}{2(t-1)n - \ell - (t-1) + 1}\right)^{t-1}\\
            &\le C'' \left(\frac{n}{\ell}\right)^{t-1},
        \end{align*}
        for some constant $C'' = C''(k) > 0$. Therefore,
        \begin{align*}
            &\left|\val(\sI; \pi) - \frac{1}{t!}\right|\\
            &\le \frac{1}{(1 - \eps')\overline{m}}L^{t-1} \sqrt{\left(\frac{\binom{2(t-1)n}{\ell}}{\binom{2(t-1)n-2(t-1)}{\ell - (t-1)}\binom{t-1}{\frac{t-1}{2}}^2} \left(\frac{\ell}{n}\right)^{\frac{t}{2}} + 1\right) C'n \overline{m}\sqrt{\log n}} + \frac{\eps'}{t!}\\
            &\le 2 \sqrt{\frac{\left(C''\left(\frac{n}{\ell}\right)^{\frac{t}{2}-1} +1\right)C'n L^{2t-2}\sqrt{\log n}}{\overline{m}}} + \frac{\eps'}{t!}.
        \end{align*}
        Recall $L = \lceil \log_2(n+1) \rceil$. When $\overline{m} \ge C\frac{n}{\eps^2}\left(\frac{n}{\ell}\right)^{\frac{d}{2}-1}\log(n)^{2d-\frac{3}{2}}$ for a large enough constant $C = C(k) > 0$, we get
        \begin{align*}
            &\left|\val(\sI; \pi) - \frac{1}{t!}\right|\\
            &\le 2\sqrt{\frac{\left(C''\left(\frac{n}{\ell}\right)^{\frac{t}{2}-1} +1\right)C'n L^{2t-2}\sqrt{\log n}}{\overline{m}}} + \frac{\eps'}{t!}\\
            &\le \eps',
        \end{align*}
        as desired. Moreover, since constructing the tensor $F$ and Kikuchi matrix $M = M^{t, \ell}(F)$, computing the spectral norm $\|M\|_2$, and counting the number of clauses $m$ can be done in time $n^{O(\ell)}$, we conclude that above the stated clause density, we have a two-sided $\eps'$-refutation for a $(k,p)$-random $t$-OCSP instance with the canonical ordering predicate $\Phi_{[t], \Id}$ in time $n^{O(\ell)}$. 
    \end{itemize}

    Combining the two-sided $\eps'$-refutation for $t$-OCSP with canonical ordering predicate $\Phi_{[t], \Id}$ for both even and odd $t \le d$, by Proposition~\ref{prop:reduction-to-canonical}, we get an $\eps$-refutation for $p$-random $k$-OCSP with predicate $P$ of coordinate degree $d$ in time $n^{O(\ell)}$, provided that
    \begin{align*}
        \overline{m} \ge \begin{cases}
            C\frac{n}{\eps^2}\left(\frac{n}{\ell}\right)^{\frac{d}{2}-1}\log(n)^{2d-1} & \quad \text{ if } d \text{ is even},\\
            C\frac{n}{\eps^2}\left(\frac{n}{\ell}\right)^{\frac{d}{2}-1}\log(n)^{2d-\frac{3}{2}} & \quad \text{ if } d \text{ is odd},
        \end{cases}
    \end{align*}
    for some constant $C = C(k) > 0$.
\end{proof}

\subsection{Bucketing Method}

In this section, we prove Theorem~\ref{thm:bucketing}, which achieves strong refutation for OCSP by reducing to the situation of (non-ordering) CSPs specialized to the setting when the refutation strength $\eps > 0$ is a constant. The main idea is that instead of working with variable assignments in $\Sym([n])$, we work with a coarsened domain $D^n$ where $|D| = O_{k,\eps}(1)$. This idea of reduction to finite size alphabet also appeared in worst-case analysis of OCSP \cite{guruswami2012approximating, makarychev2012local}, but the details of our method differ from theirs. We will show how to construct two random CSPs from a random OCSP instance so that strongly refuting the two CSPs gives a strong refutation for the given OCSP.

Recall that to strongly refute $p$-random $k$-OCSP with a generic ordering predicate $P$ of coordinate degree $d$, by Proposition~\ref{prop:reduction-to-canonical} it is sufficient to provide two-sided strong refutation for $(k,p)$-random $t$-OCSP with canonical ordering predicate $\Phi_{[t],\Id}$ for $t \le d$.

We first state a simple observation that, a one-sided strong refutation for $(k,p)$-random $t$-OCSP with canonical ordering predicate $\Phi_{[t], \Id}$ can be turned into a two-sided strong refutation for the same OCSP instance.

\begin{lemma}\label{lem:one-to-two}
    Fix $\eps > 0$ a constant. Suppose there is an $\eps$-refutation algorithm for $(k,p)$-random $t$-OCSP instance with canonical ordering predicate $\Phi_{[t], \Id}$ on $n$ variables that runs in time $T$, then there is a two-sided $(t!)\eps$-refutation algorithm for the same OCSP instance that runs in time $O_{t}(T)$.
\end{lemma}

\begin{proof}
    Note that for any $\pi \in \Sym([n])$ and every $T \in \binom{[n]}{t}$, deterministically we have
    \begin{align}
        1 = \sum_{\mu \in \Sym(T)} \Phi_{[t], \Id} (\mu^{-1}(\pi[T])). \label{eq:order-identity}
    \end{align}
    Given a $(k,p)$-random $t$-OCSP instance $\sI = \{(T_i, \Phi_{[t], \Id} \circ \mu_i^{-1})\}_{i=1}^m$, we may construct $t!$ number of $t$-OCSP instances $\{\sI_{\sigma}\}_{\sigma \in \Sym([t])}$ as follows. For $\sigma \in \Sym([t])$, define
    \begin{align*}
        \sI_{\sigma} &= \{(T_i, \Phi_{[t], \sigma} \circ \mu_i^{-1})\}_{i=1}^m\\
        &= \{(T_i, \Phi_{[t], \Id} \circ \sigma^{-1} \circ \mu_i^{-1})\}_{i=1}^m.
    \end{align*}
    By the identity \eqref{eq:order-identity}, for any $\pi \in \Sym([n])$, we have
    \begin{align*}
        \sum_{\sigma \in \Sym([t])} \val(\sI_{\sigma}; \pi) = 1.
    \end{align*}
    Note that $\sI = \sI_{\Id}$, and any $\sI_{\sigma}$ follows the same distribution as $(k,p)$-random $t$-OCSP with canonical ordering predicate $\Phi_{[t], \Id}$. Therefore, applying the $\eps$-refutation algorithm to each $\sI_{\sigma}$ and taking a union bound over the failure probability, we get with high probability a certificate in time $O_t(T)$ that for all $\pi \in \Sym([n])$,
    \begin{align*}
        \val(\sI_{\sigma}; \pi) \le \frac{1}{t!} + \eps, \quad \forall\, \sigma \in \Sym([t]).
    \end{align*}
    Thus, we also certify a lower bound for $\val(\sI; \pi)$ for any $\pi \in \Sym([n])$ as
    \begin{align*}
        \val(\sI; \pi) &= 1 - \sum_{\substack{\sigma\in \Sym([t]):\\ \sigma \ne \Id}} \val(\sI_{\sigma}; \pi)\\
        &\ge 1 - \sum_{\substack{\sigma\in \Sym([t]):\\ \sigma \ne \Id}} (\frac{1}{t!} + \eps)\\
        &\ge \frac{1}{t!} - (t!)\eps.
    \end{align*}
    This concludes the proof.
\end{proof}

In what follows, we will focus on one-sided strong refutation for $(k,p)$-random $t$-OCSP with the canonical ordering predicate $\Phi_{[t],\Id}$ using the bucketing method.

\subsubsection{Domain Coarsening and Two Correlated CSPs}
Fix a constant $\eps > 0$ and $t \le d$. Let $B = B(k,\eps) \in \NN$ be a constant that we will set later. Let us consider a domain $[B]$, and two predicates $[B]^t \to \{0,1\}$ defined as follows.
\begin{definition}[Coarse Ordering Predicate and Collision Predicate]
    The coarse ordering predicate $Q_t: [B]^k \to \{0,1\}$ is defined as
    \begin{align*}
        Q_t(b_1, \dots, b_t, \dots, b_k) \colonequals \one\{b_1 < \dots < b_t\}.
    \end{align*}
    The collision predicate $R_t: [B]^k \to \{0,1\}$ is defined as
    \begin{align*}
        R_t(b_1, \dots, b_t, \dots, b_k) \colonequals \one\{b_1, \dots, b_t \text{ are not all distinct}\}.
    \end{align*}
\end{definition}

\paragraph{Bucketing Method Overview.}

We give a detailed description of how bucketing method works.

Recall that a $(k,p)$-random $t$-OCSP instance with the canonical ordering predicate $\Phi_{[t], \Id}$ is generated in the following way:
\begin{itemize}
    \item For every $S \in \binom{[n]}{k}$, sample $S$ independently with probability $p$.
    \item For every sampled $S \in \binom{[n]}{k}$, sample a uniform random local ordering $\nu_S \in \Sym(S)$.
    \item For every sampled pair $(S, \nu_S)$, let $\nu_S = (a_1, \dots, a_k)$, and add a clause $(T, \Phi_{[t], \Id} \circ \mu_T^{-1})$ where $T = \{a_1, \dots, a_t\}$ and $\mu_T = (a_1, \dots, a_t) \in \Sym(T)$.
\end{itemize}

Now let $\sI = \{(T_i, \Phi_{[t], \Id} \circ \mu_i^{-1})\}_{i=1}^m$ be a $(k,p)$-random $t$-OCSP instance with the canonical ordering predicate $\Phi_{[t], \Id}$. We construct two CSP instances $\sI_Q$ and $\sI_R$ on domain $[B]$ with predicates $Q_t$ and $R_t$ from $\sI$ as follows.

\begin{itemize}
    \item For every clause $(T_i, \Phi_{[t], \Id} \circ \mu_i^{-1})$ of $\sI$ where $1\le i \le m$, let $(S_i, \nu_i)$ be the sampled pair in the generation process that added the clause, i.e., $\nu_i = (a_1, \dots, a_k) \in \Sym(S_i)$, $\mu_i = (a_1, \dots, a_t)$, and $T_i = \{a_1, \dots, a_t\}$.
    \item Add a clause $(S_i, Q_t(x_{\nu_i(1)}, \dots, x_{\nu_i(k)}))$ to $\sI_Q$, and add a clause $(S_i, R_t(x_{\nu_i(1)}, \dots, x_{\nu_i(k)}))$ to $\sI_R$.
\end{itemize}

We will also define a map from $\Sym([n])$ to $[B]^n$. Partition $[n]$ into $B$ buckets as
\begin{align*}
    B_i \colonequals \left\{\left\lfloor (i-1)\frac{n}{B} \right\rfloor + 1, \left\lfloor (i-1)\frac{n}{B} \right\rfloor + 2, \dots, \left\lfloor i\frac{n}{B} \right\rfloor\right\}, \forall\, 1 \le i \le B.
\end{align*}
For every $a \in [n]$, denote
\begin{align*}
    \rho(a) = i \in [B],
\end{align*}
for the unique $i$ such that $a \in B_i$. For every $\pi \in \Sym([n])$, denote 
\begin{align*}
    \rho(\pi) = (\rho(\pi)_1, \dots, \rho(\pi)_n) \colonequals (\rho(\pi^{-1}(1)), \dots, \rho(\pi^{-1}(n))) \in [B]^n.
\end{align*}
In other words, $\rho(\pi) \in [B]^n$ is a coarsened statistics of the original ordering $\pi \in \Sym([n])$, where its $a$-th entry $\rho(\pi)_a$ encodes the bucket $B_i$ that the rank of $a$ under $\pi$, $\pi^{-1}(a) \in [n]$, falls into.

The following proposition shows that how we may upper bound the value of the OCSP instance $\sI$ using the value of CSP instances $\sI_Q$ and $\sI_R$ constructed above.

\begin{proposition}\label{prop:val-bucketing-bound}
    For any $t$-OCSP instance $\sI = \{(T_i, \Phi_{[t], \Id} \circ \mu_i^{-1})\}_{i=1}^m$, consider two $k$-CSP instances $\sI_Q$ and $\sI_R$ on domain $[B]$ defined above, with predicates $Q_t$ and $R_t$ respectively.
    Then, we have
    \begin{align*}
        \val(\sI; \pi) \le \val(\sI_Q; \rho(\pi)) + \val(\sI_R; \rho(\pi)).
    \end{align*}
\end{proposition}

\begin{proof}
    For any $\pi \in \Sym([n])$ and any clause $(T_i, \Phi_{[t], \Id} \circ \mu_i^{-1})$ of $\sI$, we have
    \begin{align*}
        &\Phi_{[t], \Id}(\mu_i^{-1}(\pi[T_i]))\\
        &= \one\{\mu_i^{-1}(\pi[T_i]) = \Id \}\\
        &= \one\{\pi[T_i] = \mu_i\}\\
        &= \one\{\pi^{-1}(a_1) < \dots < \pi^{-1}(a_t)\},
    \end{align*}
    where $\mu_i = (a_1, \dots, a_t) \in \Sym(T_i)$. Note that if the ranks of $a_1, \dots, a_t$ under $\pi$ fall into distinct buckets in $[B]$, then
    \begin{align*}
        &\one\{\pi^{-1}(a_1) < \dots < \pi^{-1}(a_t)\}\\
        &= \one\{\rho(\pi)_{a_1} < \dots < \rho(\pi)_{a_t}\}\\
        &= Q_t(\rho(\pi)_{a_1}, \dots, \rho(\pi)_{a_t}, \cdot)\\
        &= Q_t(\rho(\pi)_{a_1}, \dots, \rho(\pi)_{a_t}, \dots,\rho(\pi)_{a_k}),
    \end{align*}
    where we set $(a_1, \dots, a_t, \dots, a_k) = \nu_i$ to be the local ordering $\nu_i$ of $S_i$ which induced the clause given by $(T_i, \mu_i)$ in the $(k,p)$-random $t$-OCSP $\sI$. If the ranks of $a_1, \dots, a_t$ under $\pi$ do not fall into distinct buckets in $[B]$, then
    \begin{align*}
        \one\{\pi^{-1}(a_1) < \dots < \pi^{-1}(a_t)\} \le 1 = R_t(\rho(\pi)_{a_1}, \dots, \rho(\pi)_{a_k}).
    \end{align*}
    In other words, for every clause $(T_i, \Phi_{[t], \Id} \circ \mu_i^{-1})$ of $\sI$, we have
    \begin{align*}
        \Phi_{[t], \Id}(\mu_i^{-1}(\pi[T_i])) \le Q_t(\rho(\pi)_{a_1}, \dots, \rho(\pi)_{a_k}) + R_t(\rho(\pi)_{a_1}, \dots, \rho(\pi)_{a_k}). 
    \end{align*}
    Summing up all the clauses of $\sI$ and dividing by $m$, we get
    \begin{align*}
        \val(\sI; \pi) \le \val(\sI_Q; \rho(\pi)) + \val(\sI_R; \rho(\pi)).
    \end{align*}
\end{proof}

In order to obtain strong refutation for $(k,p)$-random $t$-OCSP instance $\sI$ with canonical ordering predicate $\Phi_{[t], \Id}$, our plan is to obtain strong refutation for $\sI_Q$, show that $\val(\sI_R; \rho(\pi))$ is small simultaneously for all $\pi \in \Sym([n])$, and then apply Lemma~\ref{lem:one-to-two} to obtain a two-sided strong refutation.

\subsubsection{Certifying Concentration for $\sI_Q$}

We first prove a (one-sided) strong refutation for $\sI_Q$. This is the part where we can directly apply existing result on CSP refutation.

\begin{theorem}[{\cite[Theorem 1.5]{chan2026strongly}}]\label{thm:csp-refutation}
    Let $k \in \NN$, $d \le k$, and $\eps > 0$ be constants. Let $D$ be a finite domain of constant size, $P: D^k \to \{0,1\}$ be a nontrivial predicate, and $\sI$ be a $p$-random $k$-CSP with predicate $P$ on $n$ variables. Let $\ell = \ell(n) \in \NN$. There exist $C = C(k, |D|, \eps) > 0$ and $\alpha = \alpha(k, |D|, \eps) > 0$ such that if $d-1 \le \ell \le \alpha n$, and
    \begin{align*}
        \overline{m} \ge C n \left(\frac{n}{\ell}\right)^{\frac{d}{2}-1}\log(n)
    \end{align*}
    where $\overline{m} = p\binom{n}{k}$, then there is an algorithm which runs in time $n^{O_{k, |D|, \eps}(\ell)}$ that certifies
    \begin{align*}
        \val(\sI) \le \opt_d(P) + \eps,
    \end{align*}
    where
    \begin{align*}
        \opt_d(P) \colonequals \max_{\substack{\tau \in \mathcal{P}(D^k):\\ d\text{-independent}}}\EE_{x \sim \tau }[P(x)],
    \end{align*}
    and we say a distribution $\tau$ supported on $D^k$ is $d$-independent if there exists some distribution $\xi$ on $D$ such that all the $d$-wise marginals of $\tau$ are $\xi^{\otimes d}$.
\end{theorem}

In our case, for a $(k,p)$-random $t$-OCSP instance $\sI$ with canonical ordering predicate $\Phi_{[t], \Id}$, it is easy to see that the constructed $k$-CSP instance $\sI_Q$ follows the distribution of $p$-random CSP with predicate $Q_t$. Thus, Theorem~\ref{thm:csp-refutation} applies and it remains to understand $\opt_t(Q_t)$.

\begin{lemma}\label{lem:opt_t}
    Let $B \in \NN$ be a constant. Consider the predicate $Q_t: [B]^k \to \{0,1\}$ defined by
    \begin{align*}
        Q_t(b_1, \dots, b_k) \colonequals \one\{b_1 < \dots < b_t\}.
    \end{align*}
    Then, we have $\opt_t(Q_t) = \frac{\binom{B}{t}}{B^t}$.
\end{lemma}

\begin{proof}
    By definition,
    \begin{align*}
        \opt_t(Q_t) \colonequals \max_{\substack{\tau \in \mathcal{P}([B]^k):\\ t\text{-independent}}}\EE_{x \sim \tau }[Q_t(x)].
    \end{align*}
    First, consider the uniform distribution $\Unif([B]^k)$. It's easy to see that the uniform distribution achieves the expectation $\frac{\binom{B}{t}}{B^t}$.
    
    Next, consider an arbitrary $t$-independent distribution $\tau$. Suppose under $\tau$, the first $t$ coordinates are marginally distributed as $\xi^{\otimes t}$. We may then compute
    \begin{align*}
        \EE_{x \sim \tau}[Q_t(x)] &= \EE_{(b_1, \dots, b_t) \sim \xi^{\otimes t}}[ \one\{b_1 < \dots < b_t\} ]\\
        &= \Pr_{\xi^{\otimes t}}\left(b_1 < \dots < b_t\right)\\
        &= \sum_{\{b_1, \dots, b_t\} \in \binom{[B]}{t}} \xi(b_1)\dots \xi(b_t)\\
        &= e_t(\xi(1), \dots, \xi(B)),
    \end{align*}
    where $e_t(x_1, \dots, x_N) \colonequals \sum_{S \in \binom{[N]}{t}} \prod_{i \in S} x_i$ denotes the $t$-th elementary symmetric polynomial. By Maclaurin's inequality, we have
    \begin{align*}
        \left(\frac{e_t(\xi(1), \dots, \xi(B))}{\binom{B}{t}}\right)^{\frac{1}{t}} &\le \frac{e_1(\xi(1), \dots, \xi(B))}{B}\\
        &= \frac{\sum_{i=1}^B \xi(i)}{B}\\
        &= \frac{1}{B}.
    \end{align*}
    Plugging this bound back, we get
    \begin{align*}
        \EE_{x \sim \tau}[Q_t(x)] &= e_t(\xi(1), \dots, \xi(B))\\
        &\le \frac{\binom{B}{t}}{B^t}
    \end{align*}
    as desired.
\end{proof}

Combining Theorem~\ref{thm:csp-refutation} and Lemma~\ref{lem:opt_t}, we get the following.
\begin{proposition}\label{prop:IQ-bound}
    Let $k \in \NN$, $t \le d \le k$, and $\eps > 0$ be constants. Let $\sI_Q$ be the $p$-random $k$-CSP with predicate $Q_t$ on $n$ variables constructed as in the bucketing method. Let $\ell = \ell(n) \in \NN$. There exist $C = C(k, B, \eps) > 0$ and $\alpha = \alpha(k, B, \eps) > 0$ such that if $d-1 \le \ell \le \alpha n$, and
    \begin{align*}
        \overline{m} \ge C n \left(\frac{n}{\ell}\right)^{\frac{d}{2}-1}\log(n)
    \end{align*}
    where $\overline{m} = p\binom{n}{k}$, then there is an algorithm which runs in time $n^{O_{k, B, \eps}(\ell)}$ that certifies
    \begin{align*}
        \val(\sI_{Q}) \le \frac{\binom{B}{t}}{B^t} + \eps.
    \end{align*}
\end{proposition}

\subsubsection{Certifying Concentration for $\sI_R$}

Now we move on to certify $\val(\sI_R; \rho(\pi))$ is small for all $\pi \in \Sym([n])$. Note that we cannot naively apply CSP machinery to certify this bound, since the assignment $(1, \dots, 1) \in [B]^n$ trivially satisfies all the collision predicates $R_t(b_1, \dots, b_k) = \one\{b_1, \dots, b_t \text{ are not all distinct}\}$. Nevertheless, we only care about evaluation of $\sI_R$ on assignments $\rho(\pi)$ induced by $\pi \in \Sym([n])$, and it turns out that we may exploit the regularity of $\rho(\pi) \in [B]^n$, that it contains roughly equal number of coordinates in each bucket among $[B]$, to algorithmically certify bounds on $\val(\sI_R; \rho(\pi))$ for all $\pi \in \Sym([n])$.

We first define the following collision matrix which encodes all the possible pairwise collisions in $\sI_R$.

\begin{definition}\label{def:collision-matrix}
    Let $\sI_R = \{(S_i, R_t(x_{\nu_i(1)}, \dots, x_{\nu_i(k)}))\}_{i=1}^m$ be a $k$-CSP with predicate $R_t$ on $n$ variables. The collision matrix $G = G(\sI_R) \in \RR^{n \times n}$ is defined as
    \begin{align*}
        G(u,v) &\colonequals \left|\left\{i \in [m]: \{u,v\} \subseteq \{\nu_i(1), \dots, \nu_i(t)\}  \right\}\right|, \quad \forall u \ne v \in [n]\\
        G(u,u) &\colonequals 0.
    \end{align*}
\end{definition}

The following proposition shows how to relate $\val(\sI_R; \rho(\pi))$ to this collision matrix $G$.

\begin{proposition}\label{prop:collision-matrix-bound}
    Let $\sI_R = \{(S_i, R_t(x_{\nu_i(1)}, \dots, x_{\nu_i(k)}))\}_{i=1}^m$ be a $k$-CSP with predicate $R_t$ on $n$ variables. Let $G = G(\sI_R)$ be the collision matrix in Definition~\ref{def:collision-matrix}. For any $\pi \in \Sym([n])$,
    \begin{align*}
        m \cdot \val(\sI_R; \rho(\pi)) \le \frac{1}{2}\sum_{a \in [B]} v_a^\top G v_a,
    \end{align*}
    where $v_a \in \{0,1\}^n$ is defined as the indicator vector of the set $\{i \in [n]: \rho(\pi)_i = a\}$.
\end{proposition}

\begin{proof}
    \begin{align*}
        m \cdot \val(\sI_R; \rho(\pi)) &= \sum_{i=1}^m R_t(\rho(\pi)_{\nu_i(1)}, \dots, \rho(\pi)_{\nu_i(k)})\\
        &= \sum_{i=1}^m \one\{\rho(\pi)_{\nu_i(1)}, \dots, \rho(\pi)_{\nu_i(t)} \text{ are not all distinct}\}\\
        &\le \sum_{i=1}^m \sum_{u<v\in [n]} \one\{\{u,v\} \subseteq \{\nu_i(1), \dots, \nu_i(t)\}\} \cdot \one\{\rho(\pi)_u = \rho(\pi)_v\}\\
        &= \sum_{u<v \in [n]} \one\{\rho(\pi)_u = \rho(\pi)_v\} \sum_{i=1}^m \one\{\{u,v\} \subseteq \{\nu_i(1), \dots, \nu_i(t)\}\}\\
        &= \sum_{u<v \in [n]} \one\{\rho(\pi)_u = \rho(\pi)_v\} G(u,v)\\
        &= \frac{1}{2} \sum_{a \in [B]} v_a^\top G v_a.
    \end{align*}
\end{proof}

Note that when $\sI$ is a $(k,p)$-random $t$-OCSP with canonical ordering predicate, $\sI_R$ is distributed as $p$-random $k$-CSP with predicate $R_t$. We may compute for $G = G(\sI_R)$ and $u \ne v\in [n]$ that
\begin{align*}
    \EE[G(u,v)] &= \EE\left[\left|\left\{i \in [m]: \{u,v\} \subseteq \{\nu_i(1), \dots, \nu_i(t)\}  \right\}\right|\right]\\
    &= \sum_{\substack{ S \in \binom{[n]}{k}:\\ \{u,v\} \subseteq S } } p \frac{t(t-1)(k-2)!}{k!}\\
    &= p\binom{n-2}{k-2}\frac{t(t-1)(k-2)!}{k!}\\
    &= \frac{p\binom{n}{k}\binom{t}{2}}{\binom{n}{2}},
\end{align*}
and thus $\EE[G] = \frac{p\binom{n}{k}\binom{t}{2}}{\binom{n}{2}}(J - I) = \frac{\overline{m} \binom{t}{2}}{\binom{n}{2}}(J-I)$. By Matrix Bernstein inequality, we get the following concentration result for $G$.

\begin{proposition}\label{prop:G-concentration}
    Let $\sI_R$ be a $p$-random $k$-CSP with predicate $R_t$, and $G = G(\sI_R)$ be its collision matrix. Then,
    \begin{align*}
        \EE[G] = \frac{\overline{m} \binom{t}{2}}{\binom{n}{2}}(J-I).
    \end{align*}
    
    If $\overline{m} = \Omega(n\log(n))$ where $\overline{m} \colonequals p\binom{n}{k}$, then there exists a constant $C = C(k) > 0$ such that
    \begin{align*}
        \Pr\left(\|G - \EE[G]\|_2 \ge C\sqrt{\frac{\overline{m}}{n}\log(n)}\right) = o(1).
    \end{align*}
\end{proposition}

\begin{proof}
    We have already computed $\EE[G]$ above, so it remains to prove the concentration bound for $G$.

    To generate $p$-random $k$-CSP $\sI_R$ with predicate $R_t$, we sample each $S \in \binom{[n]}{k}$ independently with probability $p$, and sample a uniform random $\nu_S \sim \Unif(\Sym(S))$. Let $I_S \sim \text{Bern}(p)$ denote the indicator variable that $S \in \binom{[n]}{k}$ is sampled by $\sI_R$.

    Then, we can decompose
    \begin{align*}
        G = \sum_{S \in \binom{[n]}{k}} X_S,
    \end{align*}
    where $X_S$ is a random matrix defined as
    \begin{align*}
        X_S(u,v) &= I_S\one\{\{u,v\} \subseteq \{\nu_S(1), \dots, \nu_S(t)\} \}, \quad \forall u\ne v\in [n]\\
        X_S(u,u) &= 0,
    \end{align*}
    and $X_S$ are independent.

    Thus, we have
    \begin{align*}
        G - \EE[G] = \sum_{S \in \binom{[n]}{k}} (X_S - \EE[X_S]).
    \end{align*}
    Note that when $X_S$ is not the zero matrix, $X_S$ is the adjacency matrix of the clique on vertices $\nu_S(1), \dots, \nu_S(t)$. Thus, deterministically,
    \begin{align*}
        \|X_S\|_2 \le t-1.
    \end{align*}
    Moreover, by Jensen's inequality,
    \begin{align*}
        \|\EE[X_S]\|_2 \le \EE\|X_S\|_2 \le t-1.
    \end{align*}
    Therefore, 
    \begin{align}
        \|X_S - \EE[X_S]\|_2 \le 2(t-1). \label{ineq:R-matrix-bernstein}
    \end{align}

    With probability $1-p$, $X_S$ is the zero matrix. With probability $p$, $X_S$ is the adjacency matrix of the clique on a uniformly random $t$-subset of $S$ by symmetry. Thus, with probability $p$, for a uniformly random $t$-subset $T \subseteq S$, $X_S^2 = (J_T-I_T)^2 = (t-2)J_T + I_T$. Averaging over $t$-subsets $T \subseteq S$, we get
    \begin{align*}
        \EE[X_S^2] &= p\cdot \EE_{T \sim \Unif(\binom{S}{t})}[(t-2)J_T + I_T]\\
        &=p\left(\frac{(t-2)\binom{k-2}{t-2}}{\binom{k}{t}} J_S + \frac{\binom{k-1}{t-1}}{\binom{k}{t}}I_S\right)\\
        &= p\left(\frac{t(t-1)(t-2)}{k(k-1)} J_S + \frac{t}{k}I_S\right).
    \end{align*}

    We may further compute
    \begin{align*}
        \left\|\sum_{S \in \binom{[n]}{k}} \EE[(X_S - \EE[X_S])^2]\right\|_2 &\le \left\|\sum_{S \in \binom{[n]}{k}} \EE[X_S^2]\right\|_2\\
        &= p \cdot \left\|\sum_{S \in \binom{[n]}{k}} \left(\frac{t(t-1)(t-2)}{k(k-1)} J_S + \frac{t}{k}I_S\right)\right\|_2\\
        &= p\binom{n}{k} \left\| \frac{t(t-1)(t-2)}{k(k-1)} \cdot \frac{\binom{n-2}{k-2}}{\binom{n}{k}} J+ \frac{t}{k}\cdot \frac{\binom{n-1}{k-1}}{\binom{n}{k}} I \right\|_2\\
        &= \overline{m} \left\| \frac{t(t-1)(t-2)}{n(n-1)} J+ \frac{t}{n} I \right\|_2\\
        &\le t^3\frac{\overline{m}}{n}
    \end{align*}

    By Matrix Bernstein inequality (Theorem~\ref{thm:matrix-bernstein}), we have
    \begin{align*}
        \Pr\left(\|G - \EE[G]\|_2 \ge h\right) &= \Pr\left(\left\|\sum_{S \in \binom{[n]}{k}}(X_S - \EE[X_S])\right\|_2 \ge h\right) \\
        &\le 2n \exp\left(-\frac{h^2}{t^3\frac{\overline{m}}{n} + 2(t-1)\frac{h}{3}}\right).
    \end{align*}
    Taking $h = C \sqrt{\frac{\overline{m}}{n}\log(n)}$ for a large enough constant $C = C(k)>0$, since $\overline{m} = \Omega(n\log(n))$, 
    \begin{align*}
        \Pr\left(\|G - \EE[G]\|_2 \ge C\sqrt{\frac{\overline{m}}{n}\log(n)}\right) &\le 2n \exp\left(- \frac{C^2 \frac{\overline{m}}{n}\log(n)}{t^3 \frac{\overline{m}}{n} + \frac{2}{3}(t-1)C \sqrt{\frac{\overline{m}}{n}\log(n)}}\right) = o(1).
    \end{align*}
\end{proof}

Finally, we establish the following certificate for $\val(\sI_R; \rho(\pi))$.
\begin{proposition}\label{prop:IR-bound}
    Let $k \in \NN$, $t \le d \le k$, and $\eps > 0$ be constants. Let $\sI_R$ be the $p$-random $k$-CSP with predicate $R_t$ on $n$ variables constructed as in the bucketing method. For a suitable choice of constant $B = B(k,\eps) > 0$, there exists $C = C(k, \eps) > 0$ such that if 
    \begin{align*}
        \overline{m} \ge C n \left(\frac{n}{\ell}\right)^{\frac{d}{2}-1}\log(n)
    \end{align*}
    where $\overline{m} = p\binom{n}{k}$, then there is an algorithm which runs in time polynomial in $n$ that certifies
    \begin{align*}
        \val(\sI_{R}; \rho(\pi)) \le \eps, \quad \forall\, \pi \in \Sym([n]).
    \end{align*}
\end{proposition}

\begin{proof}
    By Proposition~\ref{prop:collision-matrix-bound}, for any $\pi \in \Sym([n])$,
    \begin{align*}
        m \cdot \val(\sI_R; \rho(\pi)) &\le \frac{1}{2} \sum_{a \in [B]} v_a^\top G v_a\\
        &= \frac{1}{2} \sum_{a \in [B]} \left( v_a^\top \EE[G] v_a + v_a^\top (G - \EE[G]) v_a \right).
    \end{align*}
    By Proposition~\ref{prop:G-concentration}, $\EE[G] = \frac{\overline{m} \binom{t}{2}}{\binom{n}{2}}(J-I)$ and for some constant $C = C(k) > 0$,
    \begin{align*}
        \Pr\left(\|G - \EE[G]\|_2 \ge C\sqrt{\frac{\overline{m}}{n}\log(n)}\right) &= o(1).
    \end{align*}
    Therefore, by computing the spectral norm of $G - \EE[G]$, with high probability we can in polynomial time certify
    \begin{align*}
        m \cdot \val(\sI_R; \rho(\pi)) &\le \frac{1}{2} \sum_{a \in [B]} \left( v_a^\top \EE[G] v_a + v_a^\top (G - \EE[G]) v_a \right)\\
        &\le \frac{1}{2} \sum_{a \in [B]} \left(\frac{\overline{m} \binom{t}{2}}{\binom{n}{2}} v_a^\top (J-I)  v_a + \|G - \EE[G]\|_2 \|v_a\|_2^2 \right)\\
        &\le \frac{1}{2} \sum_{a \in [B]} \left(\frac{\overline{m} \binom{t}{2}}{\binom{n}{2}} (\|v_a\|_1^2- \|v_a\|_2^2) + C\sqrt{\frac{\overline{m}}{n}\log(n)} \|v_a\|_2^2 \right)\\
        &= \frac{1}{2}\frac{\overline{m} \binom{t}{2}}{\binom{n}{2}} \left(\sum_{a \in [B]} (\|v_a\|_1^2 - \|v_a\|_2^2)\right) + \frac{C}{2} \sqrt{\frac{\overline{m}}{n}\log(n)} \sum_{a \in [B]} \|v_a\|_2^2.
    \end{align*}
    Recall that for each $a \in [B]$, $v_a \in \{0,1\}^n$ is the indicator vector of the set $\{i \in [n]: \rho(\pi)_i = a\}$. Thus, $v_a$ has exactly $\lfloor a\frac{n}{B} \rfloor - \lfloor (a-1)\frac{n}{B} \rfloor = (1 + o(1))\frac{n}{B}$ number of $1$'s, and
    \begin{align*}
        \|v_a\|_1^2 &= (1 + o(1))\frac{n^2}{B^2},\\
        \|v_a\|_2^2 &= (1+o(1))\frac{n}{B}.
    \end{align*}
    Consequently, with high probability we have a polynomial time certificate that
    \begin{align*}
        m \cdot \val(\sI_R; \rho(\pi)) &\le \frac{1}{2}\frac{\overline{m} \binom{t}{2}}{\binom{n}{2}} \left(\sum_{a \in [B]} (\|v_a\|_1^2 - \|v_a\|_2^2)\right) + \frac{C}{2} \sqrt{\frac{\overline{m}}{n}\log(n)} \sum_{a \in [B]} \|v_a\|_2^2\\
        &\le \frac{1+o(1)}{2}\left(\frac{\overline{m} \binom{t}{2}}{\binom{n}{2}} \frac{n^2}{B} + C\sqrt{\overline{m}n\log(n)}\right).
    \end{align*}
    Finally, by Lemma~\ref{lem:m-concentration}, we have
    \begin{align*}
        \Pr\left(|m - \overline{m}| \ge \frac{1}{2} \overline{m}\right) \le 2 \exp\left(-\frac{3 \overline{m}}{32}\right) = o(1).
    \end{align*}
    Therefore, with high probability, $m \ge \frac{1}{2}\overline{m}$ and
    \begin{align*}
        \val(\sI_R; \rho(\pi)) &\le \frac{1}{m}\cdot \frac{1+o(1)}{2}\left(\frac{\overline{m} \binom{t}{2}}{\binom{n}{2}} \frac{n^2}{B} + C\sqrt{\overline{m}n\log(n)}\right)\\
        &\le \frac{2}{\overline{m}}\cdot \frac{1+o(1)}{2}\left(\frac{\overline{m} \binom{t}{2}}{\binom{n}{2}} \frac{n^2}{B} + C\sqrt{\overline{m}n\log(n)}\right)\\
        &\le (1+o(1))\left(\frac{2\binom{t}{2}}{B} + C\sqrt{\frac{n\log(n)}{\overline{m}}}\right),
    \end{align*}
    and this can be certified efficiently. Now choose $B = B(k,\eps) > 0$ large enough and $\overline{m} \ge \Omega_{k,\eps}(n\log(n))$, we get a certificate that $\val(\sI_R; \rho(\pi)) \le \eps$ for all $\pi \in \Sym([n])$ with high probability.
\end{proof}

\subsubsection{Finishing Theorem~\ref{thm:bucketing}}

Now we are ready to finish the proof of strong refutation via bucketing method, when the refutation strength $\eps > 0$ is a fixed constant.

\begin{proof}[Proof of Theorem~\ref{thm:bucketing}]
    By Proposition~\ref{prop:reduction-to-canonical}, it suffices to give a two-sided $\eps'$-refutation algorithm for $(k,p)$-random $t$-OCSP instance $\sI$ with the canonical ordering predicate $\Phi_{[t], \Id}$ for $\eps' = \frac{\eps}{2^{2k}k!}$ that runs in time $n^{O_{k,\eps}(\ell)}$ when
    \begin{align*}
        \overline{m} \ge C_1n\left(\frac{n}{\ell}\right)^{\frac{d}{2}-1}\log(n),
    \end{align*}
    for some constant $C_1 = C_1(k, \eps') > 0$, for every $2 \le t \le d$.

    By Lemma~\ref{lem:one-to-two}, it suffices to obtain an (one-sided) $\eps''$-refutation algorithm for $(k,p)$-random $t$-OCSP instance $\sI$ with canonical ordering predicate $\Phi_{[t], \Id}$ for $\eps'' = \frac{\eps'}{k!}$ that runs in time $n^{O_{k,\eps}(\ell)}$ when
    \begin{align*}
        \overline{m} \ge C_2n\left(\frac{n}{\ell}\right)^{\frac{d}{2}-1}\log(n),
    \end{align*}
    for some constant $C_2 = C_2(k, \eps'') > 0$, for every $2 \le t \le d$.

    By Proposition~\ref{prop:val-bucketing-bound}, for any $\pi \in \Sym([n])$, we have
    \begin{align}
        \val(\sI; \pi) \le \val(\sI_Q; \rho(\pi)) + \val(\sI_R; \rho(\pi)). \label{ineq:val-bucketing-bound}
    \end{align}

    By Proposition~\ref{prop:IQ-bound}, with high probability we can certify
    \begin{align*}
        \val(\sI_Q; \rho(\pi))\le \val(\sI_Q) \le \frac{\binom{B}{t}}{B^t} + \frac{\eps''}{2}
    \end{align*}
    in time $n^{O_{k,B,\eps}(\ell)}$, when
    \begin{align*}
        \overline{m} &\ge C_3n\left(\frac{n}{\ell}\right)^{\frac{d}{2}-1}\log(n),\\
        t-1 &\le \ell \le \alpha n,
    \end{align*}
    for some constant $C_3 = C_3(k, B, \eps'') > 0$ and $\alpha = \alpha(k, B, \eps'')$.

    By Proposition~\ref{prop:IR-bound}, with high probability we can certify that for all $\pi \in \Sym([n])$,
    \begin{align*}
        \val(\sI_R; \rho(\pi)) \le \frac{\eps''}{2}
    \end{align*}
    in polynomial time, when $B = B(k,\eps'') > 0$ is set to be a large enough constant and
    \begin{align*}
        \overline{m} \ge C_4n\left(\frac{n}{\ell}\right)^{\frac{d}{2}-1}\log(n)
    \end{align*}
    for some constant $C_4 = C_4(k, \eps'') > 0$.

    Since $\eps'' = \Omega(\eps)$, combining the certificates for $\val(\sI_Q; \rho(\pi))$ and for $\val(\sI_R; \rho(\pi))$ and choosing $B = B(k, \eps) > 0$, $\alpha = \alpha(k, \eps) > 0$ and $C = C(k, \eps) > 0$ appropriately, we can certify that for all $\pi \in \Sym([n])$,
    \begin{align*}
        \val(\sI; \pi) &\le \val(\sI_Q; \rho(\pi)) + \val(\sI_R; \rho(\pi))\\
        &\le \frac{\binom{B}{t}}{B^t} + \frac{\eps''}{2} + \frac{\eps''}{2}\\
        &\le \frac{1}{t!} + \eps''
    \end{align*}
    in time $n^{O_{k,\eps}(\ell)}$ when
    \begin{align*}
        \overline{m} &\ge Cn\left(\frac{n}{\ell}\right)^{\frac{d}{2}-1}\log(n),\\
        d-1 &\le \ell \le \alpha n.
    \end{align*}
    This finishes the proof.
\end{proof}

\section{Low Coordinate Degree Lower Bounds}

In this section, we prove that Theorem~\ref{thm:kikuchi} is nearly optimal under the Low Degree Conjecture. We first give an overview of the low coordinate algorithms. Then, we use the strategy of quiet planting to construct a hypothesis testing problem, and reduce the hypothesis testing problem to the strong refutation task. Finally, we show that below the clause density $\overline{m} \ll\frac{n}{\eps^2} \left(\frac{n}{D}\right)^{\frac{d}{2}-1}$, the hypothesis testing problem is hard for low coordinate degree-$D$ algorithms, and conclude that the strong refutation task is hard for algorithms with running time $\exp(O(D/\poly\log(n)))$ under the generalized Low Degree Conjecture (see Conjecture~\ref{conj:generalized-low-deg-conj}).

\subsection{Low Coordinate Degree Algorithm}\label{sec:low-deg}

\begin{definition}[Coordinate Degree]
    Let $N \in \NN$ and $\Omega^N$ be a product space. The coordinate degree $D(f)$ of a function $f: \Omega^N \to \R$ is the minimum $d \in \NN$ such that $f$ can be written as a linear combination of functions that each depend on at most $d$ coordinates.

    We will use $V_d$ to denote the space of functions with coordinate degree at most $d$.
\end{definition}

As remarked by \cite{kunisky2025low}, low coordinate degree functions were proposed in the thesis \cite{hopkins2018statistical} as a broader class functions compared to low-degree polynomials to be considered in the Low Degree Conjecture.

To state the Low Degree Conjecture (strengthened by replacing low degree polynomials with low coordinate degree algorithms), we define the following quantity called low coordinate degree advantage, proposed by \cite{kunisky2025low} as a generalization of low degree advantage considered in the context of low degree polynomials.

\begin{definition}[Advantage]
    Let $N = N(n) \in \NN$. Let $\QQ = \QQ_n$ and $\PP = \PP_n$ be sequences of distributions supported on $\Omega^N$. The coordinate degree-$D$ advantage of testing $\PP$ against $\QQ$ is
    \begin{align*}
        \Adv_{\le D}(\PP, \QQ) \colonequals \underset{\substack{f \in V_D: \\ \EE_{\QQ}[f^2] \le 1 }}{\max} \EE_{\PP}[f]. 
    \end{align*}
\end{definition}

We note that a bounded $\Adv_{\le D}(\PP, \QQ)$ rules out any coordinate degree-$D$ function for achieving strong separation between $\PP$ and $\QQ$. Now, the generalized Low Degree Conjecture can be stated using the low coordinate degree advantage as follows.

\begin{conjecture}[Generalized Low Degree Conjecture (Informal)]\label{conj:generalized-low-deg-conj}
    Let $N = N(n) \in \NN$, where $N$ is bounded by a polynomial in $n$. Let $\QQ = \QQ_n$ and $\PP = \PP_n$ be sequences of distributions supported on $\Omega^N$, where $\QQ$ is a product distribution. For ``sufficiently nice''\footnote{We omit the precise conditions of ``niceness'' here and treat the Low Degree Conjecture as a heuristic conjecture. We refer the interested reader to \cite[Chapter 2]{hopkins2018statistical} for how these conditions may be formalized. We note that one necessary condition is that $\PP$ should be permutation invariant under a natural action of the symmetric group $S_n$.} sequences of distributions $\QQ$ and $\PP$, the following holds:
    \begin{itemize}
        \item If for some $D \ge \log(n)^{1 + \eps}$ with a constant $\eps > 0$, $\Adv_{\le D}(\PP, \QQ) = O(1)$, then no polynomial time algorithm achieves strong detection between $\PP$ and $\QQ$.
        \item If $\Adv_{\le D}(\PP, \QQ) = O(1)$, then no algorithm that runs in time $\exp(O(D/\poly\log(n)))$ achieves strong detection between $\PP$ and $\QQ$.
    \end{itemize}
\end{conjecture}

Due to the generalized Low Degree Conjecture stated above, proving hardness against low coordinate degree algorithms is considered as evidence for average-case computational hardness for hypothesis testing problems. We will base our computational lower bounds on this conjecture.

\begin{remark}
    We note that caution needs to be taken on the Generalized Low Degree Conjecture and the precise scope of problems it may apply to. Recently, it was shown in \cite{buhai2025quasi} that the version of Generalized Low Degree Conjecture that only requires permutation invariance of $\PP$ is false. However, we believe that for many natural high-dimensional hypothesis testing problems, a scope not fully understood right now, the generalized Low Degree Conjecture holds and proving low coordinate degree lower bound can be regarded as convincing evidence for computational hardness.
\end{remark}

\subsection{Quiet Planting}

We will follow the quiet planting strategy, introduced by \cite{bandeira2021spectral} and later employed by \cite{kothari2023planted, kunisky2024computational}, to create a hypothesis testing problem between a null distribution $\QQ$ and a planted distribution $\PP$, where the planted signal in $\PP$ is ``computationally quiet''. We will show that any $\eps$-refutation algorithm for the random OCSP instance can be used to strongly distinguish $\PP$ and $\QQ$, and thus any computational hardness for the hypothesis testing problem between $\PP$ and $\QQ$ transfers to the strong refutation task.

Let $P: \Sym([k]) \to \{0,1\}$ be a non-trivial ordering predicate with coordinate degree $d$. Let $n \in \NN$, $m = m(n) \in \NN$, and $\eps = \eps(n) \in (0,\frac{1}{2^{2k+1}(2k)!}]$. Define $f: [0,1]^k \to \{0,1\}$ by
\begin{align*}
    f(x_1, \dots, x_k) \colonequals P(\ord(x)).
\end{align*}
Since the coordinate degree of $P$ is $D(P) = D(f) = d$, the Efron-Stein decomposition of $f$ can be written as 
\begin{align*}
    f(x) = \sum_{\substack{S \subseteq [k]: \\|S| \le d} } f_S(x_S),
\end{align*}
where the Efron-Stein components are
\begin{align*}
    f_S(x_S) = \sum_{T \subseteq S} (-1)^{|S| - |T|} \underset{x \sim \Unif([0,1]^k)}{\EE}[f(x) \vert x_T].
\end{align*}
We further define
\begin{align*}
    f^{=d}(x) \colonequals \sum_{\substack{S \subseteq [k]: \\|S| = d} } f_S(x_S),
\end{align*}
and denote $\|f^{=d}\|_2^2 \colonequals \underset{y \sim \Unif([0,1]^k) }{\EE}[f^{=d}(y)^2]$. We record the following facts about $f^{=d}$.

\begin{fact}\label{fact:fd}
    For $f(x_1, \dots, x_k) \colonequals P(\ord(x))$ and $f^{=d}(x) \colonequals \sum_{\substack{S \subseteq [k]: \\|S| = d} } f_S(x_S)$, we have
    \begin{itemize}
        \item $|f^{=d}(y)| \le 2^{2k}$ for all $y \in [0,1]^k$.
        \item $\|f^{=d}\|_2^2 \ge \frac{1}{(2k)!}$.
        \item $\sum_{\nu \in \Sym([k])} f^{=d}(y_{\nu}) = 0$ for all $y \in [0,1]^k$ with distinct coordinates, where $y_{\nu} = (y_{\nu_1}, \dots, y_{\nu_k})$ for $\nu = (\nu_1, \dots, \nu_k) \in \Sym([k])$.
    \end{itemize}
\end{fact}

\begin{proof}
    For any $y \in [0,1]^k$,
    \begin{align*}
    |f^{=d}(y)| &= \left|\sum_{\substack{S \subseteq [k]: \\|S| = d} } f_S(y_S)\right|\\
    &\le \left|\sum_{\substack{S \subseteq [k]: \\|S| = d} } \sum_{T \subseteq S} (-1)^{|S| - |T|} \underset{x \sim \Unif([0,1]^k)}{\EE}[f(y) \vert y_T]\right|\\
    &\le \sum_{\substack{S \subseteq [k]: \\|S| = d} } \sum_{T \subseteq S} \left|\underset{x \sim \Unif([0,1]^k)}{\EE}[f(y) \vert y_T]\right|\\
    &\le 2^{2k}.
\end{align*}
Next, we lower bound the norm of $f^{=d}$ as
\begin{align*}
    \|f^{=d}\|_2^2 &= \underset{y \sim \Unif([0,1]^k) }{\EE}[f^{=d}(y)^2]\\
    &= \underset{y \sim \Unif([0,1]^k) }{\EE}\left[\left(\sum_{\substack{S \subseteq [k]:\\ |S| = d} } f_S(y_S)\right)^2\right]\\
    &= \sum_{\substack{S \subseteq [k]:\\ |S| = d} } \underset{y \sim \Unif([0,1]^k) }{\EE}\left[f_S(y_S)^2\right] && \quad \text{(Fact~\ref{fact:ES-orthogonal})}\\
    &= \sum_{\substack{S \subseteq [k]:\\ |S| = d} } \underset{y \sim \Unif([0,1]^k) }{\EE}\left[\left(\sum_{T \subseteq S} (-1)^{|S| - |T|} \underset{x \sim \Unif([0,1]^k)}{\EE}[f(x) \vert x_T=y_T]\right)^2\right]\\
    &= \sum_{\substack{S \subseteq [k]:\\ |S| = d} } \sum_{T_1, T_2 \subseteq S} (-1)^{2d - |T_1| - |T_2|} \underset{\substack{y \sim \Unif([0,1]^k),\\ x \sim \Unif([0,1]^k),\\ z \sim \Unif([0,1]^k) }}{\EE}\left[f(x)f(z) \vert x_{T_1} = y_{T_1}, z_{T_2} = y_{T_2}\right]\\
    &= \sum_{\substack{S \subseteq [k]:\\ |S| = d} } \sum_{T_1, T_2 \subseteq S} (-1)^{2d - |T_1| - |T_2|} \underset{\substack{x \sim \Unif([0,1]^k),\\ z \sim \Unif([0,1]^{[k]}) }}{\EE}\left[f(x)f(z) \vert z_{T_1 \cap T_2} = x_{T_1 \cap T_2}\right]\\
    &\ge \frac{1}{(2k)!},
\end{align*}
where the last step follows because $\|f^{=d}\|_2^2 > 0$, the inner conditional expectation depends only on the ordering $\ord((\hat{x}, x_{T_1 \cap T_2} ,\hat{z}))$ with $\hat{x} = (x_i)_{i \notin T_1 \cap T_2}$ and $\hat{z} = (z_i)_{i \notin T_1 \cap T_2}$, all the orderings are equally likely in the conditional distribution, and thus all the conditional expectations are a multiple of $\frac{1}{(2k)!}$ as there are at most $(2k)!$ orderings of $(\hat{x}, x_{T_1 \cap T_2} ,\hat{z})$.

Finally, let $y \in [0,1]^k$ have distinct coordinates. We have
\begin{align*}
    &\sum_{\nu \in \Sym([k])} f^{=d}(y_{\nu})\\
    &= \sum_{\nu \in \Sym([k])}\sum_{\substack{S \subseteq [k]: \\|S| = d} } f_S((y_\nu)_S)\\
    &= \sum_{\nu \in \Sym([k])}\sum_{\substack{S \subseteq [k]: \\|S| = d} } \sum_{T \subseteq S} (-1)^{|S| - |T|} \underset{y \sim \Unif([0,1]^n)}{\EE}[f(y_\nu) \vert y_T]\\
    &= \sum_{\substack{S \subseteq [k]: \\|S| = d} } \sum_{T \subseteq S} (-1)^{|S| - |T|} \underset{y \sim \Unif([0,1]^n)}{\EE}\left[\sum_{\nu \in \Sym([k])} f(y_\nu) \bigg\vert y_T\right]\\
    &= h^{=d}(y),
\end{align*}
where we define $h: [0,1]^k \to \RR$ to be
\begin{align*}
    h(y) \colonequals \sum_{\nu \in \Sym([k])} f(y_{\nu}).
\end{align*}
Note that for $y \in [0,1]^k$ with distinct coordinates, 
\begin{align*}
    h(y) &= \sum_{\nu \in \Sym([k])} f(y_{\nu})\\
    &= \sum_{\mu \in \Sym([k])} P(\mu)\\
    &= k!\cdot \avg(P)
\end{align*}
is a constant. As a result,
\begin{align*}
    h^{=d}(y) &= \sum_{\substack{S \subseteq [k]: \\|S| = d} } \sum_{T \subseteq S} (-1)^{|S| - |T|} \underset{y \sim \Unif([0,1]^n)}{\EE}[h(y) \vert y_T]\\
    &= \sum_{\substack{S \subseteq [k]: \\|S| = d} } \sum_{T \subseteq S} (-1)^{|S| - |T|} k!\cdot \avg(P)\\
    &= 0,
\end{align*}
which proves that $\sum_{\nu \in \Sym([k])} f^{=d}(y_{\nu}) = h^{=d}(y) = 0$.

\end{proof}

\paragraph{Construction of the Quiet Planting.}
We take the null distribution $\QQ = \QQ_n$ to be the distribution of $p$-random $k$-OCSP instances with predicate $P$ on $n$ variables.

Now we construct the planted distribution $\PP = \PP_n(\eps)$, also supported on $k$-OCSP instances with predicate $P$ on $n$ variables with the same clause density $p$. First, we sample $x=(x_1, x_2, \dots, x_n) \sim \Unif([0,1]^n)$. Note that $x$ has distinct coordinates with probability $1$. Then, we generate the clauses in the following way:

\begin{itemize}
        \item For every $k$-subset $S \in \binom{[n]}{k}$, sample $S$ independently with probability $p$.
        \item For every sampled $S \in \binom{[n]}{k}$, sample a random ordering $\nu = (\nu_1, \dots, \nu_k) \in \Sym(S)$ according to the distribution
        \begin{align*}
            \Pr(\nu) = \frac{1}{k!}\left(1 + 2\eps \frac{f^{=d}(x_{\nu})}{\|f^{=d}\|_2^2} \right),
        \end{align*}
        where $x_{\nu} \colonequals (x_{\nu_1}, x_{\nu_2}, \dots, x_{\nu_k}) \in [0,1]^k$.
        \item Add a clause $(S, P \circ \nu^{-1})$.
    \end{itemize}
The only difference between $\PP$ and $\QQ$ is that the local ordering $\nu \in \Sym(S)$ is sampled to correlate with the degree-$d$ Efron-Stein components $f^{=d}$ evaluated at $x_{\nu}$ instead of the uniform distribution in the random OCSP distribution $\QQ$. We remark that
    \begin{align*}
        \Pr(\nu) = \frac{1}{k!}\left(1 + 2\eps \frac{f^{=d}(x_{\nu})}{\|f^{=d}\|_2^2} \right),
    \end{align*}
is a well-defined distribution over $\Sym([k])$ when $x \in [0,1]^n$ has distinct coordinates, since $\eps \le \frac{1}{2^{2k+1}(2k)!}$, and by Fact~\ref{fact:fd}, $|f^{=d}(y)| \le 2^{2k}$ for any $y \in [0,1]^k$, $\|f^{=d}\|_2^2 \ge \frac{1}{(2k)!}$, and for $x \in [0,1]^n$ with distinct coordinates,
\begin{align*}
    \sum_{\nu \in \Sym(S)} \Pr(\nu) &= \sum_{\nu \in \Sym(S)}\frac{1}{k!}\left(1 + 2\eps \frac{f^{=d}(x_{\nu})}{\|f^{=d}\|_2^2} \right)\\
    &= 1 + \frac{2\eps\|f^{=d}\|_2^2}{k!}\sum_{\nu \in \Sym(S)}f^{=d}(x_{\nu})\\
    &= 1.
\end{align*}

\begin{proposition}\label{prop:detection-to-refutation-reduction}
    Let $p = p(n) > 0$ and $\eps = \eps(n) \in (0, \frac{1}{2^{2k+1}(2k)!}]$. Suppose there is an $\eps$-refutation algorithm for $p$-random $k$-OCSP instance $\sI$ with predicate $P$ on $n$ variables that runs in time $T(n)$. If $\overline{m} \eps^2 = \omega(1)$ where $\overline{m} = p\binom{n}{k}$, then there is test that achieves strong detection between $\PP = \PP_n(\eps)$ and $\QQ = \QQ_n$ that runs in time $T(n)$.
\end{proposition}

\begin{proof}
    For $\sI \sim \QQ$, the $\eps$-refutation algorithm with high probability outputs a certificate that $\val(\sI) - \avg(P) \le \eps$ with high probability over the randomness of $\QQ$.

    Now, for $\sI \sim \PP$, we will show that with high probability, $\val(\sI) - \avg(P) > \eps$. Recall that to generate an instance $\sI = \{(S_i, P \circ \nu_i^{-1})\}_{i=1}^m \sim \PP$, we first draw $x = (x_1, \dots, x_n) \sim \Unif([0,1]^n)$, then sample every $S \in \binom{[n]}{k}$ independently with probability $p$, independently draw ordering $\nu_S \in \Sym(S)$ that correlates with $f^{=d}$, and finally add clauses $(S, P \circ \nu_S^{-1})$ for every sampled $S$. Let $I_S$ denote the indicator that $S \in \binom{[n]}{k}$ is sampled, and $\pi = \ord(x)$. We have
    \begin{align*}
        m\cdot (\val(\sI) - \avg(P)) &\ge m\cdot (\val(\sI; \pi) - \avg(P))\\
        &= \sum_{i=1}^m (P(\nu_i^{-1}(\pi[S_i])) - \avg(P))\\
        &= \sum_{i=1}^m (f(x_{\nu_i}) - \avg(P)) \\
        &= \sum_{S \in \binom{[n]}{k} } I_S (f(x_{\nu_S}) - \avg(P)),
    \end{align*}
    which is a sum of conditionally independent random variables $I_S (f(x_{\nu_S}) - \avg(P))$ given $x \in [0,1]^n$. We may compute
    \begin{align*}
        \underset{\nu_S}{\EE}[f(x_{\nu_S}) \vert x] &= \sum_{\nu_S \in \Sym(S)} \frac{1}{k!}\left(1 + 2\eps \frac{f^{=d}(x_{\nu_S})}{\|f^{=d}\|_2^2}\right) f(x_{\nu_S})\\
        &=  \sum_{\nu_S \in \Sym(S)} \frac{1}{k!} P(\ord(x_{\nu_S})) + \frac{2\eps}{k!} \sum_{\nu_S \in \Sym(S)} \frac{f^{=d}(x_{\nu_S}) f(x_{\nu_S})}{\|f^{=d}\|_2^2}\\
        &= \avg(P) + \frac{2\eps}{k!} \sum_{\nu_S \in \Sym(S)} \frac{f^{=d}(x_{\nu_S}) f(x_{\nu_S})}{\|f^{=d}\|_2^2}\\
        &= \avg(P) + 2\eps \Delta_S(x),
    \end{align*}
    where $\Delta_S(x) \colonequals \frac{1}{k!}\sum_{\nu_S \in \Sym(S)} \frac{f^{=d}(x_{\nu_S}) f(x_{\nu_S})}{\|f^{=d}\|_2^2}$. Since $I_S$ are independent $\text{Bern}(p)$ random variables, we have
    \begin{align*}
        \EE[I_S (f(x_{\nu_S}) - \avg(P)) \vert x] &=  2\eps p\Delta_S(x).
    \end{align*}
    Furthermore, since $I_S$ and $f(x_{\nu_S}) - \avg(P)$ are independent, and $|f(x_{\nu_S}) - \avg(P)| \le 1$, we may use the formula $\Var(XY) = \Var(X)\Var(Y) + \Var(X)\EE[Y]^2 + \Var(Y)\EE[X]^2$ to compute
    \begin{align*}
        &\Var\left(I_S (f(x_{\nu_S}) - \avg(P)) \vert x\right)\\
        &= \Var(I_S)\Var(f(x_{\nu_S}) - \avg(P)) + \Var(I_S)\EE[f(x_{\nu_S}) - \avg(P)]^2+ \EE[I_S]^2\Var(f(x_{\nu_S}) - \avg(P))\\
        &\le p(1-p) + p(1-p) + p^2\\
        &\le 2p.
    \end{align*}
    
    By Bernstein inequality (Theorem~\ref{thm:bernstein}) we get
    \begin{align*}
        &\Pr\left(\sum_{S \in \binom{[n]}{k} } I_S (f(x_{\nu_S})-\avg(P)) \le 2\eps  p\sum_{S \in \binom{[n]}{k}} \Delta_S(x) - \frac{\eps}{2} \overline{m}  \Bigg\vert x\right)\\
        &\le \exp\left(-\frac{\frac{1}{8}\eps^2 \overline{m}^2}{2p\binom{n}{k} + \frac{1}{6}\eps \overline{m}}\right)\\
        &= \exp\left(-\frac{\eps^2 \overline{m}}{8\left(2 + \frac{1}{6}\eps \right) }\right). \numberthis\label{ineq:val-conditional-bound} 
    \end{align*}

    Next, we consider the randomness of $x \sim \Unif([0,1]^n)$ and get
    \begin{align*}
        \underset{x \sim \Unif([0,1]^n)}{\EE}\left[\sum_{S \in \binom{[n]}{k}}\Delta_S(x)\right] &= \sum_{S \in \binom{[n]}{k}}\frac{1}{k!}\sum_{\nu_S \in \Sym(S)} \underset{x \sim \Unif([0,1]^n)}{\EE}\left[\frac{f^{=d}(x_{\nu_S}) f(x_{\nu_S})}{\|f^{=d}\|_2^2}\right]\\
        &= \binom{n}{k}\underset{y \sim \Unif([0,1]^k)}{\EE}\left[\frac{f^{=d}(y) f(y)}{\|f^{=d}\|_2^2}\right]\\
        &= \binom{n}{k}, \numberthis \label{eq:Delta-mean}
    \end{align*}
    since $\EE[f^{=d}(y) f(y)] = \EE[f^{=d}(y)^2] = \|f^{=d}\|_2^2$ by Fact~\ref{fact:ES-orthogonal}. Moreover, for any coordinate $i \in [n]$ and $x \in [0,1]^n$, we have
    \begin{align*}
        &\sup_{x_i' \in [0,1]}\left|\sum_{S \in \binom{[n]}{k}}\Delta_S(x) - \sum_{S \in \binom{[n]}{k}}\Delta_S(x_1, \dots, x_{i-1}, x_i', x_{i+1}, \dots, x_n)\right|\\
        &=  \sup_{x_i' \in [0,1]}\left| \sum_{S \in \binom{[n]}{k} }\frac{1}{k!}\sum_{\nu \in \Sym(S)} \frac{f^{=d}(x_{\nu})f(x_{\nu}) - f^{=d}(x'_{\nu})f(x'_{\nu})}{\|f^{=d}\|_2^2}\right|,
    \end{align*}
    where $x' = (x_1, \dots, x_{i-1}, x_i', x_{i+1}, \dots, x_n)$. Note that the inner sum could differ only if $S \in \binom{[n]}{k}$ contains the changed coordinate $i$, and thus at most $\binom{n-1}{k-1}/\binom{n}{k} = \frac{k}{n}$ fraction of $S$ contributes to the difference. Moreover, since $|f| \le 1$, $|f^{=d}| \le 2^{2k}$, and $\|f^{=d}\|_2^2 \ge \frac{1}{(2k)!}$ by Fact~\ref{fact:fd}, we see that
    \begin{align*}
        &\sup_{x_i' \in [0,1]}\left|\sum_{S \in \binom{[n]}{k}}\Delta_S(x) - \sum_{S \in \binom{[n]}{k}}\Delta_S(x_1, \dots, x_{i-1}, x_i', x_{i+1}, \dots, x_n)\right|\\
        &\le \frac{k}{n} \cdot \binom{n}{k}\cdot 2 \cdot 2^{2k} \cdot (2k)!\\
        &\le \frac{2^{2k+1} k(2k)!\binom{n}{k}}{n}.
    \end{align*}
    By McDiarmid's inequality (Theorem~\ref{thm:mcdiarmid}), we have
    \begin{align*}
        &\Pr_{x \sim \Unif([0,1]^n) }\left(\sum_{S \in \binom{[n]}{k}} \Delta_S(x) \le \EE\left[\sum_{S \in \binom{[n]}{k}}\Delta_S(x)\right] - \frac{1}{8}\binom{n}{k}\right)\\
        &\le \exp\left(- \frac{\frac{1}{32}\binom{n}{k}^2}{n \left(\frac{2^{2k+1} k(2k)!\binom{n}{k}}{n}\right)^2} \right) \\
        &= \exp\left(- \frac{1}{32} \cdot \frac{n}{(2^{2k+1} k(2k)!)^2}\right). \numberthis \label{ineq:Delta-bound}
    \end{align*}
    As a result, combining \eqref{ineq:val-conditional-bound}, \eqref{eq:Delta-mean}, and \eqref{ineq:Delta-bound}, we get 
    \begin{align*}
        &\Pr\left(m\cdot (\val(\sI)-\avg(P)) \le \frac{5}{4}\overline{m}\eps\right)\\
        &\le \Pr\left( \sum_{S \in \binom{[n]}{k}} I_S (f(x_{\nu_S})-\avg(P)) \le \frac{5}{4}\overline{m} \eps\right)\\
        &\le \Pr\left(\sum_{S \in \binom{[n]}{k}}\Delta_S(x) \le \frac{7}{8}\binom{n}{k}, \text{ or } \sum_{S \in \binom{[n]}{k} } I_S( f(x_{\nu_S})-\avg(P)) \le 2\eps p \sum_{S\in \binom{[n]}{k} } \Delta_S(x) - \frac{\eps}{2}\overline{m} \right)\\
        &\le \exp\left(- \frac{1}{32} \cdot \frac{n}{(2^{2k+1} k(2k)!)^2}\right) + \exp\left(-\frac{\eps^2 \overline{m}}{8\left(2 + \frac{1}{6}\eps \right) }\right)\\
        &= o(1),
    \end{align*}
    provided that $\eps^2 \overline{m} = \omega(1)$. Finally, by Lemma~\ref{lem:m-concentration},
    \begin{align*}
        \Pr\left(|m - \overline{m}| \ge \frac{1}{4} \overline{m}\right) \le 2 \exp\left(-\frac{3\overline{m}}{128}\right)= o(1),
    \end{align*}
    and we conclude that with high probability, $m \le \frac{5}{4}\overline{m}$, and
    \begin{align*}
        \val(\sI)-\avg(P) \ge \frac{\overline{m}}{m}\cdot\frac{5}{4}\eps \ge \eps.
    \end{align*}

    Thus, the value of $\sI \sim \PP$ is with high probability greater than $\avg(P) + \eps$. In other words, with high probability for $\sI \sim \PP$, the $\eps$-refutation algorithm for $\QQ$ cannot find a certificate that proves $\val(\sI) - \avg(P) \le \eps$. Consequently, running the $\eps$-refutation algorithm and checking if the algorithm returns an $\eps$-refutation certificate achieves strong detection between $\PP$ and $\QQ$.
\end{proof}

\subsection{Bounding the Low Coordinate Degree Advantage}

Let us state a useful characterization for the low coordinate degree advantage of testing $\PP$ against $\QQ$. We will follow the terminology in \cite{kunisky2025low}.

\begin{definition}[Latent Variable Model] \label{def:LVM}
    Let $\Sigma, \Omega$ be measurable spaces and $N \in \NN$. A latent variable model (LVM) consists of the following pair of distributions over $\Omega^N$:
    \begin{itemize}
        \item Sample $y \sim \QQ$ by sampling each coordinate $y_i \sim \QQ_i$ independently, where $\QQ = \QQ_1 \otimes \dots \otimes \QQ_N$ is a product distribution over $\Omega^N$.
        \item Sample $y \sim \PP$ by first sampling $x \sim \mathcal{X}$ for a distribution over $\Sigma$, and then sampling each coordinate $y_i \sim \PP_{i,x}$, where $\PP_x = \PP_{1,x} \otimes \dots \otimes \PP_{N,x}$ is a product distribution over $\Omega^N$.
    \end{itemize}
\end{definition}

Now, we observe the hypothesis testing problem in our quiet planting construction is an LVM. In the null distribution $\QQ$, each of $S \in \binom{[n]}{k}$ is sampled independently with probability $p$, and among the sampled $S \in \binom{[n]}{k}$, a uniform random clause with predicate $P$ supported on $S$ is added. In the planted distribution $\PP$, we first sample $x \sim \Unif([0,1]^n)$, then sample each $S \in \binom{[n]}{k}$ independently with probability $p$, and among the sampled $S \in \binom{[n]}{k}$ add a random clause with predicate $P$ supported on $S$ depending on $x$. Thus, we may use the following result to upper bound the low coordinate degree advantage. 

\begin{theorem}[Low Coordinate Degree Advantage for LVMs {\cite[Theorem 4.4]{kunisky2025low}}] \label{thm:adv-LVM}
    Consider the LVM defined in Definition~\ref{def:LVM}. Suppose that for all $x \in \Sigma$ and all $i \in [N]$, $\PP_{i,x}$ is absolutely continuous with respect to $\QQ_i$, and $d \PP_{i,x}/ d \QQ_i \in L^2(\QQ_i)$. The coordinate degree-$D$ advantage of testing $\PP$ against $\QQ$ satisfies
    \begin{align*}
        \Adv_{\le D}(\PP, \QQ)^2 &= \sum_{\substack{\mathcal{T} \subseteq [N]:\\ |\mathcal{T}| \le D} } \underset{x^{(1)},x^{(2)} \sim \mathcal{X} }{\EE}\left[\prod_{i\in \mathcal{T}} \underset{y \sim \QQ_i}{\EE}\left[\left(\frac{d\PP_{i,x^{(1)}}(y)}{d\QQ_i(y)}-1\right)\left(\frac{d\PP_{i,x^{(2)}}(y)}{d\QQ_i(y)}-1\right)\right]\right].
    \end{align*}
\end{theorem}

In our case, $N = \binom{n}{k}$ is the index set for all possible $S \in \binom{[n]}{k}$, $\PP_{S,x}$ for $S \in \binom{[n]}{k}$ is a discrete distribution over the possible clauses with predicate $P$ supported on $S$ and $\perp$ (we use $\perp$ to denote the nonexistence of a clause on $S$, i.e., $S$ is not sampled by $\PP_{S,x}$), and $\QQ_S$ places positive probability on all possible clauses with predicate $P$ on $S$ and $\perp$. Thus, the conditions that $\PP_{S,x}$ is absolutely continuous with respect to $\QQ_S$ and that $ \PP_{S,x}/  \QQ_S \in L^2(\QQ_S)$ hold. Next, we establish the following explicit bound for the summands in the coordinate advantage in Theorem~\ref{thm:adv-LVM}.

\begin{proposition}\label{prop:adv-summand-bound}
    Let $P: \Sym([k]) \to \{0,1\}$ be a non-trivial ordering predicate with coordinate degree $d$. Let $p = p(n) > 0$ and $\eps = \eps(n) \in (0, \frac{1}{2^{2k+1}(2k)!}]$. Let $\QQ = \QQ_n$ be the distribution of $p$-random OCSP with predicate $P$ on $n$ variables and $m$ clauses. Let $\PP = \PP_n(\eps)$ be the planted distribution in the quiet planting construction. Then, for any $\mathcal{T} \subseteq \binom{[n]}{k}$,
    \begin{align*}
        &\underset{x^{(1)},x^{(2)} \sim \Unif([0,1]^n) }{\EE}\left[\prod_{S\in \mathcal{T}} \underset{Y \sim \QQ_S}{\EE}\left[\left(\frac{\PP_{S,x^{(1)}}(Y)}{\QQ_S(Y)}-1\right)\left(\frac{\PP_{S,x^{(2)}}(Y)}{\QQ_S(Y)}-1\right)\right]\right]\\
        &\le \left(4\cdot 2^{4k}((2k)!)^2 p\eps^2\right)^{|\mathcal{T}|} \one\{\mathcal{T} \in \mathcal{G}_d\},
    \end{align*}
    where $\mathcal{G}_d$ denotes the $d$-overlap condition, defined as follows. We say a collection $\mathcal{T} \subseteq \binom{[n]}{k}$ satisfies the $d$-overlap condition $\mathcal{G}_d$ if for all $T \in \mathcal{T}$, 
    \begin{align*}
        \left|\left(\cup_{S\in \mathcal{T} \setminus \{T\} } S\right) \cap T\right| \ge d.
    \end{align*}
\end{proposition}

\begin{proof}
    As in the quiet planting construction, define $f: [0,1]^k \to \{0,1\}$ by
    \begin{align*}
        f(x_1, \dots, x_k) = P(\ord(x)).
    \end{align*}
    Let the Efron-Stein decomposition of $f$ be 
    \begin{align*}
        f(x) = \sum_{\substack{S \subseteq [k]: \\|S| \le d} } f_S(x_S),
    \end{align*}
    with Efron-Stein components $f_S$, and
    \begin{align*}
        f^{=d}(x) \colonequals \sum_{\substack{S \subseteq [k]: \\|S| = d} } f_S(x_S).
    \end{align*}

    For $S \in \binom{[n]}{k}$, consider $Y \in \{(S, P \circ \nu^{-1}): \nu \in \Sym(S)\} \cup \{\perp\}$. Under $\QQ_S$ ,
    \begin{align*}
        \QQ_S((S, P \circ \nu^{-1})) &= \frac{p}{k!},\\
        \QQ_S(\perp) &= 1 - p.
    \end{align*}
    Under $\PP_{S,x}$ as in the quiet planting construction is
    \begin{align*}
        \PP_{S,x}((S, P \circ \nu^{-1})) &= \frac{p}{k!} \left(1 + 2\eps \frac{f^{=d}(x_{\nu})}{\|f^{=d}\|_2^2}\right),\\
        \PP_{S,x}(\perp) &= 1 - p.
    \end{align*}
    Therefore, for $\mathcal{T} \subseteq \binom{[n]}{k}$, we get
    \begin{align*}
        &\underset{x^{(1)},x^{(2)} \sim \Unif([0,1]^n) }{\EE}\left[\prod_{S\in \mathcal{T}} \underset{Y \sim \QQ_S}{\EE}\left[\left(\frac{\PP_{S,x^{(1)}}(Y)}{\QQ_S(Y)}-1\right)\left(\frac{\PP_{S,x^{(2)}}(Y)}{\QQ_S(Y)}-1\right)\right]\right]\\
        &= \underset{x^{(1)},x^{(2)} \sim \Unif([0,1]^n) }{\EE}\left[\prod_{S\in \mathcal{T}} \underset{\nu \sim \Unif(\Sym(S))}{\EE}\left[4p\eps^2 \frac{f^{=d}(x^{(1)}_{\nu})f^{=d}(x^{(2)}_{\nu})}{\|f^{=d}\|_2^4}\right]\right]\\
        &= \left(\frac{4p\eps^2}{\|f^{=d}\|_2^4}\right)^{|\mathcal{T}|} \underset{\{\nu_S\}_{S\in \mathcal{T}} \sim \bigotimes_{S\in \mathcal{T}} \Unif(\Sym(S)) }{\EE} \left[\underset{x^{(1)}, x^{(2)}\sim \Unif([0,1]^n)}{\EE} \left[\prod_{S\in \mathcal{T}} f^{=d}(x^{(1)}_{\nu_S})f^{=d}(x^{(2)}_{\nu_S})\right]\right]\\
        &= \left(\frac{4p\eps^2}{\|f^{=d}\|_2^4}\right)^{|\mathcal{T}|} \underset{\{\nu_S\}_{S\in \mathcal{T}} \sim \bigotimes_{S\in \mathcal{T}} \Unif(\Sym(S)) }{\EE} \left[\left(\underset{x \sim \Unif([0,1]^n)}{\EE} \left[\prod_{S\in \mathcal{T}} f^{=d}x_{\nu_S}\right]\right)^2\right]. \numberthis \label{eq:summand-adv}
    \end{align*}
    Next we show that \eqref{eq:summand-adv} vanishes whenever $\mathcal{T} \subseteq \binom{[n]}{k}$ does not satisfy the $d$-overlap condition $\mathcal{G}_d$.
    Observe that for any $T \in \mathcal{T}$, we may write
    \begin{align*}
        &\underset{x \sim \Unif([0,1]^n)}{\EE} \left[\prod_{S\in \mathcal{T}} f^{=d}(x_{\nu_S})\right]\\
        &= \underset{z \sim \Unif([0,1]^{[n] \setminus T })}{\EE} \left[\underset{\substack{y \sim \Unif([0,1]^{T }),\\ x \in [0,1]^n:\\ x_{T} = y, x_{[n] \setminus T } = z }}{\EE} \left[f^{=d}(y_{\nu_T}) \prod_{S\in \mathcal{T} \setminus \{T\}} f^{=d}(x_{\nu_S})\right]\right].
    \end{align*}
    If the union of $S \in \mathcal{T} \setminus \{T\}$ share at most $d-1$ variables with $T$, i.e., $$\left|\left(\cup_{S\in \mathcal{T} \setminus \{T\} } S\right) \cap T\right| \le d-1,$$
    then conditioned on $z \sim \Unif([0,1]^{[n] \setminus T })$,
    \begin{align*}
        \prod_{S\in \mathcal{T} \setminus \{T\}} f^{=d}(x_{\nu_S})
    \end{align*}
    is a $(d-1)$-junta on $[0,1]^{T}$, i.e., it depends on at most $d-1$ coordinates of $T$, and it has coordinate degree at most $d-1$ by Fact~\ref{fact:junta-degree}. Using the orthogonality of the Efron-Stein decomposition stated in Fact~\ref{fact:ES-orthogonal}, we see that if
    $$\left|\left(\cup_{S\in \mathcal{T} \setminus \{T\} } S\right) \cap T\right| \le d-1,$$
    then $f^{=d}(y_{\nu_T})$ consists of degree $d$ Efron-Stein components, $\prod_{S\in \mathcal{T} \setminus \{T\}} f^{=d}(x_{\nu_S})$ has coordinate degree at most $d-1$, and thus
    \begin{align*}
        \underset{z \sim \Unif([0,1]^{[n] \setminus T })}{\EE} \left[\underset{\substack{y \sim \Unif([0,1]^{T }),\\ x \in [0,1]^n:\\ x_{T} = y, x_{[n] \setminus T } = z }}{\EE} \left[f^{=d}(y_{\nu_T}) \prod_{S\in \mathcal{T} \setminus \{T\}} f^{=d}(x_{\nu_S})\right]\right]= 0.
    \end{align*}
    Thus, whenever the $d$-overlap condition $\mathcal{G}_d$ is not met by $\mathcal{T}$, \eqref{eq:summand-adv} vanishes as promised. 
    
    Now suppose $\mathcal{T} \subseteq \binom{[n]}{k}$ does satisfy the $d$-overlap condition. Recall $|f^{=d}(y)| \le 2^{2k}$ for any $y \in [0,1]^k$ and $\|f^{=d}\|_2^2 \ge \frac{1}{(2k)!}$ by Fact~\ref{fact:fd}. We can therefore upper bound \eqref{eq:summand-adv} as
    \begin{align*}
        &\left(\frac{4p\eps^2}{\|f^{=d}\|_2^4}\right)^{|\mathcal{T}|} \underset{\{\nu_S\}_{S\in \mathcal{T}} \sim \bigotimes_{S\in \mathcal{T}} \Unif(\Sym(S)) }{\EE} \left[\left(\underset{x \sim \Unif([0,1]^n)}{\EE} \left[\prod_{S\in \mathcal{T}} f^{=d}x_{\nu_S}\right]\right)^2\right]\\
        &\le \left(4((2k)!)^2p\eps^2\right)^{|\mathcal{T}|} \underset{\{\nu_S\}_{S\in \mathcal{T}} \sim \bigotimes_{S\in \mathcal{T}} \Unif(\Sym(S)) }{\EE} \left[\left(2^{2k|\mathcal{T}|} \one\{\mathcal{T} \in \mathcal{G}_d\}\right)^2\right]\\
        &= \left(4\cdot 2^{4k}((2k)!)^2 p\eps^2\right)^{|\mathcal{T}|} \one\{\mathcal{T} \in \mathcal{G}_d\}. \numberthis \label{ineq:summand-adv-ub}
    \end{align*}
    
\end{proof}

Now we are ready to complete the proof of the low coordinate degree lower bound in Theorem~\ref{thm:low-deg-hardness}.

\subsection{Proof of Theorem~\ref{thm:low-deg-hardness}}

\begin{proof}
    By Theorem~\ref{thm:adv-LVM}, the coordinate degree-$D$ advantage of the LVM constructed by the quiet planting strategy can be bounded as
    \begin{align*}
        &\Adv_{\le D}(\PP, \QQ)^2\\
        &= \sum_{\substack{\mathcal{T} \subseteq \binom{[n]}{k}:\\ |\mathcal{T}| \le D} } \underset{x^{(1)},x^{(2)} \sim \Unif([0,1]^n) }{\EE}\left[\prod_{S\in \mathcal{T}} \underset{Y \sim \QQ_S}{\EE}\left[\left(\frac{\PP_{S,x^{(1)}}(Y)}{\QQ_S(Y)}-1\right)\left(\frac{\PP_{S,x^{(2)}}(Y)}{\QQ_S(Y)}-1\right)\right]\right]
        \intertext{Now we bound each summand by Proposition~\ref{prop:adv-summand-bound} and get}
        &\le \sum_{\substack{\mathcal{T} \subseteq \binom{[n]}{k}:\\ |\mathcal{T}| \le D} } \left(4\cdot 2^{4k}((2k)!)^2 p\eps^2\right)^{|\mathcal{T}|} \one\{\mathcal{T} \in \mathcal{G}_d\}\\
        &= 1 + \sum_{r=1}^D \sum_{\substack{\mathcal{T} \subseteq \binom{[n]}{k}:\\ |\mathcal{T}| = r} } \left(4\cdot 2^{4k}((2k)!)^2 p\eps^2\right)^{r} \one\{\mathcal{T} \in \mathcal{G}_d\}. \numberthis \label{ineq:adv-bound}
    \end{align*}
    Now we need to bound the number of $\mathcal{T} \subseteq \binom{[n]}{k}$ with size $|T| = r$ that satisfy the $d$-overlap condition $\mathcal{G}_d$. Recall that $\mathcal{T}$ satisfies $\mathcal{G}_d$ if for all $T \in \mathcal{T}$, 
    \begin{align*}
        \left|\left(\cup_{S\in \mathcal{T} \setminus \{T\} } S\right) \cap T\right| \ge d.
    \end{align*}
    Fix one such $\mathcal{T} \in \mathcal{G}_{d}$. For every $i \in [n]$, define $c_{\mathcal{T}}(i) \colonequals |\{S \in \mathcal{T}: i \in S\}|$. Then, the following holds
    \begin{align*}
        \left|\cup_{S \in \mathcal{T}} S_i\right| &= \sum_{i \in \cup_{S \in \mathcal{T}} S} \sum_{S \in \mathcal{T}: i\in S} \frac{1}{c_{\mathcal{T}}(i)}\\
        &=\sum_{S \in \mathcal{T}} \sum_{i \in S} \frac{1}{c_{\mathcal{T}}(i)}.
    \end{align*}
    The $d$-overlap condition $\mathcal{G}_d$ implies that for every $S \in \mathcal{T}$,
    \begin{align*}
        |\{i \in S: c_{\mathcal{T}}(i) \ge 2\}| \ge d.
    \end{align*}
    As a result,
    \begin{align*}
        \left|\cup_{S \in \mathcal{T}} S\right| &=\sum_{S \in \mathcal{T}} \sum_{i \in S} \frac{1}{c_{\mathcal{T}}(i)}\\
        &\le \sum_{S \in \mathcal{T}} \left((|S| - d) + \frac{d}{2}\right)\\
        &\le \left(k-\frac{d}{2}\right)|\mathcal{T}|.
    \end{align*}

    Consequently, we can upper bound the number of $\mathcal{T} \subseteq \binom{[n]}{k}$ with size $|\mathcal{T}| = r$ that satisfy $\mathcal{G}_d$ by first enumerating $U \subseteq [n]$ of size $\lfloor \left(k - \frac{d}{2}\right)r \rfloor$, and then enumerating the number of ways to choose $r$ different $k$-subsets within $U$:
    \begin{align*}
        \sum_{\substack{U \subseteq [n]:\\ |U| = \lfloor \left(k - \frac{d}{2}\right)r \rfloor} } \binom{\binom{|U|}{k}}{r} &= \binom{n}{\lfloor \left(k - \frac{d}{2}\right)r \rfloor} \binom{\binom{\lfloor \left(k - \frac{d}{2}\right)r \rfloor}{k}}{r}\\
        &\le \left(\frac{en}{ \left(k-\frac{d}{2}\right)r }\right)^{\left(k-\frac{d}{2}\right)r} \frac{\left(e \left(\frac{e \left(k - \frac{d}{2}\right)r }{k}\right)^k\right)^r}{r^r}\\
        &\le \left(\left(\frac{e}{k-\frac{d}{2}}\right)^{k - \frac{d}{2}} e \left(\frac{e\left(k - \frac{d}{2}\right)}{k}\right)^k\right)^r n^{\left(k - \frac{d}{2}\right)r}r^{\left(\frac{d}{2}-1\right)r},\numberthis \label{ineq:r-bound}
    \end{align*}
    where in the second-to-last line we used that $\left(\frac{en}{ \lfloor\left(k-\frac{d}{2}\right)r\rfloor }\right)^{\lfloor \left(k-\frac{d}{2}\right)r \rfloor} \le \left(\frac{en}{ \left(k-\frac{d}{2}\right)r }\right)^{\left(k-\frac{d}{2}\right)r}$, because by assumption $\left(k-\frac{d}{2}\right)r \le k D \le n$, and for $a > 0$, the function $\left(\frac{a}{x}\right)^x$ is increasing for $0 < x \le \frac{a}{e}$.

    Now we plug the bound \eqref{ineq:r-bound} back to \eqref{ineq:adv-bound} and get
    \begin{align*}
        &1 + \sum_{r=1}^D \sum_{\substack{\mathcal{T} \subseteq \binom{[n]}{k}:\\ |\mathcal{T}| = r} } \left(4\cdot 2^{4k}((2k)!)^2 p\eps^2\right)^{r} \one\{\mathcal{T} \in \mathcal{G}_d\}\\
        &\le 1 + \sum_{r=1}^D  \left(\left(\frac{e}{k-\frac{d}{2}}\right)^{k - \frac{d}{2}} e \left(\frac{e\left(k - \frac{d}{2}\right)}{k}\right)^k\right)^r n^{\left(k - \frac{d}{2}\right)r}r^{\left(\frac{d}{2}-1\right)r}\left(4\cdot 2^{4k}((2k)!)^2 p\eps^2\right)^{r} \\
        &\le 1 + \sum_{r=1}^D  \left(4\cdot 2^{4k}((2k)!)^2 e^{2k+1-\frac{d}{2}}\left(k - \frac{d}{2}\right)^{\frac{d}{2}} \cdot \frac{\overline{m} \eps^2 r^{\frac{d}{2}-1} }{n^{\frac{d}{2}}} \right)^r, \numberthis \label{ineq:quiet-planting-adv-bound}
    \end{align*}
    where we used $p = \frac{\overline{m}}{\binom{n}{k}} \le \frac{\overline{m} k^k}{n^k}$ in the last step.
    
    Since $4\cdot 2^{4k}((2k)!)^2 e^{2k+1-\frac{d}{2}}\left(k - \frac{d}{2}\right)^{\frac{d}{2}}$ is a constant depending on $k$, we conclude that when $\overline{m} = o\left(\frac{n}{\eps^2} \left(\frac{n}{D}\right)^{\frac{d}{2}-1}\right)$, the summation in \eqref{ineq:quiet-planting-adv-bound} is $o(1)$, and
    \begin{align*}
        \Adv_{\le D}(\PP, \QQ)^2 \le 1 + o(1).
    \end{align*}
    Assuming Conjecture~\ref{conj:generalized-low-deg-conj}, no algorithm that runs in time $\exp(O(D/\poly\log(n)))$ achieves strong detection between $\PP$ and $\QQ$. By contrapositive of Proposition~\ref{prop:detection-to-refutation-reduction}, no algorithm that runs in time $\exp(O(D/\poly\log(n)))$ achieves $\eps$-refutation of a $p$-random OCSP instance with predicate $P$ on $n$ variables, which finishes the proof.
\end{proof}

\section*{Acknowledgement}
The rank decomposition procedure (see Section~\ref{sec:rank-decom}) was discovered in an interactive session with ChatGPT 5.5-thinking.

\clearpage
}

\bibliographystyle{alpha}
\bibliography{main}

\end{document}